\documentclass[twocolumn]{IEEEtran}

\usepackage{cite}
\usepackage{gensymb}
\usepackage{amsmath,amssymb,amsfonts}
\usepackage{bm}
\usepackage{amsthm}
\usepackage{textcomp}

\usepackage{algorithm}
\usepackage{algpseudocode}

\usepackage{multirow}
\usepackage{booktabs}
\usepackage{siunitx}
\usepackage{graphicx}
\usepackage{textcomp}
\usepackage{xcolor}
\usepackage{comment}
\usepackage{svg}
\usepackage{subcaption}

\usepackage[margin=18.75mm,top=19mm,bottom=19mm]{geometry}

\usepackage[normalem]{ulem} 
\usepackage{cancel}
\usepackage{mathtools, cuted}

\usepackage{tikz}
\usetikzlibrary{shapes,arrows,positioning}
\usetikzlibrary{3d}

\newtheorem{Lemma}{Lemma}

\newtheorem{Remark}{Remark}

  {\proof}{\proofend}
\newtheorem{proposition}{Proposition}

\newcommand{\qa}{{\bf a}}

\newcommand{\qf}{{\bf f}}
\newcommand{\qg}{{\bf g}}
\newcommand{\qh}{{\bf h}}

\newcommand{\qn}{{\bf n}}

\newcommand{\qp}{{\bf p}}

\newcommand{\qu}{{\bf u}}
\newcommand{\qv}{{\bf v}}
\newcommand{\qw}{{\bf w}}
\newcommand{\qx}{{\bf x}}
\newcommand{\qy}{{\bf y}}

\newcommand{\qA}{{\bf A}}
\newcommand{\qB}{{\bf B}}

\newcommand{\qG}{{\bf G}}

\newcommand{\qI}{{\bf I}}

\newcommand{\qN}{{\bf N}}

\newcommand{\qV}{{\bf V}}
\newcommand{\qW}{{\bf W}}

\newcommand{\qY}{{\bf Y}}
\newcommand{\qZ}{{\bf Z}}

\DeclareMathOperator*{\argmin}{arg\,min}

\newcommand{\trace}{\mathrm{tr}}

\newcommand{\ps}{{\qp_{s}}} 
\newcommand{\pkr}{{\qp_{k}}} 

\newcommand{\gsk}{{\qg_{sk}}} 

\newcommand{\gskn}{{\qg_{sk,\rm{n}}}} 
\newcommand{\gskf}{{\qg_{sk,\rm{f}}}} 

\newcommand{\hgsk}{{\hat{\qg}_{sk}}} 

\newcommand{\hgskn}{{\hat{\qg}_{sk,\rm{n}}}} 

\newcommand{\hgskf}{{\hat{\qg}_{sk,\rm{f}}}} 

\newcommand{\hgsm}{{\hat{\qg}_{sm}}} 

\newcommand{\gamsk}{{{\gamma}_{sk}}}

\newcommand{\rsmm}{{{\varrho}_{s,mm}}} 
\newcommand{\rsmlnsq}{{{\varrho}^2_{s,ml}}} 

\newcommand{\rsmmp}{{{\varrho}_{s,mm'}}} 
\newcommand{\rsmmpsq}{{{\varrho}^2_{s,mm'}}} 

\newcommand{\rspmmp}{{{\varrho}_{s',mm'}}} 

\newcommand{\kmrt}{{\kappa_{sm}}}

\newcommand{\kmrtmp}{{\kappa_{sm'}}}
\newcommand{\kmrtsmp}{{\kappa_{s'm'}}}

\newcommand{\kzf}{{\kappa_{sl}}}

\newcommand{\bklq}{{\varsigma_{sk}}}

\newcommand{\Ex}{\mathbb{E}}

\newcommand{\tesk}{{\Tilde{\hat{\pmb{\varepsilon}}}_{sk}}}

\newcommand{\tekl}{{\Tilde{\hat{\pmb{\varepsilon}}}_{kl}}}

\newcommand{\nue}{{\nu_{\tilde{\hat{\pmb{\varepsilon}}}}}}

\newcommand{\bgskn}{{\Bar{\qg}_{sk,\rm{n}}}} 
\newcommand{\bgskf}{{\Bar{\qg}_{sk,\rm{f}}}} 

\newcommand{\bhskn}
{{\Bar{\qh}_{sk,\rm{n}}}} 

\newcommand{\bhskf}
{{\Bar{\qh}_{sk,\rm{f}}}} 

\newcommand{\tgkl}{{\Tilde{\qg}_{sk}}} 
\newcommand{\tgklf}{{\Tilde{\qg}_{sk,\rm{f}}}} 

\newcommand{\zekl}{{\xi_{sk}}} 
\newcommand{\bsk}{{\beta_{sk}}} 

\newcommand{\as}{{a_{s}}} 
\newcommand{\ats}{{\tilde{a}_{s}}} 

\newcommand{\asp}{{a_{s'}}} 

\newcommand{\nx}{{N_{s,x}}} 
\newcommand{\ny}{{N_{s,y}}} 
\newcommand{\ns}{{N}} 
\newcommand{\nrf}{{N_{RF}}} 

\newcommand{\sus}{{\sum\nolimits^{S}_{s=1}}}
\newcommand{\summ}{{\sum\nolimits^{M}_{m=1}}}
\newcommand{\suml}{{\sum\nolimits^{L}_{l=1}}}
\newcommand{\sumln}{{\sum\nolimits^{L_{\rm{n}}}_{l=1}}}
\newcommand{\sumlf}{{\sum\nolimits^{L_{\rm{f}}}_{l=1}}}
\newcommand{\sumk}{{\sum\nolimits^{K}_{k=1}}}

\newcommand{\susp}{{\sum\nolimits_{s'=1}^{S}}}

\newcommand{\In}{{\qI_N}} 

\newcommand{\ID}{{\mathsf{ID}} }

\newcommand{\PIDC}{{\Psi^{\mathsf{\,c}}_{\ID,l}(\boldsymbol{\Omega}^{\mathsf{ID}}\!\!,\qa)}} 

\newcommand{\PIDNC}{{\Psi^{\mathsf{\, nc}}_{\ID,l}(\boldsymbol{\Omega}^{\mathsf{ID}}\!\!,\boldsymbol{\Omega}^{\mathsf{EH}}\!\!,\qa)}} 

\newcommand{\OID}{{\boldsymbol{\Omega}^{\mathsf{ID}}}} 
\newcommand{\OEH}{{\boldsymbol{\Omega}^{\mathsf{EH}}}} 
\newcommand{\tOmgID}{\tilde{\boldsymbol{\Omega}}^{\mathsf{ID}}} 
\newcommand{\tOmgEH}{\tilde{\boldsymbol{\Omega}}^{\mathsf{EH}}} 

\newcommand{\tOmgEHt}{\tilde{\boldsymbol{\Omega}}^{\mathsf{EH,(t)}}} 
\newcommand{\tOmgIDt}{\tilde{\boldsymbol{\Omega}}^{\mathsf{ID,\mathsf{(t)}}}} 

\newcommand{\OmgIDst}{{\boldsymbol{\Omega}}^{{\star}\mathsf{ID}}} 
\newcommand{\OmgEHst}{{\boldsymbol{\Omega}}^{{\star}\mathsf{EH}}} 

\newcommand{\osm}{{\Omega_{sm}^{\mathsf{EH}}}}
\newcommand{\osl}{{\Omega_{sl}^{\mathsf{ID}}}}
\newcommand{\osln}{{\Omega_{sl,\rm{n}}^{\mathsf{ID}}}}
\newcommand{\oslf}{{\Omega_{sl,\rm{f}}^{\mathsf{ID}}}}

\newcommand{\ME}{{\mathbb{E}}} 
\newcommand{\MV}{{\mathrm{Var}}} 

\DeclareMathAccent{\doubleacute}{\mathalpha}{operators}{"7D}

\title{\fontsize{0.80cm}{1cm}\selectfont Power-Efficient XL-MIMO Design for Mixed Near- and Far-Field SWIPT Systems}

\author{Muhammad Zeeshan Mumtaz,~\IEEEmembership{Student Member,~IEEE,} Mohammadali Mohammadi,~\IEEEmembership{Senior Member,~IEEE,} \\
Hien Quoc Ngo,~\IEEEmembership{Fellow,~IEEE,} Hyundong Shin,~\IEEEmembership{Fellow,~IEEE,} and  Michail Matthaiou,~\IEEEmembership{Fellow,~IEEE}
\\

\thanks{\vspace{0.0em}

Parts of this paper were presented at the 2025 IEEE ICC~\cite{Zeeshan:ICC:2025}.

This work was supported by the UK Engineering and Physical Sciences
Research Council (EPSRC) grant EP/X04047X/2 for TITAN Telecoms Hub.
The work of H. Q. Ngo was supported by the a research grant
from the Department for the Economy Northern Ireland under the US-Ireland
R\&D Partnership Programme, and MULTIPLY-6G project that has received funding from the Smart Networks and Services Joint Undertaking (SNS JU) under the European Union’s Horizon Europe research and innovation programme under Grant Agreement No 101293106. This work was also supported by the European
Research Council (ERC) under the European Union’s Horizon 2020 research
and innovation programme (grant agreement No. 101001331). The work of  H. Shin was supported by the National Research Foundation of Korea (NRF) grant funded by the Korean government (MSIT) (RS-2025-00556064 and RS-2025-25442355) and by the MSIT (Ministry of Science and ICT), Korea, under the ITRC (Information Technology Research Center) support program (IITP-2025-RS-2021-II212046) supervised by the IITP (Institute for Information \& Communications Technology Planning \& Evaluation). (\textit{Corresponding authors: H. Shin and M. Matthaiou}).

M. Z. Mumtaz, M. Mohammadi, H. Q. Ngo, and M. Matthaiou are with the Centre for Wireless Innovation (CWI), Queen's University Belfast, BT3 9DT Belfast, U.K., (email:\{z.mumtaz, m.mohammadi, hien.ngo, m.matthaiou\}@qub.ac.uk). M. Matthaiou is also with the Department of Electronic Engineering, Kyung Hee University, Yongin-si, Gyeonggi-do 17104, Republic of Korea.}
\thanks{M. Z. Mumtaz is also with the College of Aeronautical Engineering, National University of Sciences \& Technology (NUST), Pakistan, (email: zmumtaz@cae.nust.edu.pk).}
\thanks{H.~Shin is with the Department of Electronics and Information Convergence Engineering, Kyung Hee University, Yongin-si, Gyeonggi-do 17104, Republic of Korea (e-mail: hshin@khu.ac.kr).}
}

\date{}

\begin{document}

\maketitle
\begin{abstract}
    This paper examines the power consumption (PC) efficiency of a mixed near- and far-field (MF) simultaneous wireless information and power transfer (SWIPT) system underpinned by a hybrid beamforming (HB)-based modular extra-large multiple-input-multiple output (XL-MIMO) array. Multiple information decoding (ID) and energy harvesting (EH) users are served by multiple constituent subarrays in both the near-field (NF) and far-field (FF) region of the transmit array.
    A novel decision method is proposed for accurate classification of different field users using Frobenius norm-based frequency correlation of the least square (LS) channel estimates. The NF spatial non-stationarities (SnS) effects entail distinct electromagnetic (EM) visibility regions (VRs), which can be customized to employ strategic activation of the constituent XL-MIMO subarrays. We formulate a two-tier joint optimization problem to minimize the overall PC, considering the power allocation (PA) for both ID and EH users in addition to the subarray activation (SA). This challenging mixed-integer problem is transformed into computationally tractable formulations, accompanied by the development of well-optimized algorithms.
    Our simulation results demonstrate an overall PC reduction for our proposed PA-SA-HB scheme by up to $93\%$ against the equal PA with full array (FA) and up to $18\%$ with respect to the PA-FA-HB case.
\end{abstract}
\vspace{-0.5em}
\vspace{-1em}
\section{Introduction}
\vspace{-0.1em}
Minimizing network power consumption (PC) has emerged as a critical performance objective in fifth-generation (5G) and beyond-5G technologies, and is actively pursued by mobile network operators~\cite{David,Mata,Jaewon}. This industrial orientation is further incentivized by widely adopted sustainability policies, in light of growing concerns related to operational costs and the carbon footprint of the global telecommunication ecosystem \cite{Kundu}. Meanwhile, due to its substantial capacity enhancement, extra-large MIMO (XL-MIMO) technology is anticipated to serve as a critical enabler for sixth-generation (6G) deployments. However, there is an urgent necessity to address the PC concerns of this sophisticated architectural design \cite{Wesemann, Han}.

In parallel, simultaneous wireless information and power transfer (SWIPT) has emerged as a key enabling technology for 6G networks, addressing the stringent energy and connectivity requirements of massive Internet-of-Things (IoT) deployments, wearable electronics, and battery-constrained sensing devices \cite{matthaiou,Haiyang_Shlezinger}. The massive aperture size and enhanced spatial resolution of XL-MIMO architectures provide unprecedented capabilities for high-gain energy beamforming and user-specific spatial multiplexing, thereby enabling efficient and concurrent wireless power delivery and data transmission \cite{Zhang:JSAC:2024}. These features are particularly suitable for practical 6G systems, where large-scale, maintenance-free device operation and sustainable network densification are indispensable.
  
In the realm of the near-field (NF) communication, XL-MIMO systems present a promising avenue for enhancing the efficiency of SWIPT~\cite{Haiquan}. This technological marriage becomes more relevant when we account for the spatial diversity and the unique propagation characteristics within the NF region, where the electromagnetic (EM) waves follow a spherical wavefront model as opposed to the planar waves in the far-field (FF) \cite{Haiquan2,Haiyang,cui3}.
This shift leads to the concept of visibility regions (VRs), which delineates the certain portions of the large antenna array visible to the different users, in the context of EM exposure. This selective EM visibility can be leveraged to devise a strategic resource allocation framework at the transmit array, as the spatial non-stationarities (SnS) of these communication channels implies that each user only interacts with a subset of XL-MIMO arrays~\cite{Zeeshan:WCL:2025, Carvalho,kangda,Xiangjun,Xiaomin_TWC2025}. Thus, this notion of VRs enables could underpin an efficient SWIPT solution by effective energy focusing and interference mitigation~\cite{Chen}.

\subsection{Review of Related Literature}

Over the past years, numerous research works have addressed important scientific problems, such as spherical wavefront-based NF communication, SnS characteristics, mixed near- and far-field (MF) regions and SWIPT applications.  
The concept of non-uniform spherical waves in the Fresnel region of XL-MIMO arrays has been introduced as a plausible basis for the NF channel model, contrasting with the conventional planar-wave assumption \cite{Haiquan,Haiquan2,Haiyang}.
This interpretation leads to the NF SnS effects giving rise to selective EM visibility.
In \cite{Carvalho}, Carvalho \textit{et al.}  presented the notion of VR, i.e., the signal power in the NF is received by a limited portion of the XL-MIMO array by leveraging the SnS effects. In \cite{kangda}, the authors analytically established this fact along with the development of a VR detection algorithm based on a low-complexity symbol detection scheme. Moreover, the knowledge of user VR information was used by the authors in \cite{Haohong} to design and schedule non-orthogonal pilot signals, using a two-stage graph partitioning (TSGP) approach, which exploits both VR and angular domain information.
Recently, the influence of molecular absorption and spherical wave reflection coefficients were considered in \cite{Huawei} along with the conventional NF SnS characteristics for designing effective beamforming schemes for XL-MIMO systems. 

The practical deployment of XL-MIMO systems will almost certainly involve network scenarios with service users located within the MF region. In \cite{Xiuhong}, the authors proposed an effective hybrid channel estimation using an orthogonal matching pursuit algorithm in a XL-MIMO system, which includes NF and FF path components. An inter-user interference-based perspective was presented in \cite{Yunpu} for MF communications, where a detailed analytical foundation was constructed for the interference at the NF user caused by the FF EM beams transmitted by the XL-array. Furthermore, a two-tier hybrid NF and FF beam training approach was proposed in \cite{Luo}. This scheme uses rough beam sweeping and fine hierarchical codebook with reduced training overhead. This network scenario was also explored on up-scaled terahertz ultra-massive MIMO systems investigated in \cite{Madhukumar}, where a dictionary learning-based Bayesian approach was used for compressed sensing-aided channel estimation. Recently, an NN-aided Hankelization-based NF/FF classifier was proposed in \cite{Kim_Access_2024} using singular-value features of partially observed channels.


The precise beamfocusing capability of NF communications can be further leveraged to harness the energy harvesting capability along with the information transfer. Zhang \textit{et al.} proposed an XL-array-based SWIPT system which delivers harvested energy to  NF users and information services to  FF users in \cite{Zhang:JSAC:2024}. A successive convex approximation-based optimization algorithm was conceived to also maximize the sum-power harvested to EH users by jointly devising the beam scheduling and power allocation, under maximum sum-rate constraints. Furthermore, a hybrid beamforming (HB) architecture-based SWIPT network was considered in \cite{Zhang:IoT:2024} using a penalty-based two-layer optimization algorithm. Finally, in \cite{Haiyang_Shlezinger_Nir}, a dynamic metasurface antennas-based energy transmission system was developed to provide wireless power to NF users by generating focused energy beams while accounting for hardware constraints.

From a practical 6G perspective, SWIPT deployments will naturally involve heterogeneous device classes, including conventional broadband terminals requiring reliable information decoding (ID), and battery-limited sensors and IoT nodes that benefit from dedicated wireless energy harvesting (EH)  \cite{Haiyang_Shlezinger}. When XL-MIMO arrays are deployed on building facades, indoor walls or ceiling-mounted access points, their large apertures will imply that some users will fall inside the radiative NF region, where spherical wavefronts and VR effects arise, while other users will remain in the FF region and experience approximately planar-wave propagation \cite{Haiquan}. This MF topology is therefore not an exception but a likely operating mode for dense networks. Accordingly, NF/FF regime-aware processing and VR-aware resource allocation become essential to simultaneously guarantee ID services, provide targeted energy focusing for EH devices, and reduce PC by avoiding unnecessary activation of subarrays and RF chains that contribute marginally to the SWIPT demands.

 \vspace{-1em}
\subsection{Research Gap and Key Contributions}
Although the above research works consider SWIPT-based XL-MIMO systems in the context of MF scenarios and SnS effects separately, these topological characteristics have not been effectively utilized to address the following three important aspects. Firstly, the procedure for accurate identification of the operational EM NF/FF regions remains an open problem. Although separate channel estimation processes for NF/FF users have been devised \cite{cui3,Xiuhong}, we need to categorize users into distinct field regions using channel estimates. Secondly, the upper bounds for the DL harvested energy (HE) and spectral efficiency (SE) for the MF spatially distributed users have not been analyzed. 
Lastly, the literature has not exploited the phenomenon of dormant portions of the large antenna aperture of XL-MIMO arrays to reduce the overall PC by strategic activation of modular parts of the transmitter array. In our related conference paper \cite{Zeeshan:ICC:2025}, we have addressed the subarray activation (SA) based on the SnS effects, but the impact of MF regions is yet to be evaluated. In this context, Table \ref{tabel:Survey} compares our contributions against the contemporary literature.   

The main contributions of this paper are as follows:
\begin{table*}
	\centering
	\caption{\label{tabel:Survey}      Contrasting our contributions to the relevant literature}
	\vspace{-0.6em}
	\small
        \begin{tabular}{|p{3.6cm}|p{1.5cm}|p{0.7cm}|p{0.7cm}|p{0.7cm}|p{0.7cm}|p{0.7cm}|p{0.7cm}|p{0.7cm}|p{0.7cm}|p{0.7cm}|p{0.7cm}|}
	\hline
        \centering\textbf{Contributions} 
        &\centering \textbf{This paper}
        &\centering\cite{Zeeshan:ICC:2025}
        &\centering\cite{Zhang:JSAC:2024}
        &\centering\cite{Haiyang}
        &\centering\cite{cui3}
        &\centering\cite{Zeeshan:WCL:2025}
        &\centering\cite{kangda}
        &\centering\cite{Chen}
        &\centering\cite{Yunpu}
        &\centering\cite{Zhang:IoT:2024}  
        &\centering\cite{Jun_Zhang}
        \cr

        \hline

        Spherical wave NF channel
        & \centering \checkmark
        & \centering \checkmark
        & \centering \checkmark  
        & \centering \checkmark  
        & \centering \checkmark  
        & \centering \checkmark  
        & \centering \checkmark
        & \centering -
        & \centering \checkmark
        & \centering \checkmark      
        & \centering -  
         \cr

         \hline

        MF scenarios 
        & \centering \checkmark
        & \centering -
        & \centering \checkmark  
        & \centering -  
        & \centering \checkmark
        & \centering -  
        & \centering -
        & \centering -
        & \centering \checkmark
        & \centering -      
        & \centering -  
         \cr

         \hline

         Spatial non-stationarities  
        &\centering \checkmark 
        & \centering \checkmark
        & \centering -
        & \centering -
        & \centering -
        & \centering \checkmark
        & \centering \checkmark
        & \centering -
        & \centering -
        & \centering -    
        & \centering \checkmark
        \cr

         \hline

        Channel estimation
        & \centering \checkmark 
        & \centering - 
        & \centering \checkmark  
        & \centering - 
        & \centering \checkmark 
        & \centering - 
        & \centering -  
        & \centering \checkmark  
        & \centering -  
        & \centering \checkmark     
        & \centering \checkmark 
         \cr

        \hline

        SWIPT operation  
        & \centering \checkmark 
        & \centering \checkmark 
        & \centering \checkmark  
        & \centering -  
        & \centering -  
        & \centering -  
        & \centering -  
        & \centering \checkmark
        & \centering -  
        & \centering \checkmark     
        & \centering - 
         \cr

        \hline

       Hybrid beamforming         
        &\centering\checkmark 
        & \centering -
        &\centering \checkmark
        & \centering\checkmark
        & \centering\checkmark
        &\centering  -
        &\centering  -
        &\centering\checkmark 
        & \centering \checkmark
        &\centering  \checkmark
        & \centering -
        \cr

        \hline

       Subarray activation         
        &\centering\checkmark 
        & \centering \checkmark
        &\centering - 
        & \centering -
        & \centering -
        & \centering \checkmark
        &\centering \checkmark
        &\centering - 
        & \centering - 
        &\centering -
        & \centering -  
        \cr
     \hline
  
\end{tabular}
 \vspace{-1.3em}
\label{Contribution}
\end{table*}
\begin{itemize}

   \item We consider a modular XL-MIMO system composed of HB-based multiple uniform planar subarrays (UPSAs) serving ID and EH users located in both NF and FF radiating regions. We also account for the disparate EM visibility linked to the SnS effects in the NF EM channels. 

    
    \item A novel decision criterion, based on the Frobenius norm of channel estimate correlation across subcarriers, is proposed for classifying users into NF and FF regions. Numerical results show that increasing the number of XL-MIMO subarrays enhances the classification accuracy even with few pilot subcarriers. Additionally, our asymptotic analysis reveals that the considered information transfer and EH metrics converge as peripheral subarray contributions diminish, reinforcing the concept of spatial VRs.

    \item We formulate a mixed-integer optimization problem to minimize the overall PC of the proposed XL-MIMO system, subject to fair quality-of-service (QoS) thresholds for both ID and EH user groups, under equal power allocation (PA) and full SA. To solve this, we develop an iterative joint optimization algorithm that decouples the original problem into two routines: a PA routine and an SA routine. The PA routine leverages the Douglas-Rachford (DR) splitting-based alternating direction method of multipliers (ADMM) framework \cite{pontus} to optimize power allocation while satisfying SE and HE constraints. Complementarily, the SA routine strategically activates only the high-performing subarrays, guided by the parametrized PA variables and balancing factors derived from NF and FF user classification. While the PA algorithm inherently assigns lower power to less efficient subarrays, the SA routine further reduces PC by deactivating circuitry associated with these subarrays.

    
    \item The numerical results for the proposed joint PA-SA optimization demonstrate a substantial reduction in the overall PC of the devised XL-MIMO system for all topological scenarios of multiple VRs in both NF and FF regions. 
\end{itemize}

Overall, the proposed framework provides a consolidated design process for a power-efficient MF-SWIPT-based XL-MIMO system. Specifically, the NF/FF classification enables a regime-aware processing and balancing across propagation domains. Moreover, our asymptotic SE and HE analysis quantifies the diminishing contribution of peripheral subarrays and thereby supports VR interpretation under SnS effects, whereas the resulting two-tier PA-SA optimization leverages these insights to deactivate low-impact subarrays and reduce circuit power while preserving the SE and HE QoS guarantees. In this way, classification, analysis, and optimization act coherently towards the single objective of power-efficient SWIPT operation in mixed NF/FF XL-MIMO deployments.

\textit{Notation:} Lower-case and upper-case boldface letters denote vectors and matrices, respectively. The superscripts $(\cdot)^{\rm{H}}$, $(\cdot)^{\rm{T}}$ and $(\cdot)^{\rm{*}}$  represent the Hermitian transpose, transpose and conjugate of a matrix, respectively. The subscripts  $(\cdot)_{\rm{n}}$ and $(\cdot)_{\rm{f}}$ represent parameters associated to NF and FF users, respectively; $\qI_N$ represents the $N\times N$ identity matrix; $\boldsymbol{1}_N$ indicates an all-one vector of size $N$; $\qA{(i,j)}$ denote the $(i,j)$-th element of $\qA$; $[\qA \qB]_{(p,q)}$ denote the $(p,q)$-th element of the product matrix of $\qA$ and $\qB$;
$\Vert \cdot \Vert$ returns the norm of a matrix; $\trace(\cdot)$ denotes the trace of a matrix; $\Vert \cdot\Vert_F$ represents the Frobenius norm of a matrix; $\mathbb{E}\{\cdot\}$ and $\MV\{\cdot\}$ represent the statistical expectation and variance, respectively; 
$\Gamma(a)$ is the Gamma function~\cite[Eq. (8.310)]{Integral:Series:Ryzhik:1992}. Finally, $\delta_{kl}$ is the Kronecker delta function.

 \vspace{-0.5em}
\section{System Model}
We consider a multiuser MIMO system, where a modular XL-MIMO array provides SWIPT services to users located in both NF and FF regions of the array.
This modular array consists of $S$ individually controllable UPSAs, whose activation can be determined by a binary variable $a_s$ \cite{Xinrui}. Each subarray consists of $\ns$ antenna elements, arranged as $\nx$ elements along the $x$-axis and $\ny$ elements along the $y$-axis. These subarray elements can be represented by the set $\mathcal{S}_{(s)}=\{1,2,\hdots, \nx \ny\}$. We consider a HB design, where each subarray is equipped with $N_{RF}< N$ RF chains with the objective of reducing the PC associated with fully connected arrays. We consider $K$ single-antenna users in the network, out of which $M$ EH users are grouped in a single VR, namely $\mathrm{{V}_{\!EH}}$, located in the NF region of the array.\footnote{ The assumption of single EH-VR considers a localized wireless charging hotspot which is commonly encountered in concurrent IoT networks, while the extension to multiple EH-VRs is left for future work.} On the other hand, the remaining $L=K-M$ ID users are divided into $L_{\rm{f}}$ FF users and $L_{\rm{n}}$ NF users, which are further grouped into VRs, called $\mathrm{{V}_{\!ID}}$. The downlink (DL) signals for both ID and EH users are transmitted in the form of independent $K$ streams using a joint precoding matrix $\bar{\qW}_s \!\in\! \mathbb{C}^{K \times \ns}$, which is the product of the digital precoding ${\qW}_{D,s}\!\in \!\mathbb{C}^{K \times \nrf}$ and analog precoding ${\qW}_{A,s}\in \!\mathbb{C}^{\nrf \times \ns}$. Here, ${\qW}_{D,s}$ converts $K$ streams for the input stage of $N_{RF}$ RF chains, while ${\qW}_{A,s}$ tunes the analog phase shifters for the construction of combined transmission signal by $N$ antennas. In view of the different possible scenarios regarding the number of ID VRs ($\mathrm{{V}_{\!ID}}$) and the placement of ID users in either the NF or FF regions, we consider four ($4$) configurations, as illustrated in Fig. \ref{fig:Scenario_XL_MIMO}.
   


\begin{figure*}[t]
    \centering
    \begin{minipage}[t]{0.6\textwidth}
        \vspace{-5.65cm} 
        \centering
        \begin{subfigure}[t]{0.45\linewidth}
            \centering
            \includegraphics[trim=0.5cm 0cm 0.5cm 0cm,clip,width=\linewidth]{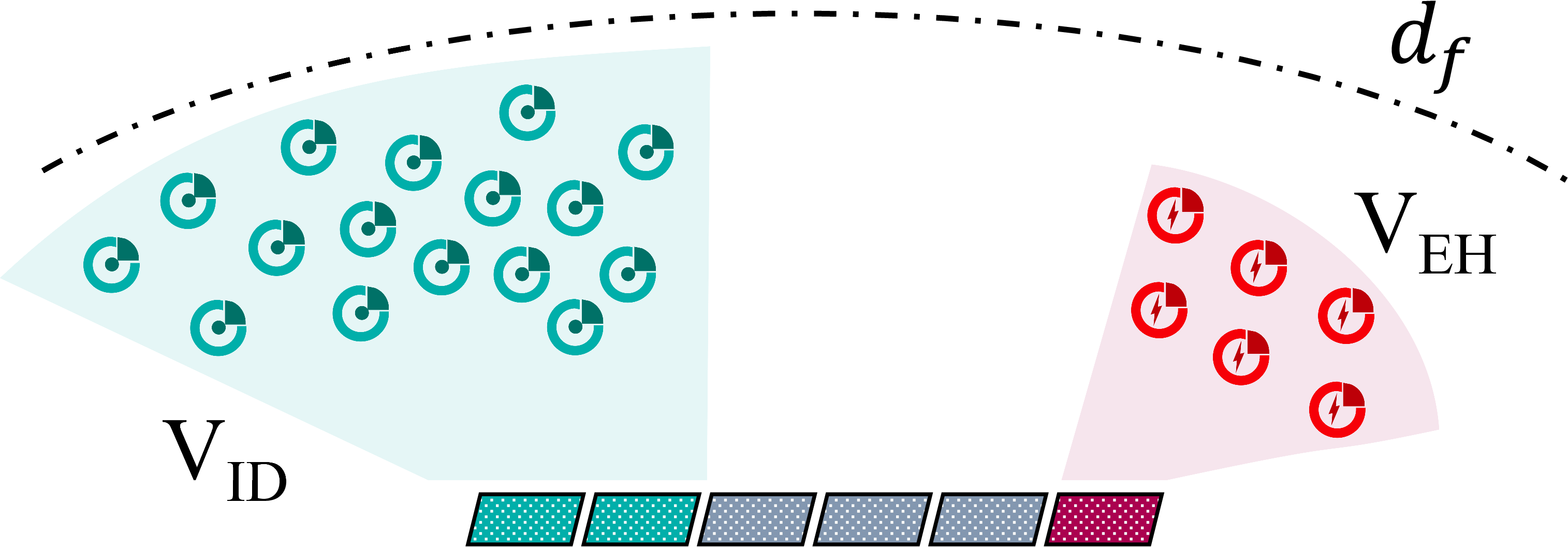}
            \caption{\small NF-VR1 case.\normalsize}\label{fig:NF_VR1}
        \end{subfigure}
        \hfill
        \begin{subfigure}[t]{0.45\linewidth}
            \centering
            \includegraphics[trim=0.5cm 0cm 0.45cm 0cm,clip,width=\linewidth]{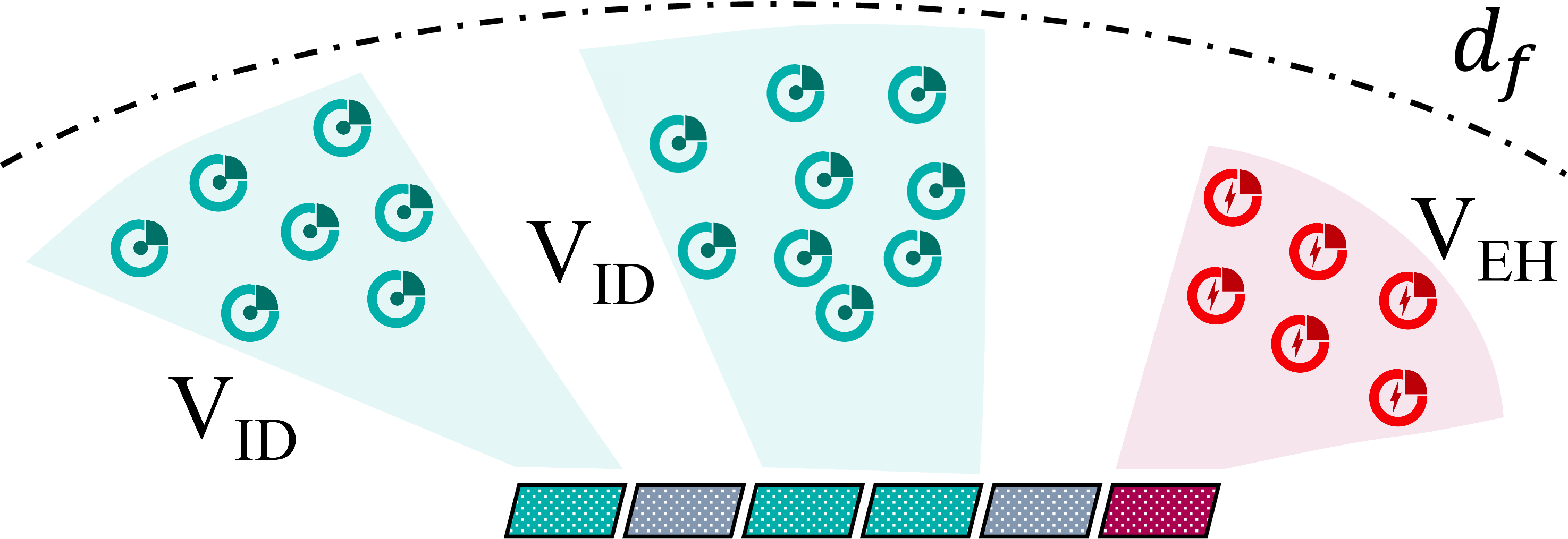}
            \caption{\small NF-VR2 case.\normalsize}\label{fig:NF_VR2}
        \end{subfigure}

        \begin{subfigure}[t]{0.45\linewidth}
            \vspace{0.5em}
            \centering
            \includegraphics[trim=0.5cm 0cm 0.5cm 0cm,clip,width=\linewidth]{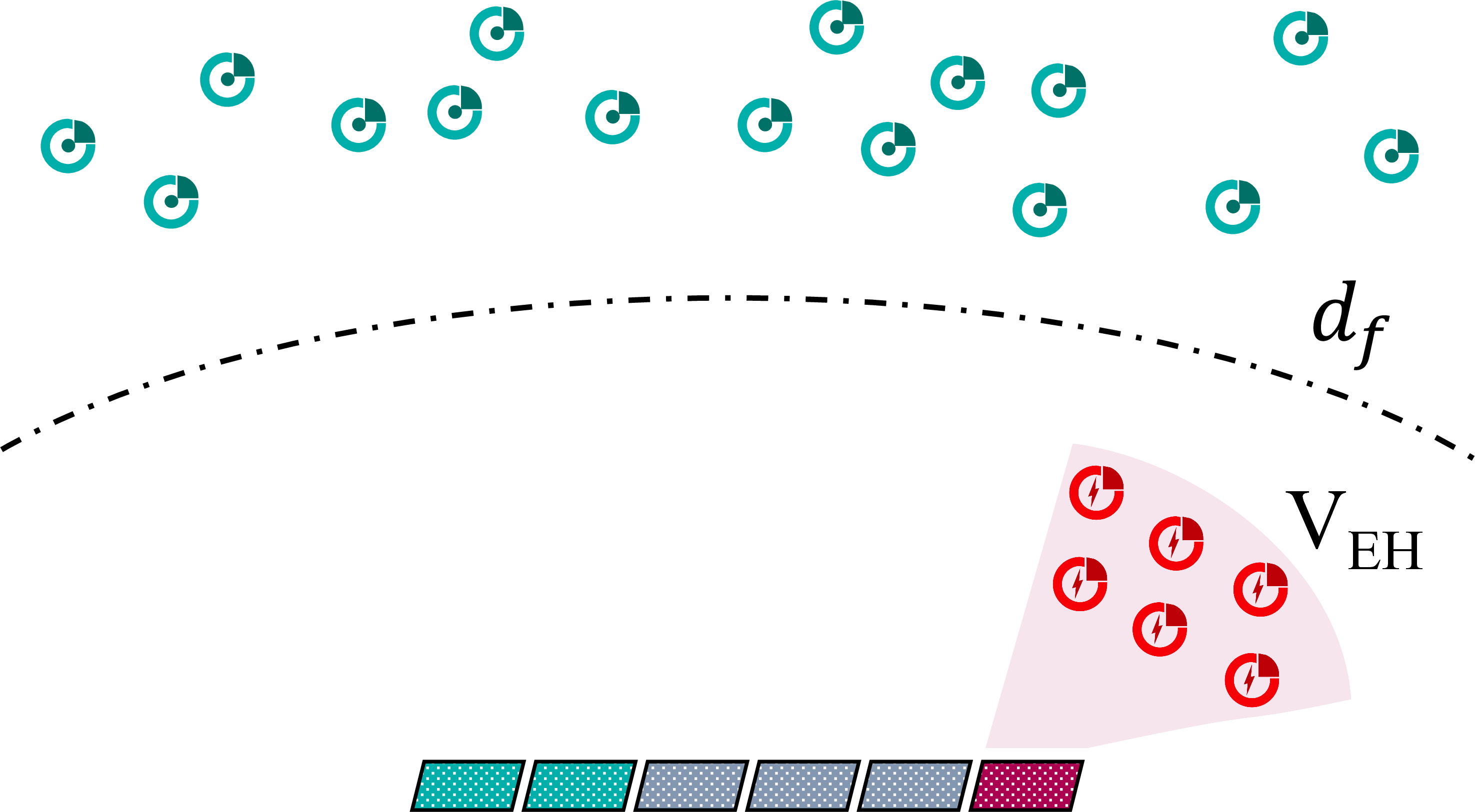}
            \caption{\small MF-VR1 case.\normalsize}\label{fig:MF_VR1}
        \end{subfigure}
        \hfill
        \begin{subfigure}[t]{0.45\linewidth}
            \vspace{0.5em}
            \centering
            \includegraphics[trim=0.5cm 0cm 0.5cm 0cm,clip,width=\linewidth]{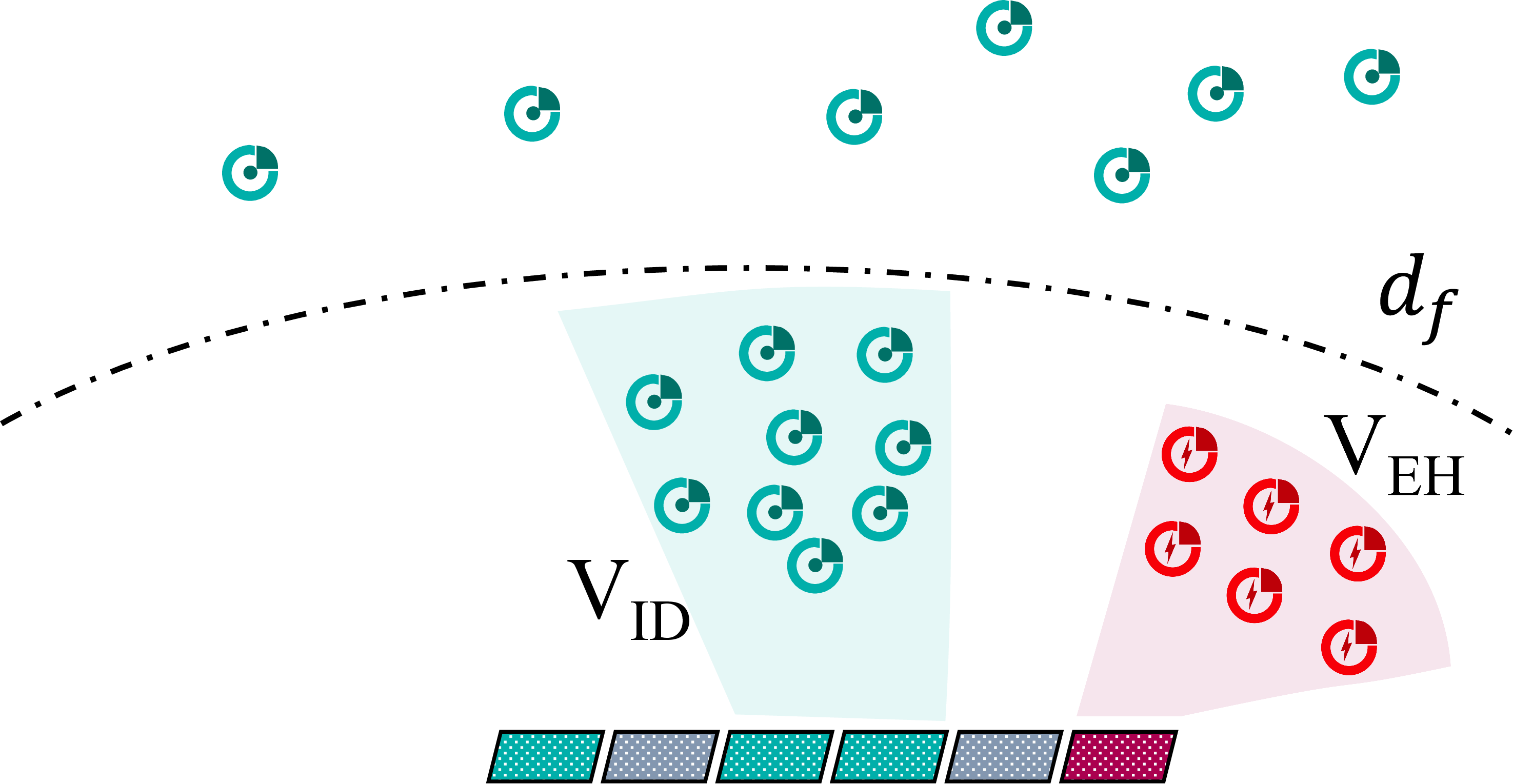}
            \caption{\small MF-VR2 case.\normalsize}\label{fig:MF_VR2}
        \end{subfigure}

        \caption{\small Illustration of different network scenarios.\normalsize}
        \label{fig:Scenario_XL_MIMO}
    \end{minipage}
    \hfill
    \begin{minipage}[t]{0.38\textwidth}
        \centering
        \includegraphics[trim=0.5cm 0cm 2cm 0cm,clip,width=0.95\textwidth]{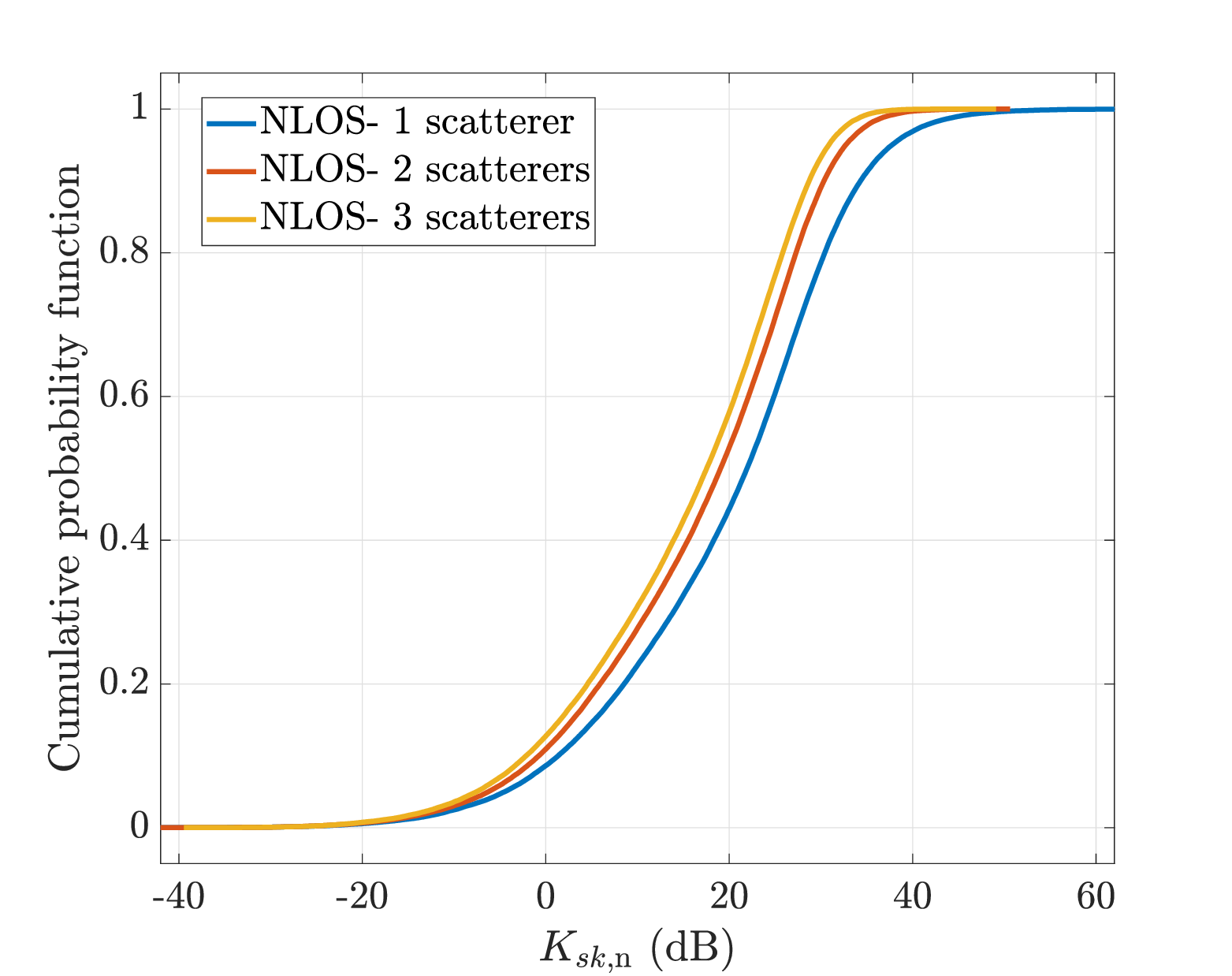}
        \vspace{0.45em}
        \caption{\small CDF of $K_{sk,\rm n}$ for $S=4$.\normalsize}
        \label{fig:CDF_LOS_NLOS}
    \end{minipage}
    \vspace{-1.3em}
\end{figure*}

 \vspace{-0.5em}
\subsection{Channel Model}
We denote the communication channel between a particular subarray $s$ and any user $k$ as $\gsk \in \mathbb{C}^{N \times 1}$, while $\gskn$, $\gskf$ are used for differentiating between the NF and FF channels. The boundary between the NF and FF regions is demarcated by the widely acknowledged Fraunhofer array distance $d_{\rm{f}}=2 D^2 (S \ns)/\lambda$, where $D$ is the largest dimension of each antenna element and $\lambda$ is the carrier wavelength  \cite{Ramezani2024}.
\subsubsection{Near-Field Channel Model}
The channel model for NF user exhibits a spherical EM wavefront behavior within the radiative zone of the transmit array. The deterministic line-of-sight (LoS) channel coefficients can be represented as:
\begin{equation}\label{eq:near_channel}
    \bgskn = \gamsk\bhskn ,
\end{equation}
where
\vspace{-0.5em}
\begin{align*}
    \!\bhskn\!\!= \!\!\Big[e^{\!-j \breve{k} \lVert \qp_k\! -\! \qp_{s_{1,1}} \! \rVert}\!,\!e^{\!-j \breve{k} \lVert \qp_k\! - \!\qp_{s_{1,2}}\!  \rVert}\!,\!\hdots\!,\! 
e^{\!-j \breve{k} \lVert \qp_k \!-\! \qp_{s_{N_x,N_y}}\!  \rVert}\Big]^{{\rm{T}}}\!\!\!,
\end{align*}
denotes the channel vector for user $k$ from a certain subarray $s$ with antenna elements positioned in the $xy$-plane at the locations $(x_{s_{1,1}}, y_{s_{1,1}}), (x_{s_{1,2}},y_{s_{1,2}}), \hdots, (x_{s_{N_x,N_y}},y_{s_{N_x,N_y}})$. The term $e^{-j \breve{k}\lVert \pkr -\qp_{s_{x,y}} \rVert}$ shows the phase component associated to the distance $\lVert \pkr -\qp_{s_{x,y}} \rVert$ traversed by the EM wave from a particular antenna indexed $s_{x,y}$ to user $k$ with wave number defined as $\breve{k}=2\pi/\lambda$. Moreover, the directional gain component of the NF channel is given by  
\begin{equation}\label{eq:directional_gain}
    \gamsk= \dfrac{\lambda}{4 \pi \lVert \pkr -\ps \rVert}\sqrt{F(\theta_{sk})},
\end{equation}
where $F(\theta_{sk})$ is the cosine radiation profile of each antenna element, defined as~\cite{Haiquan}
\begin{align}
    F(\theta_{sk}) &= 
    \begin{cases}
        2 (b_n+1) \,\text{cos}^{b_n}(\theta_{sk}) \quad  \theta_{sk} \in [0,\pi/2],\\
        0 \hspace{3.15cm} \text{otherwise}.
    \end{cases}
\end{align}
Here, $\theta_{sk}$ is the azimuth angle from the subarray $s$ to the user $k$ based on the definition of the $xy$-plane, while $b_n$ is the element directivity related to the boresight gain. We set this parameter $b_n=2$ by following the convention for the standard dipole case, which yields $F(\theta_{sk})\!=\! 6 \cos^2(\theta_{sk})$ for $\theta_{sk} \in [0,\pi/2]$ \cite{Haiyang, Haiyang_Shlezinger_Nir}.

To account for local scattering around NF users, the NF non-line-of-sight (NLoS) term is modeled as a sparse superposition of structured spherical-wave paths \cite{Lu2023}. The scatterers represent a small number of dominant reflectors in the propagation environment, such as walls, objects, or other local structures. Accordingly, each NLoS path is parameterized by its own spatial characteristics, which can be mathematically modeled as
\begin{equation}\label{eq:nf_nlos_component}
    \Tilde{\qg}_{sk,\rm n}=\sum\nolimits_{\ell=1}^{L_p} g_{k,\ell}\,\qa_{\rm NF}(\theta_{k,\ell},d_{k,\ell}),
\end{equation}
where $L_p$ is the number of effective scatterers, $g_{k,\ell}$ is the complex gain of the $\ell$-th scattered path, while $\qa_{\rm NF}(\theta_{k,\ell},d_{k,\ell})$ is the corresponding NF response vector. For the $(x,y)$-th antenna element of subarray $s$, we define $r^{k,\ell}_{s_{x,y}}=\lVert \qp_{s_{x,y}}-\qp^{\rm sc}_{k,\ell}\rVert$ where $\qp^{\rm sc}_{k,\ell}$ is the position of the $\ell$-th virtual scatterer, which is parameterized along the $xz$-plane as $\big(d_{k,\ell}\sin(\theta_{k,\ell}),\,0,\,d_{k,\ell}\cos(\theta_{k,\ell})\big)$. Then, the NF NLoS response vector is defined as
\begin{equation}\label{eq:nf_response_vector}
    \qa_{\rm NF}(\theta_{k,\ell},\!d_{k,\ell})
    \!=\!\frac{1}{\!\!\sqrt{\!N}}
    \Big[
    e^{\!-j\breve{k}(r^{k,\ell}_{s_{1,1}}\!-d_{k,\ell})}\!\!,
    \ldots,\!
    e^{\!-j\breve{k}(r^{k,\ell}_{s_{N_x,N_y}}\!\!-d_{k,\ell})}
    \Big]^{\rm T}\!.\!
\end{equation}
Moreover, the path gain $g_{k,\ell}$ captures both the random scattering coefficient and the distance-dependent attenuation of each NLoS path. It can be represented as
\begin{equation}\label{eq:scatterer_gain}
    g_{k,\ell}=\frac{\lambda}{4\pi d_{k,\ell}}\tilde{g}_{k,\ell},
\end{equation}
where, $\tilde{g}_{k,\ell}\sim\mathcal{CN}\!\left(0,{1}/{L_p}\right)$. Since $\mathbb{E}\{\tilde{g}_{k,\ell}\}=0$, it follows that $\mathbb{E}\{\Tilde{\qg}_{sk,\rm n}\}=\boldsymbol{0}_{N}$ and, therefore, $\mathbb{E}\{\gskn\}=\bgskn$. As such, the resultant NF channel statistics remain centered around the deterministic LoS component $\bgskn$, while the random deviation around this center is governed by a finite number of structured NF scattered paths.

It can be substantiated that the NF LoS component $\bgskn$ remains the dominant contributor to the overall NF channel power even after incorporating the sparse NF NLoS component $\Tilde{\qg}_{sk,\rm n}$. To quantify this effect, we evaluate the effective NF power ratio as
 \vspace{-0.5em}
\begin{equation}
    K_{sk,\rm n}
    \triangleq
    10\log_{10}\!\left(
    \frac{\|\bgskn\|^2}{\|\Tilde{\qg}_{sk,\rm n}\|^2}
    \right)\,(\text{dB}).
\end{equation}
Consistent with \cite{Lu2023}, we consider $L_p=\{1,2,3\}$ NF scattered paths per user, with $\theta_{k,\ell}\sim\mathcal{U}[-\pi/3,\pi/3]$ and $d_{k,\ell}\sim\mathcal{U}[2D^2/\lambda,d_{\rm f}]$. The cumulative distribution function (CDF) of $K_{sk,\rm n}$ validates the assumption of a significantly dominant LoS component in the NF regime as shown in Fig. \ref{fig:CDF_LOS_NLOS}. For the representative cases of $L_p=\{1,2,3\}$ NLoS scatterers, the mean values of $K_{sk,\rm n}$ ($\textrm{CDF}=0.5$) are $19.63$ dB, $16.62$ dB, and $15.13$ dB, respectively. These values imply that the NF channel remains predominantly LoS-oriented in the power domain, even in the more scattering-rich cases, such as $L_p=3$. This interpretation is further supported by the CDF values at $0$ dB, which directly represent the probability that the NLoS power is at least as large as the LoS power, i.e.,
\begin{equation}
    \Pr\!\left\{K_{sk,\rm n} \leq 0 \,\text{dB}\right\}
    =
    \Pr\!\left\{\|\Tilde{\qg}_{sk,\rm n}\|^2 \geq \|\bgskn\|^2\right\}.
\end{equation}
For $L_p=\{1,2,3\}$, these probabilities are only $0.086$, $0.109$, and $0.127$, respectively. Equivalently, the LoS component is stronger than the NLoS component in approximately $91.37\%$, $89.06\%$, and $87.26\%$ of the channel realizations. These results establish that the deterministic LoS model captures the dominant propagation behavior in the NF region of our proposed XL-MIMO system, while the additional sparse NLoS components act as a secondary perturbation rather than a performance-defining factor.
\begin{Remark}
     In the remainder of this paper, the NF channel $\gskn$ is used interchangeably with its dominant LoS component $\bgskn$, while the sparse NF NLoS term $\Tilde{\qg}_{sk,\rm n}$ is neglected in the ensuing analysis. This convention is supported by Fig.~\ref{fig:CDF_LOS_NLOS}, which confirms that the LoS component dominates the NF channel power in the vast majority of channel realizations.
 \end{Remark}

\vspace{-1em}
\subsubsection{Far-Field Channel Model} 
 Ricean fading is used to model the FF user channels with the presence of a dominant LoS component from the XL-MIMO array to these users. Therefore, $\gskf$ can be mathematically expressed as \cite{Zeeshan}
\begin{equation}~\label{eq:far_channel}
    \gskf= \sqrt{\dfrac{\zekl}{K_{sk,\rm f}+1}}
    \left(\sqrt{K_{sk,\rm f}}\bhskf +\tgklf\right),
\end{equation}
where $K_{sk,\rm f}$ is the Ricean $K$-factor; the LoS component is $\bhskf = [1,e^{j\pi\sin(\phi_{sk})},\ldots,e^{j(N_x N_y)\pi\sin(\phi_{sk})} ]^{\rm{T}}$ with angle-of-arrival (AoA) $\phi_{sk}$; the NLoS component follows a complex Gaussian distribution as $\tgklf\sim\mathcal{N}_C(\boldsymbol{0}, \qI_N)$, while $\zekl$ represents the large-scale fading coefficient. Let $\bsk\! \triangleq\!{{\zekl}/{(K_{sk,\rm f}\!+\!1)}}$, $\bklq\!\triangleq\! \bsk K_{sk,\rm f}$, and
$\bgskf\triangleq \sqrt{\bklq} \bhskf$. Then, we can simplify \eqref{eq:far_channel} as
\begin{equation} \label{eq:far_field_channel}
    \gskf= \bgskf +\sqrt{\bsk}\tgklf.
\end{equation}

Furthermore, we assume that these users are either static or moving slowly in the dynamic surrounding environment \cite{Hien:Asilomar:2018}. Hence, we can assume that the large-scale channel coefficients $\bgskf$ are almost stationary over multiple coherence intervals. Based on this assumption, the channel characteristics can be estimated over multiple temporal channel realizations using LS estimation \cite{Wang:JIOT:2020} as explained in the next subsection. We also collect these channel statistics across multiple subcarriers to perform the NF or FF user classification. 

\vspace{-1em}
\subsection{Channel Estimation}\label{sec:chan_est}
For the accurate estimation of the channel $\gsk$ between sub-array $s$ and user $k$, we use the same pilot sequence $\pmb{\varphi}_{k} \in \mathbb{C}^{\tau_p \times 1}$ of $\tau_p$ symbols for $Q$ temporal observations of this channel \cite{cui3}. It is assumed that the pilot signals used for all the users are mutually orthogonal, such that $\pmb{\varphi}^{\rm{T}}_{k}\pmb{\varphi}^*_{k'}=\tau_p \geq K$ for $k'=k$, and otherwise $\pmb{\varphi}^{\rm{T}}_{k}\pmb{\varphi}^*_{k'}=0$.

The received pilot signal $\qY^{(q)}_s \in \mathbb{C}^{\nrf \times \tau_p}$ at sub-array $s$ during the $q$-th pilot observation slot can be represented as:
\begin{equation}\label{eq:rx_pilot_sig}
    \qY^{(q)}_s = \sumk \sqrt{P_p}\,{\qW}^{(q)}_{A,s} \gsk \pmb{\varphi}^{\rm T}_{k} + \qN^{(q)}_s,
\end{equation}
where the entries of ${\qW}^{(q)}_{A,s}$ are randomly generated from the set ${\qW}^{(q)}_{A,s}(i,j) \!\in\! \dfrac{1}{\sqrt{N}} \{-1,1\}$, following a Rademacher distribution \cite{cui3}. Moreover, $P_p$ is the uplink pilot power, while  $\qN^{(q)}_s \!\in\! \mathbb{C}^{\nrf \times \tau_p}$ is the additive white Gaussian noise (AWGN) matrix with independent, identically distributed (i.i.d.) $\mathcal{CN}(0, \sigma^2)$ elements. 
The projection of $\qY^{(q)}_s$ onto the pilot signal $\pmb{\varphi}_k$ can yield the sufficient statistic of the channel estimate for $\gsk$ \cite{Zeeshan}. Mathematically, it can be expressed as
\begin{equation}\label{eq:proj_pilot}
    \tilde{\qy}^{(q)}_{sk}= \dfrac{\qY^{(q)}_s \pmb{\varphi}^*_k}{\sqrt{\tau_p}}= \sqrt{\tau_p P_p} \,{\qW}^{(q)}_{\!\!A,s} \gsk + \qn^{(q)}_s,
\end{equation}
where ${\qn}^{(q)}_s={\qN^{(q)}_s \pmb{\varphi}^*_k}/{\sqrt{\tau_p}} \in \mathbb{C}^{\nrf \times 1}$. Now, the overall received pilot sequence $\tilde{\qy}_{sk}=\Big[\tilde{\qy}^{(1)^{\rm T}}_{sk}\!\!, \tilde{\qy}^{(2)^{\rm T}}_{sk}\!\!, \hdots,  \tilde{\qy}^{(Q)^{\rm T}}_{sk}\Big]^{\rm T}\!\!\!\in\! \mathbb{C}^{Q \nrf \times 1}$ at subarray $s$, is represented as 
\begin{equation}
    \tilde{\qy}_{sk}= \sqrt{\tau_p P_p} \,\hat{{\qW}}_{\!\!A,s} \gsk + \tilde{\qn}_s,
\end{equation}
where $\hat{{\qW}}_{\!\!A,s}\!=\!\Big[{\qW}^{(1)^{\rm T}}_{\!\!A,s}\!,\!{\qW}^{(2)^{\rm T}}_{\!\!A,s}\!, \hdots,\! {\qW}^{(Q)^{\rm T}}_{\!\!A,s}\Big]^{\rm T} \!\!\!\in\! \mathbb{C}^{Q \nrf \times N}$ is the overall observation matrix, such that $Q \nrf \!\geq\! N$ to preserve the degrees of freedom in the observation process. The overall noise vector is given as $\tilde{\qn}_s=\Big[\qn^{(1)^{\rm T}}_s, \qn^{(2)^{\rm T}}_s, \hdots, \qn^{(Q)^{\rm T}}_s\Big]^{\rm T}\!$. Now,
we obtain the LS estimate of the channel $\gsk$ as 
\begin{align} \label{eq:CE}
    \hgsk &= \dfrac{{\hat{\qW}}^{\dagger}_{A,s} \tilde{\qy}_{sk}}{\sqrt{\tau_pP_p}} = \gsk + \tesk, 
\end{align}
where $\tesk={\hat{\qW}^{\dagger}_{A,s} \tilde{\qn}_s/(\sqrt{\tau_pP_p})}$, with $\hat{\qW}^{\dagger}_{A,s}$
$ \in \mathbb{C}^{N \times Q \nrf}$ being the Moore-Penrose pseudo-inverse of $\hat{\qW}_{A,s}$.
Here, the statistics of channel estimates are different for NF and FF users as the FF channel estimate also involves NLoS component.
\vspace{-0.5em}
\begin{Lemma}~\label{prop:CE_stat}
 The statistics of the channel estimate, $\hgsk$, can be given as   
\begin{align}
    \ME\{\hgsk\} &= (1-\delta_{nf})\gamsk\bhskn+\delta_{nf}\sqrt{\bklq} \bhskf,\\
    \MV\{[\hgsk]_{t}\} &= \delta_{nf}\bsk + \nue,
\end{align}
where 
\vspace{-0.5em}
\begin{equation}\label{eq:var_error}
    \nue \approx \dfrac{\sigma^2}{\tau_pP_p}  \dfrac{N}{Q \nrf}\bigg(1+ \dfrac{N -1}{Q \nrf} \bigg),
\end{equation}
while $\delta_{nf}$ is 1 for the FF user and 0 for the NF user.
\end{Lemma}
\begin{proof}
    See Appendix~\ref{app:CE_stat}.
\end{proof}
\vspace{-1.5em}
\subsection{Decision Making for NF or FF User}\label{sec:NF_FF_classification}
The classification of a user as either a NF or FF user is determined based on the channel statistics. By inspecting the NF and FF channel models in (1) and (9), respectively, we note that under the considered narrowband pilot-subcarrier setting, the NF spherical-wave response is more sensitive to wavelength variations since it depends on the exact element-wise distances $\lVert\pkr-\qp_{s_{x,y}}\rVert$, whereas the adopted FF response follows the planar-wave model parameterized mainly by the AoA $\phi_{sk}$. Therefore, we can exploit the stronger frequency-dependent spatial variation of NF channels relative to the narrowband FF model in order to classify a user within either the NF or FF region. In this context, we observe $U$ instances of channel estimates $\hat{\qg}^{(1)}_{sk}, \hat{\qg}^{(2)}_{sk}, \hdots, \hat{\qg}^{(U)}_{sk}$ over independent frequency sub-carriers. Here, we propose an empirical frequency correlation-based criterion to determine the NF or FF status, leveraging the correlation of channel estimates across multiple distinct subcarriers.\footnote{For two subcarriers $f_u$ and $f_v$, the NF relative phase variation is proportional to $\frac{2\pi(f_u-f_v)}{c}(\lVert\pkr-\qp_{s_{x,y}}\rVert-\lVert\pkr-\qp_{s_{x',y'}}\rVert)$. In contrast, the adopted FF steering vector satisfies $[\bhskf]_n=e^{j(n-1)\pi\sin(\phi_{sk})}$ and is treated as invariant across the closely spaced pilot subcarriers under the narrowband FF approximation.} Using the channel estimation matrix $\hat{\qG}^{u}_{sk}\triangleq\big[\hat{\qg}^{(1)}{sk}, \hat{\qg}^{(2)}_{sk}, \hdots, \hat{\qg}^{(U)}_{sk}\big]$ for user $k$ across all subcarriers and the sample mean of estimates $\bar{\hat{\qg}}{sk}$, the frequency correlation matrix $\qA_{{sk}} \in \mathbb{C}^{U \times U}$ is defined as:
\begin{equation}\label{eq:correlation}
    \qA_{{sk}}= \dfrac{1}{N-1} \left(\hat{\qG}^{u}_{sk}- \bar{\hat{\qg}}_{sk} \boldsymbol{1}^{\rm T}_{U} \right)^{\!\rm H}\!\left(\hat{\qG}^{u}_{sk}- \bar{\hat{\qg}}_{sk} \boldsymbol{1}^{\rm T}_{U}\right).
\end{equation}

This statistical parameter encompasses the channel variation for each pair of user $k$ and subarray $s$ over the available frequency spectrum. To utilize this parameter for meaningful user classification, we define the normalized metric based on the Frobenius norm of this matrix as
\begin{equation}
\xi_{sk} = \dfrac{1}{U} \sqrt{\trace \left(\qA_{{sk}} \qA^{\rm{H}}_{{sk}}\right)}.
\end{equation}

For a particular user $k$, the average of this metric $\bar{\xi}_{k} = {\sus \xi_{sk}}/{S}$ provides the measure of spectrum-based channel variation across the whole XL-MIMO array.
We use this parameter as our decision-making metric against the threshold of the variance $\MV{[\qg_{\rm{f}}]{t}}= \beta{\rm{f}}$ of the channel $\qg_{\rm{f}}$ for a user at Fraunhofer array distance $d_{\rm{f}}$.
We evaluate $\beta_{\rm{f}}$ in \eqref{eq:far_field_channel} by substituting $d_{sk}=d_{\rm{f}}$ in \eqref{eq:PL}. We present the normalized Frobenius norm criterion for NF or FF decision-making as:
\begin{itemize}
    \item The user $k$ is considered to be located within the NF region, if $\bar{\xi}_{k} \geq \beta_{\rm{f}}$, which is attributed to the higher channel variation due to strong dependence of the NF channel directional gain on the wavelength (see \eqref{eq:directional_gain}).
    \item Otherwise, this user is considered to be located within the FF region, if $\bar{\xi}_{k} <  \beta_{\rm{f}}$, due to the frequency independent behavior of FF channels.
\end{itemize}
 
Based on this empirical criterion, we can classify the users with respect to their corresponding EM field regions. This information is then utilized to compute the balancing factors for the SA optimization routine presented in Section \ref{sec:SA_routine}. Unlike NN-aided Hankelization-based classifiers requiring singular-value extraction and offline training \cite{Kim_Access_2024}, the proposed metric is training-free, CSI-driven, and directly supports \textbf{PA-SA} optimization, while $\xi_{sk}$ can inform future learning-based extensions.

\vspace{-0.8em}
\subsection{Precoding Matrix}
Using the channel estimates, we formulate the digital precoding matrix $\qW_{D,s}$ by employing the zero-forcing (ZF) matrix $\qV_s \in \mathbb{C}^{L \times N_{RF}}$ for the ID users and the maximum ratio transmission (MRT) matrix $\qW_s \in \mathbb{C}^{M \times N_{RF}}$ for the EH users. The ZF precoding vector $\qv_{sl} = [\qV_s]_{(:,l)}$ for ID user $l$, can be written as $\qv_{sl} = \kzf {\bm{\upsilon}_{sl}}$. Here, $\kzf\triangleq1/\sqrt{\ME\{\lVert \bm{\upsilon}_{sl}\rVert\}}$ is a normalization factor, while 
$\bm{\upsilon}_{sl}$ corresponds to the $l$-th row of the Moore–Penrose pseudo-inverse matrix $\qG^{\dagger}_s = \hat{\qG}_s(\hat{\qG}^H_s\hat{\qG}_s)^{-1}$.
Moreover, the MRT precoding vector $\qw_{sm}=[\qW_s]_{(:,m)}$ for EH user $m$ can be formed as $\qw_{sm}= \kmrt{\hgsm}$, with the normalization factor defined as $\kmrt\triangleq1/\sqrt{\ME\{\lVert\hgsm\rVert\}}$. 

\vspace{-0.4em}
\section{Performance Analysis of SWIPT Operation}
In this section, we discuss SWIPT operation for different MF scenarios with a HB-based modular XL-MIMO array. The signal transmitted by subarray $s$ for both users is given as
\vspace{-0.1em}
\begin{align}\label{eq:transmitted_signal}
    \qx_s =& \as \Big(\sumln \sqrt{\osln} {\qv}^*_{sl,\rm{n}} x_{l} + \sumlf \sqrt{\oslf} {\qv}^*_{sl,\rm{f}} x_{l}\nonumber\\
    &+ \summ \sqrt{\osm} {\qw}^*_{sm} e_{m} \Big),
\end{align}
where $x_{l} \in \mathbb{C}$ denotes the information symbol for the ID user $l$, while $e_{m}\sim \mathcal{CN}(0,1)$ is the normalized zero-mean psuedo-random energy signal. Both these uncorrelated signals satisfy the condition $\ME\{\lvert x_{l} \rvert^2\}=\ME\{|e_{m}|^2\}=1$.  
The variable $\as$ is a binary SA indicator that specifies the operational status of subarray $s$. It can be mathematically represented as
\vspace{-0.2em}
\begin{align}
    \as &= 
    \begin{cases}
        1 \quad \quad \text{if subarray $s$ is switched on},\\
        0 \quad \quad \text{if subarray $s$ is not switched on}.
    \end{cases}
\end{align}

Moreover, in \eqref{eq:transmitted_signal}, $\osl$ and $\osm$ are the PA coefficients for the ID and EH users, respectively. These variables are assigned subject to the power constraint at each subarray, such that
\vspace{0.1em}
\begin{equation}\label{eq:sub_array_power_lim}
    \mathbb{E}\left\{\lVert \qx_s \rVert^2 \right\}\!=\! \sumln \osln\! +\! \sumlf \oslf \!+\! \summ  \osm \!  \leq\! P_s,
\end{equation}
where $P_s= \ns P_{\mathrm{et}}$ is the maximum subarray transmit power, while $P_{\mathrm{et}}$ is the maximum transmit power per antenna element. The maximum XL-MIMO transmit power of all modular subarrays is $P_t = S P_s$. 

 Moreover, the aggregate power consumed by the whole XL-MIMO array during the DL phase can be modeled as~\cite{Jun_Zhang}
\vspace{-0.1em}
\begin{align}\label{eq:power_consumed}
   \! P_C(\OID\!\!,\OEH\!\!,\qa) =& \sus \as \bigg[\frac{1}{\varsigma} \bigg( \!\sumln\! \, \osln\!+ \!\sumlf\! \, \oslf\!\nonumber\\
    &\!\!+ \!\summ \osm \bigg)\! +\! 2 P_{\mathrm{syn}}\! + \! N_{RF} P_{\mathrm{ct}}\bigg],\!
\end{align}
where $\OID\triangleq \{\osln\}_{l=1,\ldots,L_{\rm{n}}}^{s=1,\ldots,S} \cup \{\oslf\}_{l=1,\ldots,L_{\rm{f}}}^{s=1,\ldots,S}$;  $\OEH\triangleq \{\osm\}_{m=1,\ldots,M}^{s=1,\ldots,S}$; $\qa=\{a_s\}_{s=1,\ldots,S}$; $0<\varsigma<1$ is the power amplifier efficiency; $P_{\mathrm{syn}}$ is the power consumed by the frequency synthesizer of each subarray, whereas $P_{\mathrm{ct}}$ is the circuit power consumed by each RF chain. Note that, the power consumed by the individual phase shifter is ignored in our model.

Accordingly, the corresponding received signal at the ID user $l$ and EH user $m$ can be expressed as
\vspace{0.1em}
\begin{subequations}
    \begin{align}
        r^{\mathsf{ID}}_{l} & = \sus \qg^T_{sl} \qx_s + n_{l}, \quad \,\,\,\,\, \forall \, l=1,2, \hdots, L, ~\label{eq:rl} \\
        r^{\mathsf{EH}}_{m} &= \sus \qg^T_{sm} \qx_s+ n_{m}, \quad  \forall \, m=1,2, \hdots, M,
    \end{align} 
\end{subequations}
respectively, where $n_l$ and $n_{m} \sim \mathcal{CN}(0, \sigma^2_l)$ are the AWGN terms at ID user $l$ and EH user $m$, respectively. 
\vspace{-1em}
\subsection{Downlink Data Transmission}
Using $r^{\mathsf{ID}}_l$ in \eqref{eq:rl}, the 
DL achievable SE of ID user $l$ can be expressed as
\begin{align}~\label{eq:DL_SE}
\mathrm{SE}_l(\OID\!\!,\OEH\!\!,\qa) &\!=\! \dfrac{\tau_c\!-\!Q\tau_p}{\tau_c} \log_2 \!\Bigg(\!1\!+\! \dfrac{\PIDC}{\PIDNC}\!\Bigg),\!
\end{align}
where $\PIDC$ represents the desired coherent signal aimed for ID user $l$ from all subarrays, while $\PIDNC$ is the composite noncoherent signal, which is composed of the interference from other ID and EH signals along with the AWGN. These two terms are given by
\vspace{-0.2em}
\begin{subequations}~\label{eq:DL_SE2}
    \begin{align}
        \PIDC &\triangleq\sus \as \osl  \lvert\qg^{\rm{T}}_{sl} {\qv}^*_{sl} \rvert^2,\\
        \PIDNC &\triangleq \sum\nolimits_{l' \neq l}  \sus   \as \Omega_{sl'}^{\ID}  \lvert\qg^{\rm{T}}_{sl} {\qv}^*_{sl'} \rvert^2\nonumber \\
        &\hspace{-2em} + \summ \sus  \as \osm  \lvert\qg^{\rm{T}}_{sl} {\qw}^*_{sm} \rvert^2 \!+ \sigma^2_{l}.
    \end{align}
\end{subequations}
\vspace{-0.5em}
\subsection{Asymptotic Analysis for Downlink SE}\label{sec:asymptotic_DL}

If we consider the SnS effects of the proposed XL-MIMO system model, it can be deduced that the DL SE performance $\mathrm{SE}_l$ of a particular user $l$ is predominantly supported by a certain subset of subarrays which lie within its VR. For a fixed number of total users, $K$, and maximum transmit power $P_s$, $\mathrm{SE}_l$ is expected to converge to a certain level with an increasing number of modular subarrays. To evaluate the asymptotic limit on the DL SE, we examine the limit $S \rightarrow \infty$ for the operational scenario of equal PA for all users $\Omega^{\mathsf{ID-EA}}_{sl}\!\!=\! \Omega^{\mathsf{EH-EA}}_{sm}\!\!=\! P_s/K$ and full array activation, i.e., $\qa={\textbf{1}}_S$. Using Tchebyshev's theorem, the constituent terms of the subequations in \eqref{eq:DL_SE2} become \cite[Eq. (7.121)]{papoulis}, 
\begin{subequations}
    \begin{align}
        &\!\!\!\!\dfrac{1}{S} \sus\! \as \osl  \lvert\qg^{\rm{T}}_{sl} {\qv}^*_{sl} \rvert^2 \!\! \xrightarrow[S \rightarrow \infty]{P} \! \dfrac{1}{S} \sus \!\dfrac{P_s}{K} \ME\{\lvert\qg^{\rm{T}}_{sl} {\qv}^*_{sl} \rvert^2\},\!\label{eq:asymptotic_coherent_DL_ID}\\
        &\vspace{-0.5em}\!\!\!\!\dfrac{1}{S} \sum\nolimits_{l' \neq l} \sus  \as \Omega_{sl'}^{\ID}  \lvert\qg^{\rm{T}}_{sl} {\qv}^*_{sl'} \rvert^2 \nonumber\\
        &\hspace{1cm} \xrightarrow[S \rightarrow \infty]{P} \dfrac{1}{S} \sum\nolimits_{l' \neq l} \sus  \dfrac{P_s}{K} \ME\{\lvert\qg^{\rm{T}}_{sl} {\qv}^*_{sl'} \rvert^2\},\label{eq:asymptotic_noncoherent_DL_ID} \\
        &\!\!\!\! \dfrac{1}{S}\summ \sus  \as \osm  \lvert\qg^{\rm{T}}_{sl} {\qw}^*_{sm} \rvert^2 \nonumber\\
        &\hspace{1cm}\xrightarrow[S \rightarrow \infty]{P} \dfrac{1}{S}\summ \sus \dfrac{P_s}{K} \ME\{ \lvert\qg^{\rm{T}}_{sl} {\qw}^*_{sm} \rvert^2\},\label{eq:asymptotic_noncoherent_DL_ID_EH}
    \end{align}\label{eq:asymptotic_DL}
\end{subequations}
where $\xrightarrow[S \rightarrow \infty]{P}$ denotes the convergence in probability as $S \rightarrow \infty$. The expectation terms in \eqref{eq:asymptotic_DL} represent the received ID signals components in form of coherent ID, non-coherent ID and non-coherent EH received signals, respectively. Furthermore, these signals constitute the cases of the NF and FF channels along with the precoding vectors.

\begin{table*}
    \centering
    \caption{{Asymptotic expressions for DL SE}}
    \vspace{-0.5em}
    \begin{tabular}{|p{1.9cm}|p{3.5cm}|p{3.5cm}|p{3.5cm}|p{3.5cm}|}
    
    \hline
    \centering \vspace{0em}\textbf{Signal Type} & \centering \vspace{0em}\textbf{NF user }& \centering \vspace{0em}\textbf{FF user } & \centering \textbf{NF user \\ FF precoding}  & \centering \textbf{FF user \\ NF precoding}
    \cr
    \hline
    \centering Coherent ID signal \eqref{eq:asymptotic_coherent_DL_ID} & \centering $\!\!\kappa^{2}_{sl} \gamma^{2}_{sl}\sum^K_{k=1} \sum^K_{k'=1} \! \vartheta_{sl,kk'} \times \Theta_{n}(s,l,k,k') $ & \centering $\kappa^{2}_{sl}\sum^K_{k=1} \sum^K_{k'=1} \vartheta_{sl,kk'} \times \Theta_{\rm{f}}(s,l,k,k')$ & \centering \vspace{-0.3em}0 & \centering \vspace{-0.3em}0
    \cr
    \hline

    \centering \vspace{0.5em} Non-Coherent ID signal \eqref{eq:asymptotic_noncoherent_DL_ID} & \centering \vspace{0.5em}$\!\!\!\kappa^2_{sl'} \gamma^2_{sl} \sum^K_{k=1} \sum^K_{k'=1}\!\vartheta_{sl'\!,kk'}\times\Theta_{n}(s,l,k,k')$& \centering \vspace{0.5em}$\!\!\!\kappa^2_{sl'} \gamma^2_{sl} \sum^K_{k=1} \sum^K_{k'=1}\!\vartheta_{sl'\!,kk'} \times\Theta_{\rm{f}}(s,l,k,k')$ & \centering $\!\!\!\kappa^2_{sl'}\gamma^2_{sl}\sum^K_{k=1} \sum^K_{k'=1}\! \vartheta_{sl'\!,kk'} \times \big(\delta_{kk'}\varsigma_{sk}\doubleacute{\varrho}^2_{s,kl}\!+\!(1\!-\!\delta_{kk'})\sqrt{\varsigma_{sk}\varsigma_{sk'}}\doubleacute{\varrho}_{s,kl}\doubleacute{\varrho}_{s,k'l}+\sqrt{\beta_{sk}\beta_{sk'}} \delta_{kk'}N+\delta_{kk'}\nue\big)
    $ & \centering $\kappa^2_{sl'} \sum^K_{k=1} \sum^K_{k'=1}  \vartheta_{sl',kk'}\!\!\!\times\big(\gamma_{sk}\gamma_{sk'}\big( \varsigma_{sl}\doubleacute{\varrho}_{s,kl}(\delta_{kk'}\doubleacute{\varrho}_{s,kl}+\doubleacute{\varrho}_{s,k'l})+\beta_{sl}\delta_{kk'}N\big)+(\varsigma_{sl}+\beta_{sl})\delta_{kk'}\nue \big)$
    \cr
    \hline
    \centering Non-coherent EH signal \eqref{eq:asymptotic_noncoherent_DL_ID_EH} & \centering \vspace{-0.25em}$\kappa^2_{sm} \gamma^2_{sl}\big( \gamma^2_{sm} \rsmlnsq + \nue\big)$ & \centering$\kappa^2_{sm} \big( \gamma^2_{sm} (\varsigma_{sl}\doubleacute{\varrho}^2_{s,ml} +  \beta_{sl} N)+(\varsigma_{sl} +\beta_{sl})\nue \big)$ & \centering \vspace{-0.3em}0 & \centering \vspace{-0.3em}0
    \cr
    \hline
    \end{tabular}\label{tab:asymptotic_DL_SE}
    \vspace{-0.5em}
\end{table*}


\begin{proposition}\label{prop:asymptotic_DL_SE} The closed-form expressions for the expectation terms of the received ID signal components are provided in Table \ref{tab:asymptotic_DL_SE}. The following relationships are established for these derivations:
\subsubsection{For the NF cases}
\begin{subequations}
\begin{align}
    &\Theta_{n}(s,l,k,k')\triangleq  \delta_{kl}\delta_{k'l} \gamma^2_{sk} N^2+\gamma_{sk}\gamma_{sk'} \big(N(\delta_{kl}\varrho_{s,k'l}\nonumber\\
   &\hspace{2em}+\delta_{k'l}\varrho_{s,kl})+\varrho_{s,kl} \varrho_{s,k'l}\big)+\delta_{kk'}\nue,
   \\
    &\vartheta_{sl,kk'}\triangleq\ME\{\varpi_{s,lk} \varpi^*_{s,lk'}\}, \,\varpi_{s,lk}=\big[\hat{\qG}^{\rm{T}}_s\hat{\qG}^*_s\big]_{(l,k)}^{-1},\nonumber\\
    &\varrho_{s,kl}= \bar{\qh}^{\rm{T}}_{sk}\bar{\qh}^*_{sl} \!=\!\!\! \sum\limits_{x,y \in \mathcal{S}_{(s)}}\!\!\! e^{-j k \lVert \qp_{k}\! -\! \qp_{s_{x,y}}  \rVert}e^{j k \lVert \qp_{l}\! -\! \qp_{s_{x,y}}  \rVert},  
\end{align}    
\end{subequations}
while the case $k=l$ yields $\varrho_{s,kl}=N$ and $\delta_{kl}=1$, otherwise $\delta_{kl}=0$.
\subsubsection{For the FF cases}
\begin{align}
    &\Theta_{\rm{f}}(s,l,k,k')=\varsigma_{sl} \Big(\sqrt{\varsigma_{sk}\varsigma_{sk'}}( \delta_{kl}\delta_{k'l} N^2+(\delta_{kl}\acute{\varrho}_{s,k'l}\nonumber\\
   &\hspace{1cm}+\delta_{k'l}\acute{\varrho}_{s,kl})N+    \acute{\varrho}_{s,kl}\acute{\varrho}_{s,k'l})+\sqrt{\beta_{sk}\beta_{sk'}}\delta_{kk'}N\nonumber\\
   &\hspace{1cm}+\delta_{kk'}\nue\Big)+\beta_{sl}\Big(\sqrt{\varsigma_{sk}\varsigma_{sk'}}\delta_{kk'}N+\sqrt{\beta_{sk}\beta_{sk'}} N\nonumber\\
   &\hspace{1cm} \times(\delta_{kl}\delta_{k'l} (N+1)+ \delta_{kk'})+\delta_{kk'}\nue\Big),
\end{align}
where
\vspace{-0.5em}
\begin{subequations}
\begin{align}
    \!\!\acute{\varrho}_{s,kl} 
    &\!\triangleq \bar{\qh}^{\rm{T}}_{sk,\rm{f}}\bar{\qh}^*_{sl,\rm{f}} \!=\!\!\!\! \sum_{x,y \in \mathcal{S}_{(s)}}\!\!\!\!\! e^{j(xy)\pi\sin(\phi_{sk})}  e^{j(xy)\pi\sin(\phi_{sl})},\!
    \\ 
    \!\!\!\doubleacute{\varrho}_{s,kl}
    & \triangleq \bar{\qh}^{\rm{T}}_{sk,\rm{f}}\bar{\qh}^*_{sl,\rm{n}} \!=\!\!\!\! \sum_{x,y \in \mathcal{S}_{(s)}}\!\!\!\!\! e^{\!-j k \lVert \qp_{k} - \qp_{s_{x,y}}\!  \rVert} e^{j(xy)\pi\sin(\phi_{sl})}. 
\end{align}
\end{subequations}
\end{proposition}
\begin{proof}
    See Appendix~\ref{app:asymptotic_DL_SE}.
\end{proof}
Substituting these closed-form expressions for different types of received signals in \eqref{eq:DL_SE}, we can determine the limiting value for the $\mathrm{SE}_l$. 

\vspace{-0.5em}
\subsection{Downlink HE}
The input energy incident on the reception antenna of an EH user $m$ has two main RF signal components: (i) the received energy, $\Xi^{\mathsf{EH}}_m (\OEH\!\!,\qa)$, due to the transmitted power intended for the EH users;  (ii) the received energy, $\Xi^{\mathsf{ID}}_m  (\OID\!\!,\qa)$, of the information signals intended for the ID users. This cumulative energy signal can be expressed as
\vspace{0.5em}
\begin{equation}\label{eq:harvested_power}
    \Xi_m (\OID\!\!,\OEH\!\!,\qa) = \Xi^{\mathsf{EH}}_m (\OEH\!\!,\qa) + \Xi^{\mathsf{ID}}_m  (\OID\!\!,\qa),
\end{equation}
where
\vspace{-0.5em}
\begin{align}
    \Xi^{\mathsf{EH}}_m (\OEH\!\!,\qa) =&\sum\nolimits^M_{m'=1} \sus \susp \as \asp  \nonumber\\
     &\times \sqrt{\Omega^{\mathsf{EH}}_{sm'}\Omega^{\mathsf{EH}}_{s',m'}} \Upsilon_{s,s'\!,m, m'},\nonumber
\end{align}
and
\begin{align}
    \Xi^{\mathsf{ID}}_m  (\OID\!\!,\qa) \!=&\suml  \sus \susp \as \asp  \sqrt{\Omega^{\mathsf{ID}}_{sl}\Omega^{\mathsf{ID}}_{s',l}} \Upsilon_{s,s'\!,m, l},\nonumber
\end{align}
where $\Upsilon_{s,s'\!,m, m'}= {\qw}^{\rm{T}}_{s',m'} \qg^*_{s',m}  \qg^{\rm{T}}_{sm} {\qw}^*_{sm'} $ and $\Upsilon_{s,s'\!,m, l}={\qw}^{\rm{T}}_{s',l} \qg^*_{s',m} \qg^{\rm{T}}_{sm} {\qv}^*_{sl}$. Moreover, we have assumed that the AWGN has negligible contribution in the EH process \cite{Zhang:JSAC:2024}. 

For practical realization of wireless power transfer (WPT) of the considered XL-MIMO system, a non-linear EH (NL-EH) model is adopted \cite{Boshkovska,Nezhadmohammad}, which can be given as
\vspace{0.2em}
\begin{align}
    \Xi^{\mathsf{NL}}_m (\OID\!\!,\OEH\!\!,\qa)= \dfrac{\psi_m(\Xi_m (\OID\!\!,\OEH\!\!,\qa)) - \zeta_{max}\varphi}{1-\varphi},
\label{eq:non_linear_harvested_power}
\end{align}
where $\psi_m= \zeta_{max}/(1+ e^{-a [\Xi_m (\OID\!\!,\OEH\!\!,\qa)-b]})$ denotes the traditional logistic function of the received EH signal, while $a, b$ are the rectifier circuit parameters and $\zeta_{max}$ is the saturation DC power level of the EH receiver. Furthermore, $\varphi=1/{(1+e^{ab})}$ is a constant to ensure a zero input/zero-output response.

\vspace{-0.5em}
\subsection{Asymptotic Analysis for Downlink HE}
The asymptotic analysis of the NL-HE is primarily based on the convergence behavior of the two components of linear DL HE in \eqref{eq:harvested_power}. Similar to the approach presented in Section \ref{sec:asymptotic_DL}, we apply Tchebyshev's theorem on the constituent terms in \eqref{eq:harvested_power} to obtain
\vspace{0.5em}
\begin{subequations}
    \begin{align}
        &\!\!\!\dfrac{1}{S^2}\sum\nolimits^M_{m'=1} \sus \susp \as \asp \sqrt{\Omega^{\mathsf{EH}}_{sm'}\Omega^{\mathsf{EH}}_{s'm'}} \Upsilon_{s,s'\!,m, m'}\nonumber\\
        &\hspace{0em}\xrightarrow[S \rightarrow \infty]{P} \dfrac{1}{S^2} \sum\nolimits^M_{m'=1} \sus \susp \dfrac{P_s}{K} \ME\{\Upsilon_{s,s'\!,m, m'}\},\label{eq:asymptotic_intended_EH}\\
        &\!\!\!\dfrac{1}{S^2} \suml  \sus \susp \as \asp \sqrt{\Omega^{\mathsf{ID}}_{sl}\Omega^{\mathsf{ID}}_{s'l}} \Upsilon_{s,s'\!,m, l}\nonumber\\
        &\hspace{0em}\xrightarrow[S \rightarrow \infty]{P} \dfrac{1}{S^2} \suml \sus \susp \dfrac{P_s}{K} \ME\{\Upsilon_{s,s'\!,m, l}\}.\label{eq:asymptotic_ID_EH}
    \end{align}\label{eq:asymptotic_DL_EH}
\end{subequations}
The expectations of the two HE terms in \eqref{eq:asymptotic_DL_EH} are evaluated using \textbf{Proposition \ref{prop:asymptotic_DL_EH}} for both the coherent and noncoherent cases, while accounting for the NF EH channels as well as the NF and FF ID precoding vectors. 
\begin{proposition}\label{prop:asymptotic_DL_EH}
The closed-form expressions for the expectation terms of the received EH power  are presented in Table \ref{tab:asymptotic_DL_EH}, where $\rsmmp\triangleq \bar{\qh}^{\rm{T}}_{sm}\bar{\qh}^*_{sm'} $, is given by
\vspace{0.5em}
\begin{equation}
    \!\!\rsmmp=\! \sum\nolimits_{x,y \in \mathcal{S}_{(s)}} e^{-j k \lVert \qp_{m} - \qp_{s_{x,y}}  \rVert}e^{j k \lVert \qp_{m'} - \qp_{s_{x,y}}  \rVert} .  \! 
\end{equation}
For the case $m'=m$, $\rsmm=N$ and $\delta^{EH}_{mm'}=1$, otherwise $\delta^{EH}_{mm'}=0$. Moreover, $\acute{\vartheta}^{kk'}_{ss',l}=\ME\{\varpi_{s',lk} \varpi^*_{s,lk'}\}$.
\end{proposition}
\begin{proof}
    See Appendix~\ref{app:asymptotic_DL_EH}.
\end{proof}
\begin{table*}
    \centering
    \caption{{Asymptotic expressions for DL HE}}
    \vspace{-0.5em}
    \begin{tabular}{|p{2cm}|p{3.6cm}|p{5cm}|p{5.7cm}|}
    \hline
    \centering \textbf{Signal Type} & \centering \textbf{NF EH user }& \centering \textbf{NF ID precoding} & \centering \textbf{FF ID precoding}  
    \cr
    \hline
    \centering Coherent signal \eqref{eq:asymptotic_intended_EH} & $ \kappa^2_{sm'} \gamma^2_{sm}\big(\gamma^2_{sm}\delta^{EH}_{mm'} N^2
   \!+\! \gamma^2_{sm'}(1-\delta^{EH}_{mm'})\rsmmpsq \!+\!\nue \big) $ & $\kappa^2_{sl} \gamma^2_{sm} \sum^K_{k=1} \sum^K_{k'=1} \vartheta_{sl,kk'} \!\big(\gamma_{sk}\gamma_{sk'}\! \times \varrho_{s,kl} (\delta_{kk'}\varrho_{s,kl}\!\!+\!\!\varrho_{s,k'l})\!+\!\delta_{kk'}\nue\big) $ & $\kappa^2_{sl} \gamma^2_{sm} \sum^K_{k=1} \sum^K_{k'=1} \vartheta_{sl,kk'}\big(\sqrt{\varsigma_{sk}\varsigma_{sk'}}\varrho_{s,kl} \times(\delta_{kk'}\varrho_{s,kl}+\varrho_{s,k'l})+\beta_{sk}\delta_{kk'}N+\delta_{kk'}\nue\big)$ 
    \cr
    \hline

    \centering Non-Coherent signal \eqref{eq:asymptotic_ID_EH} & $\kmrtsmp \kmrtmp \gamma_{s'm} \gamma_{sm}\gamma_{s'm'}\times\gamma_{sm'}\rspmmp \rsmmp$& $\kappa_{s'l}\kappa_{sl} \gamma_{s'm}\gamma_{sm} \sum^K_{k=1} \sum^K_{k'=1} \!\acute{\vartheta}_{ss',l}^{kk'} \times\gamma_{s'k}\gamma_{sk'} \varrho_{s',km} \varrho_{s,k'm} $ & $\kappa_{s'l}\kappa_{sl}\gamma_{s'm}\gamma_{sm}\sum^K_{k=1} \sum^K_{k'=1} \!\acute{\vartheta}_{ss',l}^{kk'}  \times\sqrt{\varsigma_{s'k,\rm{f}} \varsigma_{sk,\rm{f}}}\varrho_{s',km} \varrho_{s,k'm}
    $ 
    \cr
    \hline
    \end{tabular}\label{tab:asymptotic_DL_EH}
    \vspace{-1.5em}
\end{table*}
By applying these derived relationships to  \eqref{eq:harvested_power}, we can compute the linear asymptotic DL EH limit, which can then be used to evaluate the NL-HE value in \eqref{eq:non_linear_harvested_power}. 

\vspace{-0.5em}
\section{Joint Optimization Process}

In this section, we formulate an optimization problem aimed at minimizing the overall PC of the proposed XL-MIMO array in \eqref{eq:power_consumed}  and develop a tractable algorithm that strategically assigns power coefficients and controls SA, while ensuring that each ID user meets the threshold DL SE and each EH user achieves the required HE. The optimization problem can be mathematically constructed as
\vspace{0.5em}
\begin{subequations}\label{eq:optimization_1}
    \begin{align}
        \mathcal{P}_1:&\min_{\{\OID\!,\OEH\!,\qa\}}\!   \,\, P_C(\OID\!\!,\OEH\!\!,\qa)\\
        \!\!\! \! \text{s.t} ~ \,\, & \,\,  \mathrm{SE}_l(\OID\!\!,\OEH\!\!,\qa) \!\geq\! \mathrm{SE}^{\mathsf{EA}}_{th},~\forall l=1,\ldots,L, \label{eq:optimization_1_C1} \\
        & \,\, \Xi^\mathsf{NL}_m(\OID\!\!,\OEH\!\!,\qa) \!\geq\! \Xi^{\mathsf{EA}}_{th},~\forall m=1,\ldots,M,\label{eq:optimization_1_C2}\\
        & \,\, \sumln \osln + \sumlf \oslf + \summ  \osm  \leq P_{s},\label{eq:optimization_1_C4}\\
        & \,\,  0 \leq \osln,\,\, 0 \leq \oslf,\,\, 0 \leq \osm, \quad \as \in \{0,1\} \label{eq:optimization_1_C5},
    \end{align}    
\end{subequations}

\hspace{-1em}where $\mathrm{SE}^{\mathsf{EA}}_{th}=\mathrm{SE}_l(\Omega^{\mathsf{ID-EA}}\!\!,\Omega^{\mathsf{EH-EA}}\!\!,\boldsymbol{1}_S)$ and $\Xi^{\mathsf{EA}}_{th} \triangleq  \Xi^\mathsf{NL}_m(\Omega^{\mathsf{ID-EA}}\!\!,\Omega^{\mathsf{EH-EA}}\!\!,\boldsymbol{1}_S)$ are the QoS requirements for the ID user $l$ and EH user $m$, respectively. For the purpose of fair comparison, the choice of these QoS thresholds for the optimization criterion is based on the equal PA and full array activation (\textbf{EA-FA}) case. This relates to the fact that we minimize the system PC subject to the constraints \eqref{eq:optimization_1_C1} and \eqref{eq:optimization_1_C2} which bound the minimum DL SE for each user $l$, and the DL HE for each user $m$, to the case of equal PA, denoted as $\Omega^{\mathsf{EA}}_{sl}$ and $\Omega^{\mathsf{EA}}_{sm}$, with full array activation $\forall \as = 1$. Moreover, the constraint \eqref{eq:optimization_1_C4} sets an upper limit on the transmit power at each individual subarray.

The optimization problem $\mathcal{P}_1$  is a mixed-integer non-convex problem, containing two second-order cone (SOC) constraints, namely \eqref{eq:optimization_1_C1} and \eqref{eq:optimization_1_C2}. The conic transformation of this problem $\mathcal{P}_1$ can lead to:
\begin{subequations}\label{eq:optimization_2}
    \begin{align}
        \mathcal{P}_2:&\,\, \min_{\{\OID,\OEH,\qa\}}   \, P_C(\OID\!\!,\OEH\!\!,\qa) \label{eq:objective_function_2}\\
        \text{s.t} ~ \,\, & \,\,\lVert \PIDNC \rVert \!\leq\! \hbar(\mathrm{SE}^{\mathsf{EA}}_l)\lVert \PIDC\rVert,\!\! \label{eq:optimization_2_C1} \\
        &\,\,\lVert \Xi_m(\OID\!\!,\OEH\!\!,\qa)\rVert \geq\Lambda (\Xi^{\mathsf{NL\!-\!EA}}_m),\label{eq:optimization_2_C2}\\
        & \,\, \sumln \osln + \sumlf \oslf + \summ  \osm  \leq P_{s},\label{eq:optimization_2_C4}\\
         & \,\,  0 \leq \osln,\,\, 0 \leq \oslf,\,\, 0 \leq \osm, \quad \as \in \{0,1\},\label{eq:optimization_2_C5}
    \end{align}    
\end{subequations}
where $\hbar\big(\mathrm{SE}^{\mathsf{EA}}_l\big)=1/\Big(2^{\mathrm{SE}^{\mathsf{EA}}_{th}}-1\Big)^{\! 1/2}$, and
\vspace{-0.8em}
\begin{align*}
    \Lambda (\Xi^{\mathsf{NL\!-\! EA}}_m)\!\!=\!\! \Bigg(\dfrac{1}{a}\log\Big[1-\dfrac{\zeta_{max}}{(1-\varphi)\Xi^{\mathsf{EA}}_{th}+\zeta_{max}\varphi}\Big]+b\Bigg)^{\! 1/2}.
\end{align*}
This mixed-integer optimization problem $\mathcal{P}_2$ depends on the set of variables $\{\OID,\OEH,\qa\}$  which include continuous PA and binary SA variables. By employing a parametric characterization based on the PA variables, the binary SA variables can be decoupled, enabling the design of a two-layer iterative optimization framework. In this hierarchical structure, the PA sub-problem is solved first for a certain fixed SA choice, followed by the SA sub-problem, which is addressed using a weighted function of the updated PA variables.


\vspace{-0.5em}
\subsection{Power Allocation Optimization (PA Routine)}\label{sec:PA_routine}
The optimization problem $\mathcal{P}_2$ can be restructured to deduce the PA sub-problem by using a certain SA choice $\tilde{\qa}$, which reduces the complexity of the original problem with an affine objective function $P_C(\OID\!\!,\OEH\!\!,\tilde{\qa})$ subject to the SOC constraints in \eqref{eq:optimization_2_C1} and \eqref{eq:optimization_2_C2}. We devise an optimization algorithm to obtain primal-dual feasible solutions using DR splitting-based ADMM. 
The optimization problems of this form can be solved by using the following mathematical operator \cite{pontus}:
\vspace{-0.8em}
\begin{equation}
    \min_{\bm{x}} \big\{f(\bm{x})+g(\bm{x})\big\},\nonumber
\end{equation}
where $\bm{x}$ is the set of optimization variables, while $f(\cdot)$ and $g(\cdot)$ are proper closed and convex functions. In analogy to this analytical structure, the optimization variable set can be represented as $\bm{x}\! \rightarrow\! \{\OID\!\!,\OEH\}$, while $f(\cdot)$ relates to the affine objective function in \eqref{eq:objective_function_2} and $g(\cdot)$ corresponds to the $\mathcal{P}_2$ constraints. For an auxiliary variable $\bm{z}$ with initial value $z^0 \in \mathbb{R}$ during a particular iteration ($\mathsf{u}$) of the PA procedure, the following primal-dual iterative ADMM updates are determined for the special case $\alpha=1/2$ ~\cite[Eq. (25)-(28)]{pontus}:
    \begin{align}\label{eq:ADMM_updates}
        \bm{x}^{{\mathsf{(u+1)}}}\!\!&=\!\argmin_{\bm{x}}  \!\Big\{\! f(\bm{x}^{{\mathsf{(u)}}})\!\!+\! 2 \tau \big \lVert A (\bm{x}^{{\mathsf{(u)}}})\!\!+\!\! B(\bm{y}^{{\mathsf{(u)}}}) \!\!-\!\! \bm{c} \!+\! \bm{z}^{{\mathsf{(u)}}}\big \rVert^2_2\!\Big\},
        \nonumber\\
        \bm{x}_A^{{\mathsf{(u+1)}}}\!\!&=\! 2 \alpha A (\bm{x}^{{\mathsf{(u+1)}}}) -(1-2 \alpha)(B(\bm{y}^{{\mathsf{(u)}}}) -\bm{c}),
        \nonumber\\
        \bm{y}^{{\mathsf{(u+1)}}}\!\!&=\!\argmin_{\bm{y}} \! \Big\{\! g(\bm{y}^{{\mathsf{(u)}}}) \!+\! \dfrac{\tau}{2} \big\lVert \bm{x}_A^{{\mathsf{(u+1)}}} \!+\! B(\bm{y}^{{\mathsf{(u)}}}) \!-\! \bm{c} \!+\! \bm{z}^{{\mathsf{(u)}}} \big\rVert^2_2 \Big\}, 
        \nonumber\\
        \bm{z}^{{\mathsf{(u+1)}}}\!\!&= \!\bm{z}^{{\mathsf{(u)}}} + \big(\bm{x}_A^{{\mathsf{(u+1)}}} + B(\bm{y}^{{\mathsf{(u+1)}}})-\bm{c} \big), 
    \end{align}
where
\vspace{-0.5em}
\begin{align*}
     \!\!A(\bm{x})\!\!=\!\!\! \left[\!\!\!\! 
        \begin{array}{c}
        \,\,\lVert  \Psi^{\mathsf{\, nc}}_{\ID,l}(\OID\!\!,\OEH\!\!,\tilde{\qa}) \rVert \!-\! \hbar(\mathrm{SE}^{\mathsf{EA}}_l)\lVert \Psi^{\mathsf{\, c}}_{\ID,l}(\OID\!\!,\tilde{\qa})\rVert\\
        \Lambda (\Xi^{\mathsf{NL\!-\!EA}}_m)-\lVert \Xi_m(\OID\!\!,\OEH\!\!,\tilde{\qa})\rVert\\
        \sumln \osln + \sumlf \oslf + \summ  \osm  - P_{s}
        \end{array} 
        \!\!\!\!\right],
\end{align*}
while $\bm{c}\!=\!\left[\hbar(\mathrm{SE}^{\mathsf{EA}}_l), \Lambda (\Xi^{\mathsf{NL\!-\!EA}}_m), P_s\right]^{\rm{T}}$, and $B(\bm{y})=\Pi_{\mathcal{K}_{\text{SOC}}}(\bm{y})$ characterizes a set of slack variables $\bm{y}\triangleq \bm{c}-A(\bm{x})$ which is introduced to relax the feasibility constraints using a projection operator. This operator $\Pi_{\mathcal{K}_{\text{SOC}}}(\bm{q})$, where $\bm{q} = (q_0,q_1) \in \mathbb{R} \times \mathbb{R}^n$ is defined as the closest point on the SOC $\mathcal{K}_{\text{SOC}} = \left\{ (r, s) \in \mathbb{R} \times \mathbb{R}^n \mid \| s \|_2 \leq r \right\} \forall \, r \in \mathbb{R}, s \in \mathbb{R}^n$ to the point $q$, can be mathematically expressed as
\begin{equation}
    \Pi_{\mathcal{K}_{\text{SOC}}}(\bm{q}) = \argmin_{(r, s) \in \mathcal{K}_{\text{SOC}}} \| (r, s) - (q_0, q_1) \|_2.
\end{equation}

During each iteration $(\mathsf{u})$, the PA-algorithm checks for the convergence condition: $$\lvert P^{{\mathsf{(u)}}}_C(\tOmgIDt\!,\tOmgEHt\!,\tilde{\qa}^{\mathsf{(t)}})\!-\! P^{\mathsf{(u-1)}}_C(\tOmgIDt\!,\tOmgEHt\!,\tilde{\qa}^{\mathsf{(t)}}) \rvert \!\leq\! \epsilon,$$ for $\epsilon >0$. If this inequality condition is satisfied, the PA routine is terminated with the output PA estimates ($\tOmgID$\!,\! $\tOmgEH$) for the particular SA choice $\tilde{\qa}$.


\vspace{-0.5em}
\subsection{Subarray Activation Optimization (SA Routine)}\label{sec:SA_routine}
Next, the SA optimization routine is devised to determine the binary SA variables, $\qa$, by parameterizing the PA updates estimated in the PA routine delineated in Section \ref{sec:PA_routine}. For a particular iteration $\mathsf{(t)}$ of the SA routine, we first calculate the PA updates ($\tOmgID$\!,\! $\tOmgEH$) for a certain SA choice ($\tilde{\qa}$), using the FA case $\qa =\boldsymbol{1}_S$ for initial realization. Now, we introduce a joint surrogate auxiliary function $h^{\mathsf{t}}_s(\tOmgID\!\!,\tOmgEH)$, which encompasses the relative contribution of each active subarray in terms of PA updates $(\tOmgID\!\!,\tOmgEH)$ for supporting both DL ID and EH functionalities. We define this normalized function as the ratio of individual subarray utility to the overall system utility, as follows
\vspace{0.5em}
\begin{align}~\label{eq:relative_contribution}
     &\!\!\! h^{\mathsf{t}}_s (\tOmgIDt\!,\tOmgEHt)= \dfrac{H^{\mathsf{t}}_{s}(\tOmgIDt\!,\tOmgEHt)}{\sus H^{\mathsf{t}}_{s}(\tOmgIDt\!,\tOmgEHt)},
\end{align}
where
\vspace{0.5em}
\begin{align}\label{eq:subarray_relative_contribution}
    H^{\mathsf{t}}_{s}(\tOmgIDt\!,\tOmgEHt)=&\varrho_{\rm{n}}^{\mathsf{(t)}} \sumln \tilde{\Omega}^{\mathsf{ID,\mathsf{(t)}}}_{sl,\rm{n}} + \rho_{\rm{f}}^{\mathsf{(t)}}\sumlf \tilde{\Omega}^{\mathsf{ID,\mathsf{(t)}}}_{sl,\rm{f}}\nonumber\\
    &+ \summ \tilde{\Omega}^{\mathsf{EH,\mathsf{(t)}}}_{sm},
\end{align}
while $\varrho_{\rm{n}}^{\mathsf{(t)}}=\sus \summ \tilde{\Omega}^{\mathsf{EH,\mathsf{(t)}}}_{sm}/\sus \sumln \tilde{\Omega}^{\mathsf{ID,\mathsf{(t)}}}_{sl,\rm{n}}$ and $\varrho_{\rm{f}}^{\mathsf{(t)}}=\sus \summ \tilde{\Omega}^{\mathsf{EH,\mathsf{(t)}}}_{sm}/\sus \sumlf \tilde{\Omega}^{\mathsf{ID,\mathsf{(t)}}}_{sl,\rm{f}}$ are the balancing factors that regulate the distribution of power between satisfying the ID QoS requirements of both NF and FF users and supporting the power allocated for EH operations. The inclusion of both these factors allows us to evaluate the cumulative performance of each subarray in SWIPT operations over two distinct functional domains (ID and EH), which operate on entirely different dynamic power ranges, while taking into account the NF and FF ID power difference. If we analyze \eqref{eq:relative_contribution} carefully, we can infer that the function $h^{\mathsf{t}}_s(\tOmgID\!,\tOmgEH)$ returns a higher value for the subarray which has a stronger contribution to the overall SWIPT performance in comparison to the other subarrays, and vice versa.

Proceeding further, we establish the association of relative contribution-based auxiliary surrogate function $h^{\mathsf{t}}_s(\tOmgID\!,\tOmgEH)$ to the binary decision SA variables $\qa$. During the iteration ($t$), we activate the subarrays which have higher $h^{\mathsf{t}}_s(\tOmgID\!,\tOmgEH)$ values than the mean threshold $\ME\{h^{\mathsf{t}}_s(\tOmgID\!,\tOmgEH)\}$ and turn-off the other subarrays. Although the proposed PA algorithm already allocates smaller amount of power to the less functional subarrays due to the network topology, this planned deactivation seeks to control the PC associated to the circuitry components of these subarrays. This iterative process provides the set of SA updates $\tilde{\qa}$ with individual elements represented as:
\vspace{0.2em}
\begin{align}
\label{eq:binary_decision}
    \!\!\! a^{\mathsf{t}}_s \!&=\! 
        \begin{cases}
            1 \quad \,\, h^{\mathsf{t}}_s(\tOmgIDt\!\!,\tOmgEHt)\! \geq \! \ME\{h^{\mathsf{t}}_s(\tOmgIDt\!\!,\tOmgEHt)\},\!\!\\
            0 \quad  \,\, h^{\mathsf{t}}_s(\tOmgIDt\!\!,\tOmgEHt)\! < \!\ME\{h^{\mathsf{t}}_s(\tOmgIDt\!\!,\tOmgEHt)\},\!\!
        \end{cases}\raisetag{20pt}
\end{align}
where the mean threshold (i.e., $\ME\{h^{\mathsf{t}}_s(\tOmgID\!,\tOmgEH)\} =1/S$) ensures that only those subarrays are activated which play a dominant role in the SWIPT performance. For the next PA iteration ($\mathsf{u}+1$), these SA updates are scaled once again with the same joint surrogate function $h^{\mathsf{t}}_s(\tOmgID\!,\tOmgEH)$. Now, the continuous SA estimates are
\begin{equation}\label{eq:parameterized}
     \tilde{a}^{\mathsf{t}}_s = h^{\mathsf{t}}_s(\tOmgIDt\!,\tOmgEHt) a^{\mathsf{t}}_s.
\end{equation}
This scaling operation preserves the relative magnitude of SWIPT contribution within the subset of activated subarrays. These parametrized SA estimates $\tilde{\qa}^{\mathsf{(t)}}$ are substituted in \eqref{eq:ADMM_updates} to trigger the next iterative cycle of PA primal-dual estimates.

\begin{algorithm}[t]
\caption{Joint Optimization of Subarray Activation ($\qa$) and SWIPT Power Allocation ($\OID \!,\OEH$)}\label{alg:opt_process}
\begin{algorithmic}[1]

\State \textbf{Initialize:} Set iteration indices $\mathsf{t} = 0$ and $\mathsf{u} = 0$, convergence thresholds $\delta\!>\!0$ and $\epsilon\!>\!0$.
\State Initial estimates: $\tilde{\Omega}^{(0)}_{sl}=\Omega^{\mathsf{ID-EA}}_{sl}\!,\tilde{\Omega}^{(0)}_{sm}= \Omega^{\mathsf{EH-EA}}_{sm}$\!, $\tilde{a}^{0}_s=1$
\State \textbf{SA Routine:}
\Repeat

    \State Compute the surrogate auxiliary function
    using~\eqref{eq:relative_contribution}.
    
    \State Update the subarray activation variable $a^{\mathsf{t}}_s$ using \eqref{eq:binary_decision}.
    \State Scale the activation variable $\ats$ using \eqref{eq:parameterized}.
    \State Formulate the parameterized optimization problem:
    \begin{equation}\label{eq:para_opt}
        \min_{\{\OID,\OEH\}} \,\,P_C(\tOmgIDt\!,\tOmgEHt,\tilde{\qa}^{\mathsf{(t)}}) \nonumber
    \end{equation}
    \State\textbf{PA Routine:}
    \Repeat
        \State Form the conic representation using~\eqref{eq:optimization_2}.
        \State Solve the 
        parametrized optimization problem in Step 8 using DR splitting based ADMM in \eqref{eq:ADMM_updates} with $\mathsf{u}$ sub-iteration.
        \State Increment inner iteration index: $\mathsf{u} = \mathsf{u} + 1$.
        \Until{$\lvert P^{{\mathsf{(u)}}}_C(\tOmgIDt\!,\tOmgEHt\!,\tilde{\qa}^{\mathsf{(t)}})\!-\! P^{(u\!-\!1)}_C(\tOmgIDt\!,\tOmgEHt\!,\tilde{\qa}^{\mathsf{(t)}}) \rvert \!\leq\! \epsilon$}
    
    \State Increment outer iteration index: $t = t + 1$.
    \State Update PA estimates $\tilde{\boldsymbol{\Omega}}^{\mathsf{ID,t}}$\!,\! $\tOmgEHt$
    
\Until{$ \lvert P_C(\tOmgIDt\!\!,\tOmgEHt\!,\tilde{\qa}^{\mathsf{(t)}})\!-\! P_C(\tilde{\boldsymbol{\Omega}}^{\mathsf{ID, (t-1)}}\!\!,\tilde{\boldsymbol{\Omega}}^{\mathsf{EH,(t-1)}}\!\!,\tilde{\qa}^{\mathsf{(t-1)}}) \rvert \!\leq\! \delta$}
\State Return optimized parameters as $\OmgIDst=\tilde{\boldsymbol{\Omega}}^{\mathsf{ID, \mathsf{(t)}}}$, $\OmgEHst=\tilde{\boldsymbol{\Omega}}^{\mathsf{EH,\mathsf{(t)}}}$ and $\qa^{\star}=\tilde{\qa}^{\mathsf{(t)}}$.
\end{algorithmic}
\end{algorithm}
\setlength{\textfloatsep}{0.05cm}


\begin{figure*}[t]
    \centering
    \begin{minipage}[t]{0.46\textwidth}
        \centering
        \includegraphics[trim=0.5cm 0cm 2cm 0cm,clip,width=0.7\textwidth]{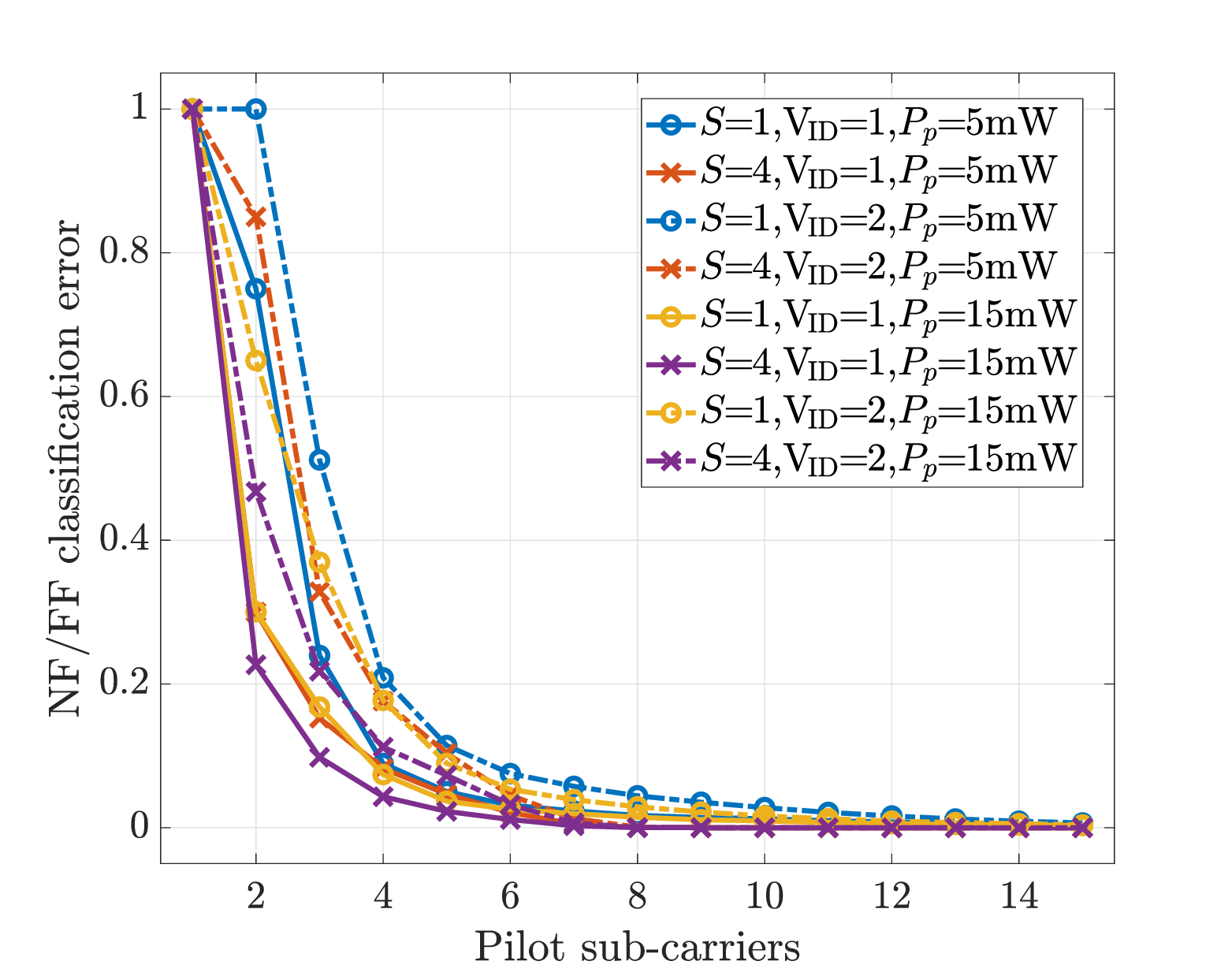}
        \vspace{-0.6em}
        \caption{\small  NF/FF classification error versus the number of pilot sub-carriers with orthogonal pilots.\normalsize}
        \label{fig:Near_far_classification_error}
    \end{minipage}
    \hfill
    \begin{minipage}[t]{0.46\textwidth}
        \centering
        \includegraphics[trim=0.5cm 0cm 2cm 0cm,clip,width=0.7\textwidth]{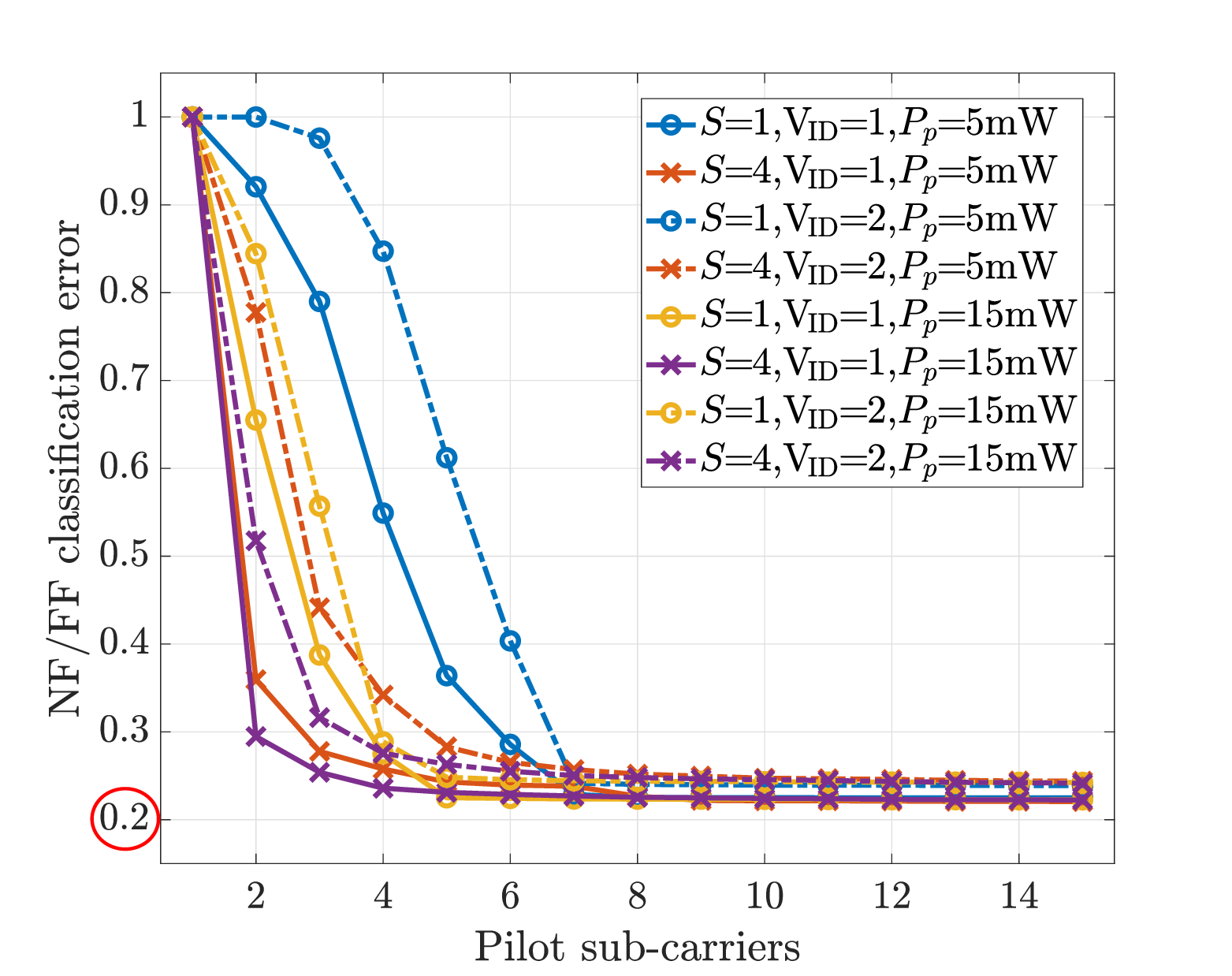}
         \vspace{-0.6em}
        \caption{\small NF/FF classification error versus the number of pilot sub-carriers under pilot contamination with 10 orthogonal pilot sequences for 20 users.\normalsize}
        \label{fig:Near_far_classification_error_10pilots}
    \end{minipage}
    \vspace{-1.7em}
\end{figure*}

\vspace{-0.7em}
\subsection{Overall Algorithm and Complexity Analysis} 
The PA and SA procedures are implemented alternatively, as delineated in \textbf{Algorithm \ref{alg:opt_process}}. For a particular set of PA coefficients, the SA vector $\qa$ is optimized. Based on the lesser number of active subarrays ($S_a\!=\!\sus a^{\mathsf{t}}_s$), the PA coefficients are re-optimized until the optimization convergence is achieved within a negligible tolerance $\delta >0$, i.e., $ \lvert P_C(\tilde{\boldsymbol{\Omega}}^{\mathsf{ID ,(t)}}\!\!,\tOmgEHt\!\!,\tilde{\qa}^{\mathsf{(t)}})- P_C(\tilde{\boldsymbol{\Omega}}^{\mathsf{ID, {\mathsf{(t-1)}}}}\!\!,\tilde{\boldsymbol{\Omega}}^{\mathsf{EH,{\mathsf{(t-1)}}}}\!\!,\tilde{\qa}^{\mathsf{(t-1)}}) \rvert \!\leq\! \delta$.

The PA routine is executed through proximal operators within each DR based ADMM iteration, leading to a computational complexity of $\mathcal{O}(SK)$ for each subarray-user pair \cite{pontus}. Moreover, the convergence rate for this sub-routine scales linearly with the problem size $\mathcal{O}(SK)$. On the other hand, the SA procedure calculates the surrogate function which involves computations over $SK$ PA coefficients. However, the binary SA assignment in \eqref{eq:binary_decision} with respect to $\ME\{h^{\mathsf{t}}_s(\tOmgIDt\!,\tOmgEHt)\}$ scales logarithmically as $\mathcal{O}(\log(SK))$. As a result, the overall computational complexity of the proposed joint PA-SA algorithm can be expressed as $\mathcal{O}(S^2K^2\log(SK))$. For very large modular XL-MIMO deployments, the proposed \textbf{PA-SA} framework remains scalable since SA reduces the PA optimization dimension from $S$ to the active set size $S_a=\sum_{s=1}^{S}a_s^{\mathsf{t}}$.

\vspace{-0.5em}
\section{Numerical Results}
\subsection{System Parameters}
In this section, we evaluate the performance of the considered SWIPT  XL-MIMO system using the proposed joint optimization algorithm across the four scenarios illustrated in Fig. \ref{fig:Scenario_XL_MIMO}. We consider that each individual UPSA is comprised of $\ns=256$ antennas arranged in $32 \times 8$ configuration. With the focus on PC reduction, the HB architecture of these UPSAs deploys $N_{RF}=16$ RF chains. Each antenna element has largest dimension $D=\lambda/4$ with inter-element spacing $d=\lambda/2$, where the carrier wavelength is $\lambda=0.1$ m \cite{Ramezani2024}. This XL-MIMO system serves a total $K=20$ number of users including $L=14$ ID and $M=6$ EH users. The NF users are confined to the limit $d_{sk}<d_{\rm{f}}/10$ in the $z$-axis \cite{Ramezani2024}, whereas the far field ID users are limited to the region $d_{\rm{f}}<d_{sk}<1.05 \,d_{\rm{f}}$. The coherence interval is assumed to have a duration of $T_c=20$ ms, corresponding to $\tau_c=6,250$ symbols. Within each coherence block, we acquire $Q=20$ independent channel observations using mutual orthogonal pilot sequences, each of length  $\tau_p=30$ symbols. We use subcarrier spacing of $\Delta f_u=60$ kHz, with uplink power $P_p=\{5, 15\}$mW, which is consistent with the sub-$6$GHz $5$G NR standard \cite{3gpp_38211,Yonghwi}. The parameters for the PC model given in \eqref{eq:power_consumed} are fixed as $\varsigma=0.35$, $P_{\mathrm{syn}}=50$ mW, $P_{\mathrm{ct}}=48.2$ mW, and $P_{\mathrm{et}}=15$ mW \cite{Jun_Zhang}. The noise power at the receiver of each ID user is considered to be $\sigma_l=-80$ dBm. The NL-EH parameters in \eqref{eq:non_linear_harvested_power} are selected as $\zeta_{max}=24$ mW and circuit parameters are $a=1500$, $b=0.0022$ \cite{Nezhadmohammad}. The convergence thresholds for PA and SA routines are set as $\epsilon=10^{-3}$ and $\delta=10^{-3}$, respectively. The choice of these values enables us to maintain the suitable precision level of $0.1\%$ improvement across the successive iterations.

\begin{figure*}[t]
    \centering
    \begin{minipage}[t]{0.32\textwidth}
        \centering
        \includegraphics[trim=0cm 0cm 2cm 0cm,clip,width=\textwidth]{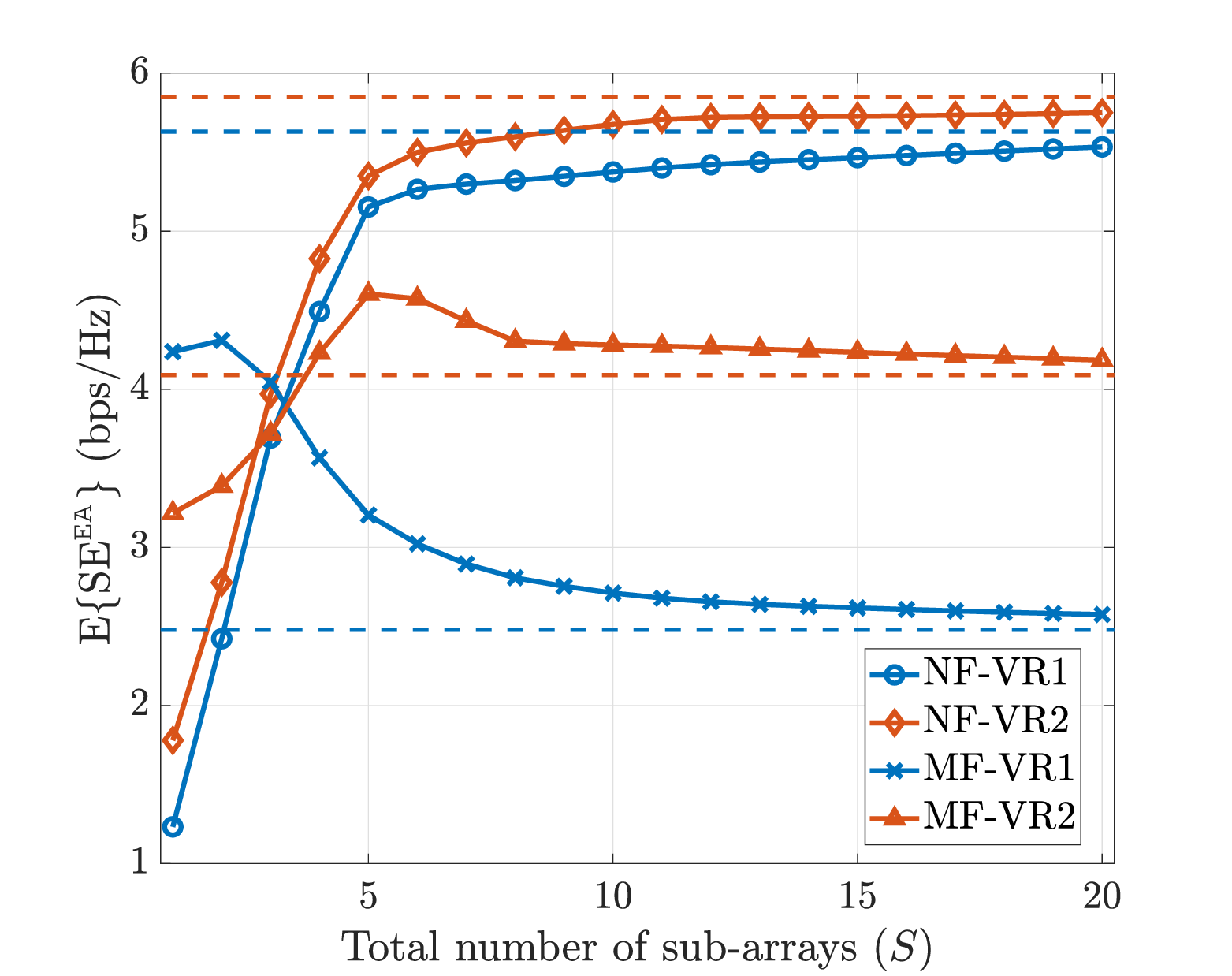}
    \vspace{-1.5em}
    \caption{\small DL SE versus the number of subarrays. The solid (dotted) lines show simulation (asymptotic) results.\normalsize}
    \label{fig:Asymptotic_DL}
    \end{minipage}
    \hfill
    \begin{minipage}[t]{0.32\textwidth}  
        \centering
        \includegraphics[trim=0cm 0cm 2cm 0cm,clip,width=\textwidth]{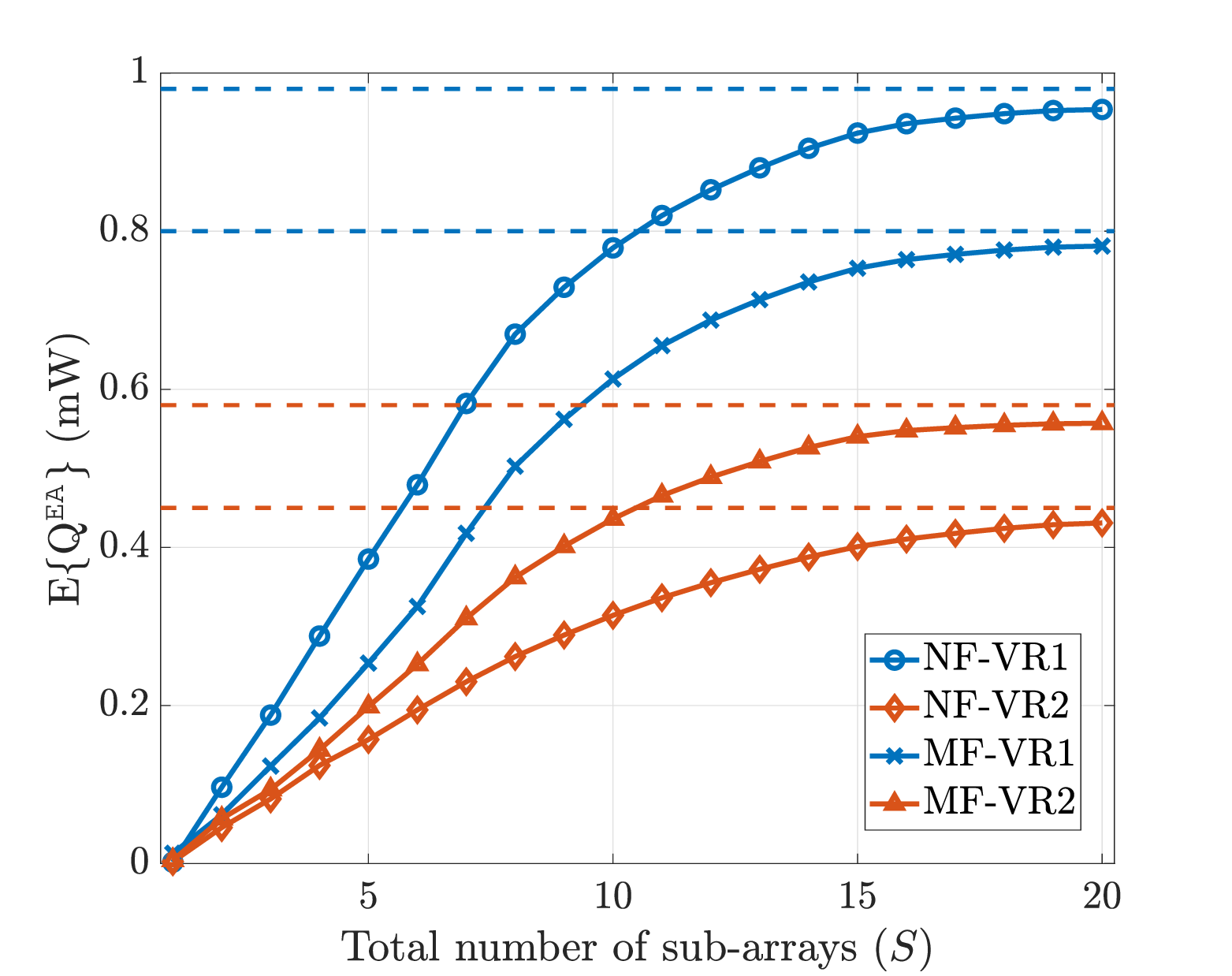}
    \vspace{-1.5em}
    \caption{\small DL HE versus the number of subarrays. The solid (dotted) lines show simulation (asymptotic) results.\normalsize}
    \label{fig:Asymptotic_EH}
    \end{minipage}
    \hfill
    \begin{minipage}[t]{0.32\textwidth}  
        \centering
        \includegraphics[trim=0cm 0cm 2cm 1.5cm,clip,width=\textwidth]{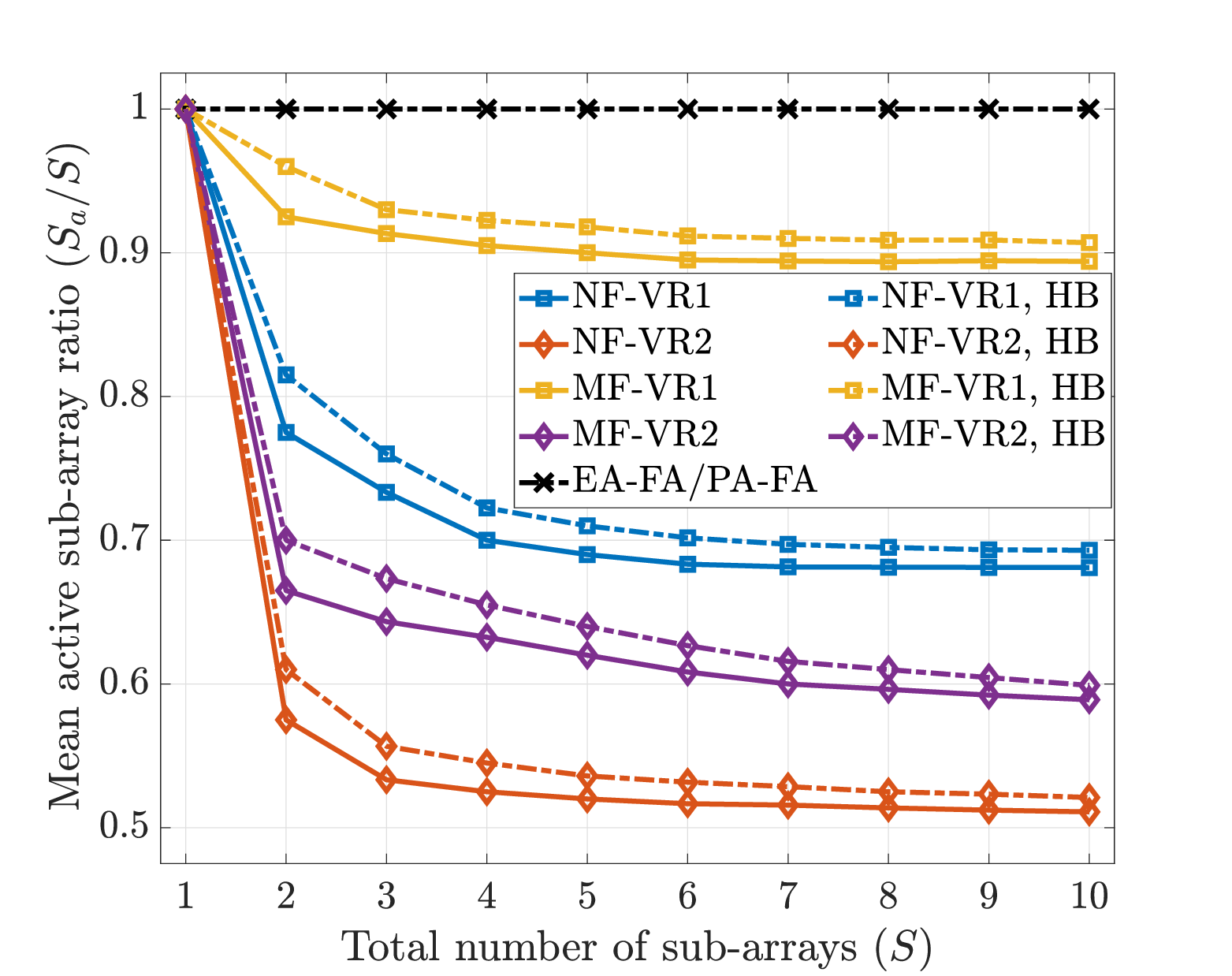}
    \vspace{-1.5em}
    \caption{\small Mean active subarray ratio versus the number of subarrays.\normalsize}
    \label{fig:subarray_active}
    \end{minipage}
\vspace{-1.3em}
\end{figure*}

The large-scale fading parameter can be modeled as $\zekl=10^{-(PL_{sk}+\Psi_{sk})/10}$, characterizing both the path loss, $PL_{sk}$, and shadowing effects $\Psi_{sk}$. The log-normal shadowing can be represented as $\Psi_{kl}=\sigma_{\Psi} \tilde{\Psi}_{kl}$, with standard deviation $\sigma_{\Psi}$ and $\tilde{\Psi}_{kl}\!\sim\!\mathcal{N}(0, 1)$. Moreover, $PL_{sk}$ can be defined using the three-slope model in \cite{Zeeshan}, given as
\vspace{-0.4em}
\begin{equation}\label{eq:PL}
\hspace{-1em}PL_{sk}\! =\! 
\begin{cases} 
   -L\!-\!35 \,{\log}_{10}(d_{sk}) &  d_{sk}\!>\! d_1, \\
   -L\!-\!15 \,{\log}_{10}(d_{1})\!-\!20 \,{\log}_{10}(d_{sk}) &  d_0\!< \!d_{sk} \!< \!d_1, \\
   -L\!-\!15 \,{\log}_{10}(d_{1})\!-\!20 \,{\log}_{10}(d_{0}) &  d_{sk}\! <\! d_0, \\
\end{cases}
\end{equation}

\hspace{-1.3em}where the reference distances are $d_0=10$ m, $d_1=50$ m \cite{zhao}, while
\begin{align}
    L&\triangleq  \,46.3+33.9 \log_{10}(f)-13.82 \log_{10}(h_{AP})\nonumber\\
    &\,\,\,-\!(1.1 \log_{10} (f)\!-\!0.7) h_{UE} +(1.56 \log_{10}(f)\!-\!0.8),
\end{align}
where $f=c/\lambda$ is the carrier frequency, and $d_{sk}, d_0$ and $d_1$ are substituted in terms of kilometers.

\subsection{Results and Discussion}

\subsubsection{NF / FF User Classification}
We first discuss the results of the proposed decision making criterion for NF or FF classification on the basis of the stochastic correlation of the channel estimates based on the orthogonal pilot sequences. We assess the statistical variation of the communication channel, ${\qg}_{sk}$, between a user $k$ and subarray $s$ over multiple frequency sub-carriers, $U$, using pilot sequences to obtain independent observations. If we observe the channel for a large number of pilot sub-carriers, we can map the channel variation for both user types more accurately, as we have a larger dimension ($U \times U$) frequency correlation matrix $\qA_{{sk}}$ given in \eqref{eq:correlation}. This assumption is validated in Fig. \ref{fig:Near_far_classification_error}, which shows that the classification error in identifying the EM field region of a particular user decreases exponentially by increasing the number of pilot sub-carriers. 

This classification process is even better as we use a higher number of modular subarrays $S$ and more uplink pilot power $P_p$. For $P_p=15$mW, we can see that the classification error for the case $S=4$ reduces to $23\%$ even if we use only $U=2$ pilot sub-carriers. On the other hand, we cannot classify users at all for the case $S=1$ up to $U=2$ pilot sub-carriers for $P_p=5$mW. The classification error increases significantly for all $S$ configurations for a lower uplink power $P_p=5$mW. Note that the classification error almost vanishes if we use $U=7$ or more subcarriers for the $S=4$ case.  Moreover, the classification performance is marginally better for the co-located ID users in a single VR (solid lines) against the case of distribution across multiple VRs (dotted lines). 

For practical assessment of our proposed NF/FF classification method, we now evaluate the impact of pilot contamination on classification error. In this scenario of pilot re-use, the mutual orthogonality condition ascertained in Section \ref{sec:chan_est} is relaxed due to the limited number of available pilot sequences. Let $\mathcal{P}_{k}$ denote the set of users assigned to the same pilot sequence as user $k$, such that $\pmb{\varphi}^{\rm T}_{i}\pmb{\varphi}^{*}_{k}=\tau_p$ for $i\in\mathcal{P}_{k}$ and $\pmb{\varphi}^{\rm T}_{i}\pmb{\varphi}^{*}_{k}=0$ otherwise. After projecting the received pilot signal onto $\pmb{\varphi}_{k}$, the sufficient statistic in \eqref{eq:proj_pilot} becomes
\begin{equation}
    \tilde{\qy}^{(q)}_{sk}
    =
    \sqrt{\tau_p P_p}\,{\qW}^{(q)}_{\!\!A,s}
    \bigg(
    \gsk + \sum\nolimits_{i\in\mathcal{P}_{k}\setminus\{k\}} \qg_{si}
    \bigg)
    +\qn^{(q)}_s .
\end{equation}

Considering this pilot contamination, the LS channel estimate in \eqref{eq:CE} is now obtained as
\begin{equation}
    \hat{\qg}_{sk}^{\rm pc}
    =
    \gsk
    +
    \sum\nolimits_{i\in\mathcal{P}_{k}\setminus\{k\}} \qg_{si}
    +
    \tilde{\hat{\pmb{\varepsilon}}}_{sk},
    \label{eq:CE_PC}
\end{equation}
where the second term in \eqref{eq:CE_PC} represents the coherent pilot contamination induced by the users sharing the same pilot sequence. Now, we consider only $10$ orthogonal pilot sequences for $K=20$ users to acquire these contaminated channel estimates $\hat{\qg}_{sk}^{\rm pc}$. We can clearly observe a significant degradation in classification performance in Fig. \ref{fig:Near_far_classification_error_10pilots} in comparison to the orthogonal pilot case in Fig. \ref{fig:Near_far_classification_error}, especially for lower $P_p$ and multiple VRs. 

It can be noticed that the coherent pilot contamination term in \eqref{eq:CE_PC} impacts the classification error significantly if the same pilot is assigned to a NF and a FF user simultaneously. This pilot assignment combination makes classification decision-making almost impossible, since we cannot distinguish between these NF and FF users based upon their estimates. Thus, the steady-state classification errors are dependent on the possible combinations of the NF-FF contamination pairs, which can be expressed as $\frac{(L_{\rm{n}}+M)L_{\rm{f}}}{K(K-1)}$. For the MF-VR1 case in Fig. Fig.~\ref{fig:Scenario_XL_MIMO}(\subref{fig:MF_VR1}), we have $M=6$ NF and $L_{\rm{f}}=14$ FF users, whereas for the MF-VR2 case in Fig. Fig.~\ref{fig:Scenario_XL_MIMO}(\subref{fig:MF_VR2}), we have $L_{\rm{n}}+M=13$ NF and $L_{\rm{f}}=7$ FF users. The steady-state errors due to NF-FF contamination pairs for these cases are $0.221$ and $0.239$, respectively. To address this problem, region-aware pilot powers may be used with the contamination term being weighted as $\sum_{i\in\mathcal{P}_{k}\setminus\{k\}}\sqrt{P_i/P_k}\,\qg_{si}$. However, this relative pilot power assignment can be considered as a future research avenue.

\begin{figure*}[t]
    \centering
    \begin{subfigure}[t]{0.24\textwidth}
        \centering
        \includegraphics[trim=0cm 0cm 0cm 0cm,clip,width=1.07\textwidth]{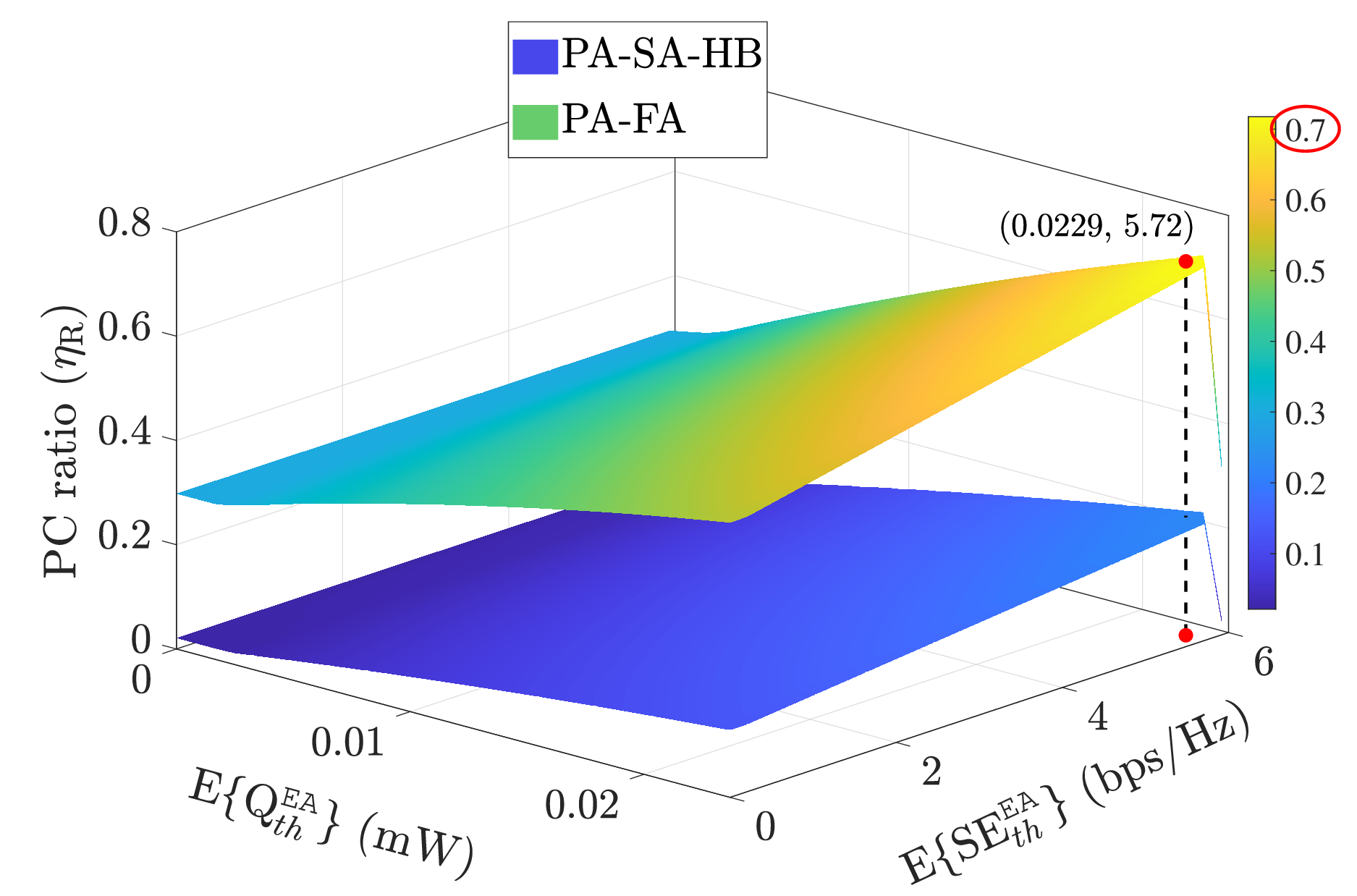}
        \caption{\small NF-VR1 case\normalsize}
        \label{fig:Thresholds_NF_VR1}
    \end{subfigure}
    \hfill
    \begin{subfigure}[t]{0.24\textwidth}
        \centering
        \includegraphics[trim=0cm 0cm 0cm 0cm,clip,width=1.07\textwidth]{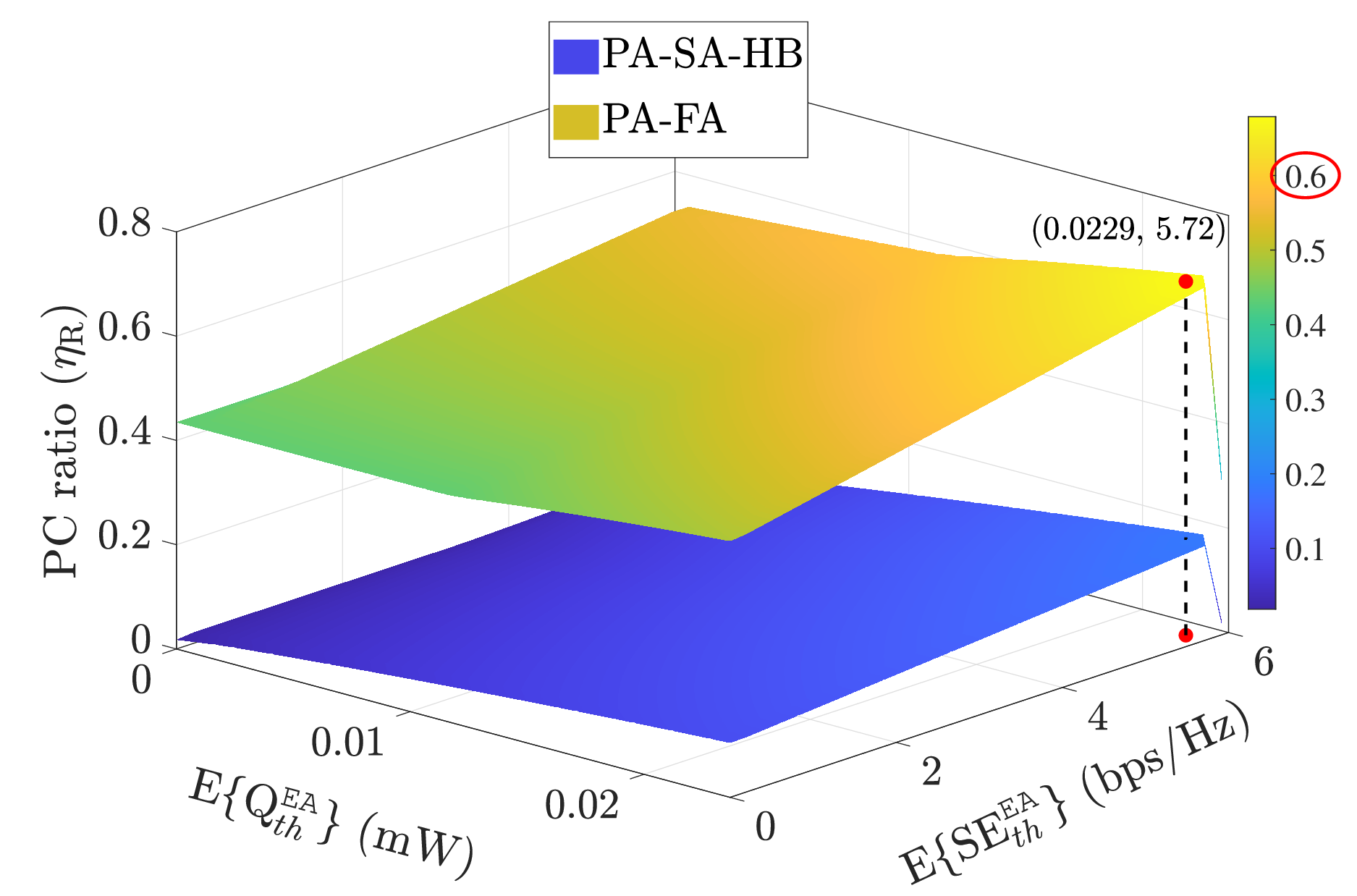}
        \caption{\small NF-VR2 case\normalsize}
        \label{fig:Thresholds_NF_VR2}
    \end{subfigure}
    \hfill
    \begin{subfigure}[t]{0.24\textwidth}
        \centering
        \includegraphics[trim=0cm 0cm 0cm 0cm,clip,width=1.07\textwidth]{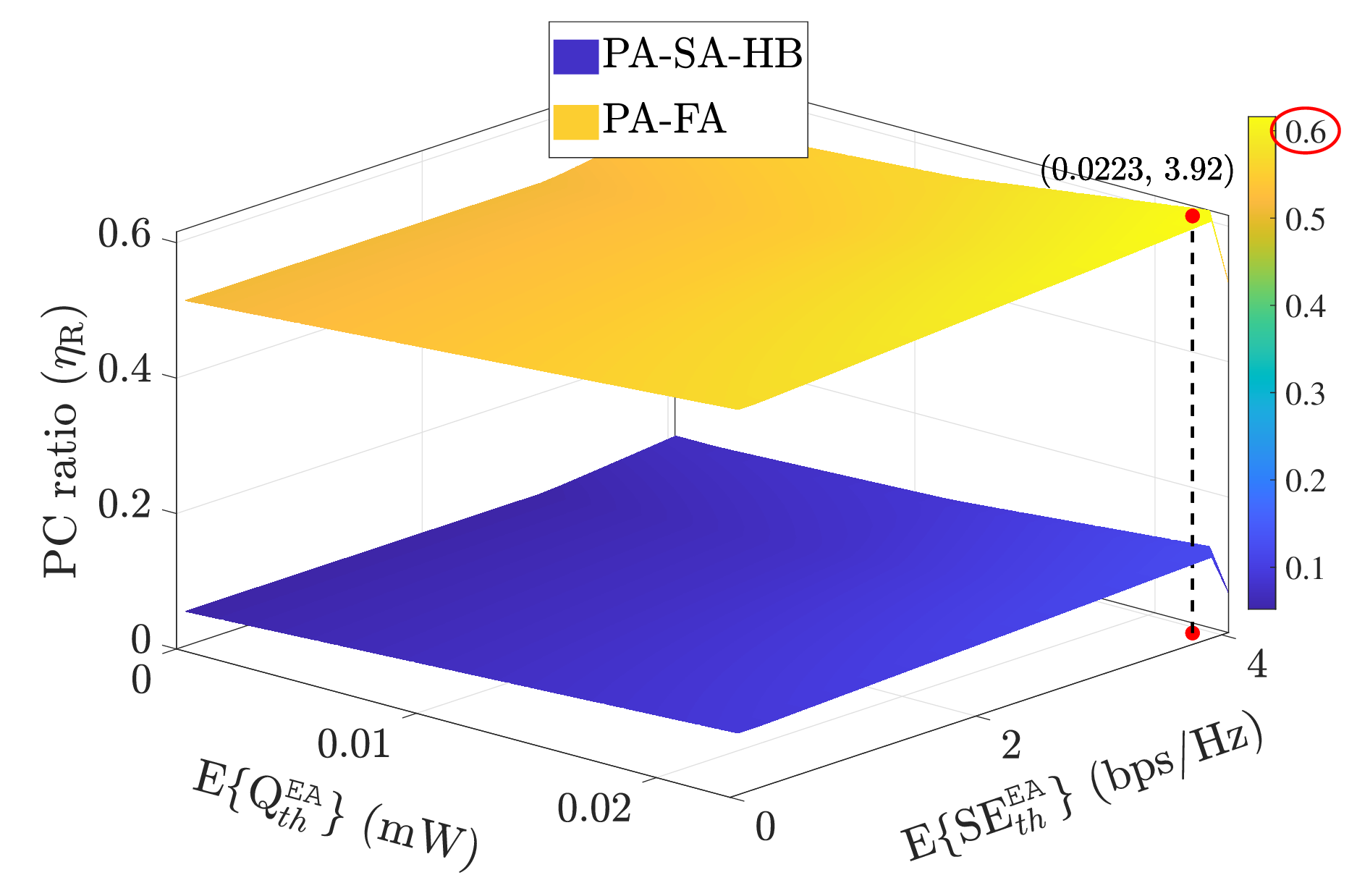}
        \caption{\small MF-VR1 case\normalsize}
        \label{fig:Thresholds_MF_VR1}
    \end{subfigure}
    \hfill
    \begin{subfigure}[t]{0.24\textwidth}
        \centering
        \includegraphics[trim=0cm 0cm 0cm 0cm,clip,width=1.07\textwidth]{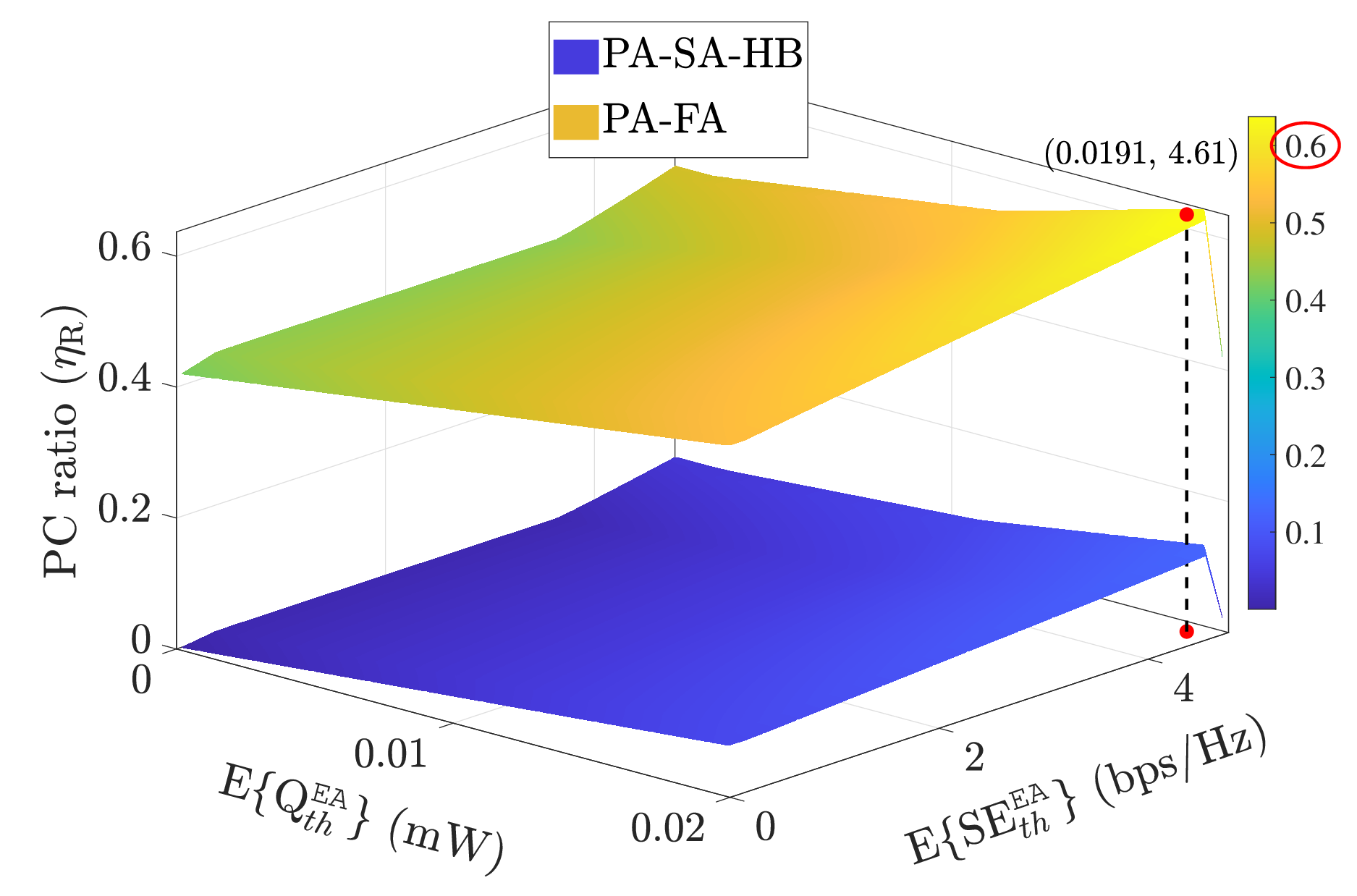}
        \caption{\small MF-VR2 case\normalsize}
        \label{fig:Thresholds_MF_VR2}
    \end{subfigure}
    \caption{\small PC ratio over DL SE $\mathrm{SE}^{\mathsf{EA}}_{th}$ and DL HE threshold $\Xi^{\mathsf{EA}}_{th}$ for the $S=4$ subarray configuration.\normalsize}\label{fig:Thresholds}
\vspace{-1em}
\end{figure*}

\begin{figure*}[t]
    \centering
    \begin{minipage}[t]{0.32\textwidth}
        \centering
        \includegraphics[trim=19 0cm 2cm 0cm,clip,width=\textwidth]{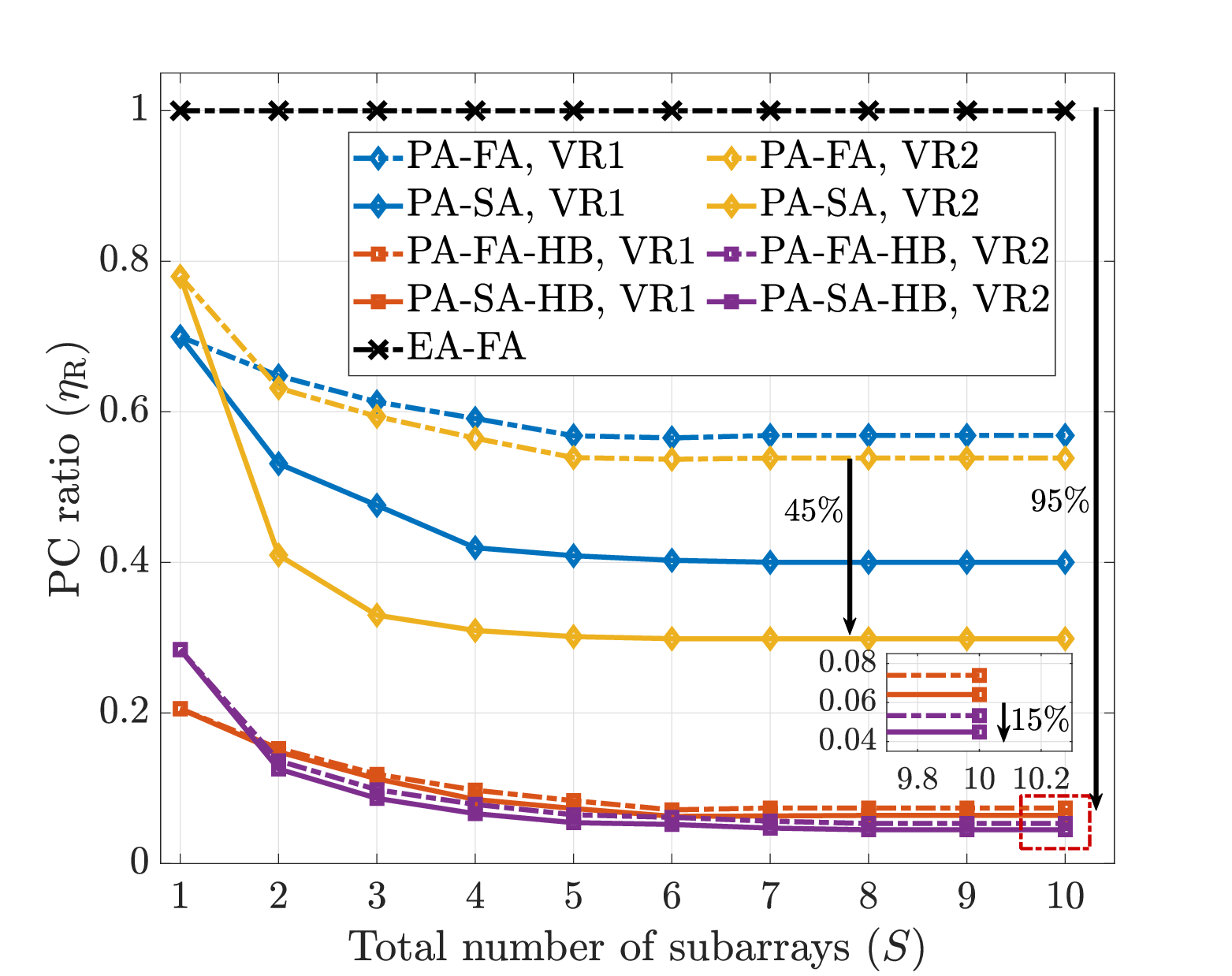}
    \vspace{-1.2em}
    \caption{\small PC ratio versus the number of subarrays for $\mathrm{NF}$-cases (Fig.~\ref{fig:Scenario_XL_MIMO}(\subref{fig:NF_VR1}) \& Fig.~\ref{fig:Scenario_XL_MIMO}(\subref{fig:NF_VR2})).\normalsize}
    \label{fig:NF_power_consumption_vs_subarrays}
    \end{minipage}
    \hfill
    \begin{minipage}[t]{0.32\textwidth}  
        \centering
        \includegraphics[trim=19 0cm 2cm 0cm,clip,width=\textwidth]{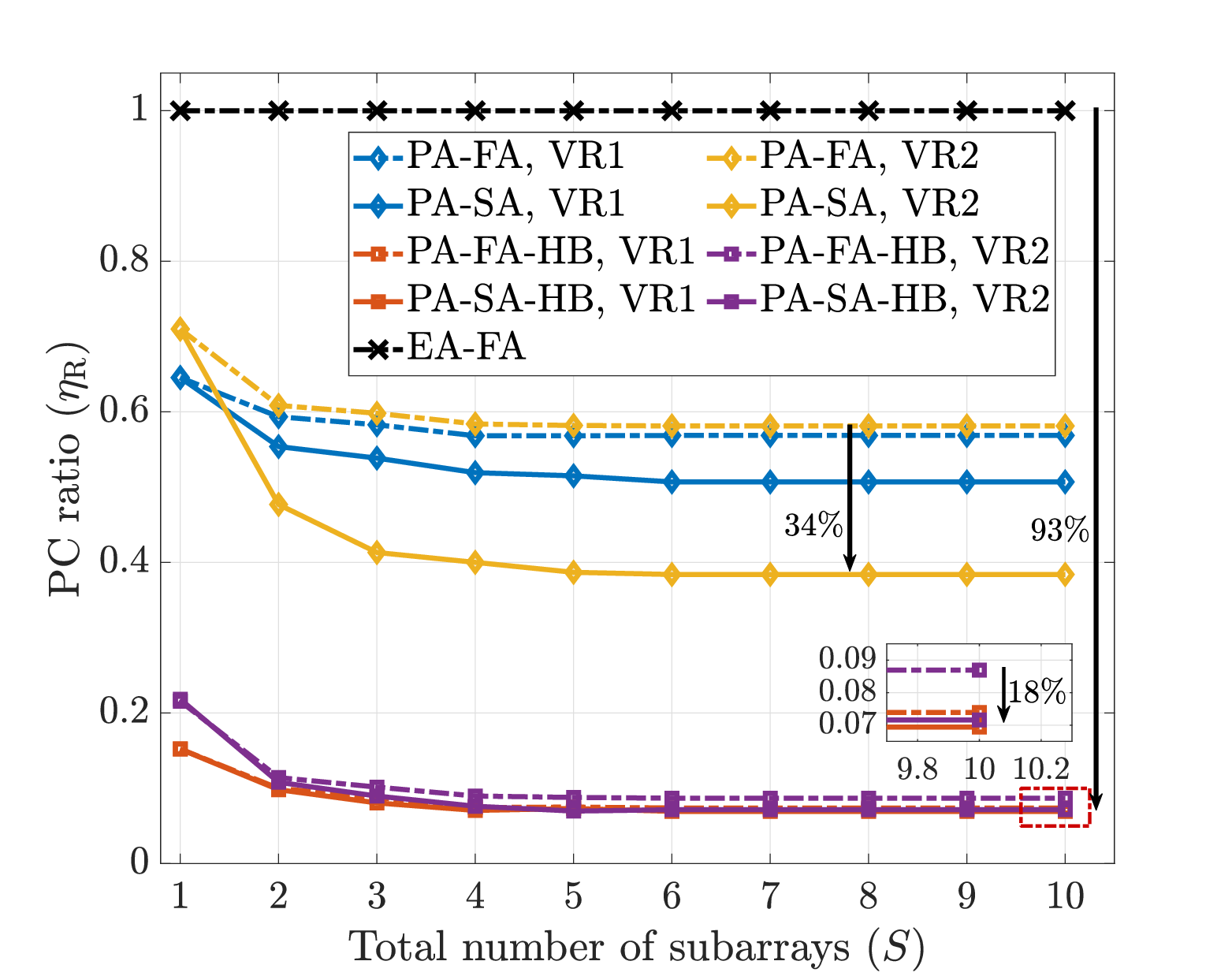}
    \vspace{-1.2em}
    \caption{\small PC ratio versus the number of subarrays for $\mathrm{MF}$-cases (Fig.~\ref{fig:Scenario_XL_MIMO}(\subref{fig:MF_VR1}) \& Fig.~\ref{fig:Scenario_XL_MIMO}(\subref{fig:MF_VR2})).\normalsize}
    \label{fig:MF_power_consumption_vs_subarrays}
    \end{minipage}
    \hfill
    \begin{minipage}[t]{0.32\textwidth}  
        \centering
        \includegraphics[trim=10 0cm 2cm 0cm,clip,width=\textwidth]{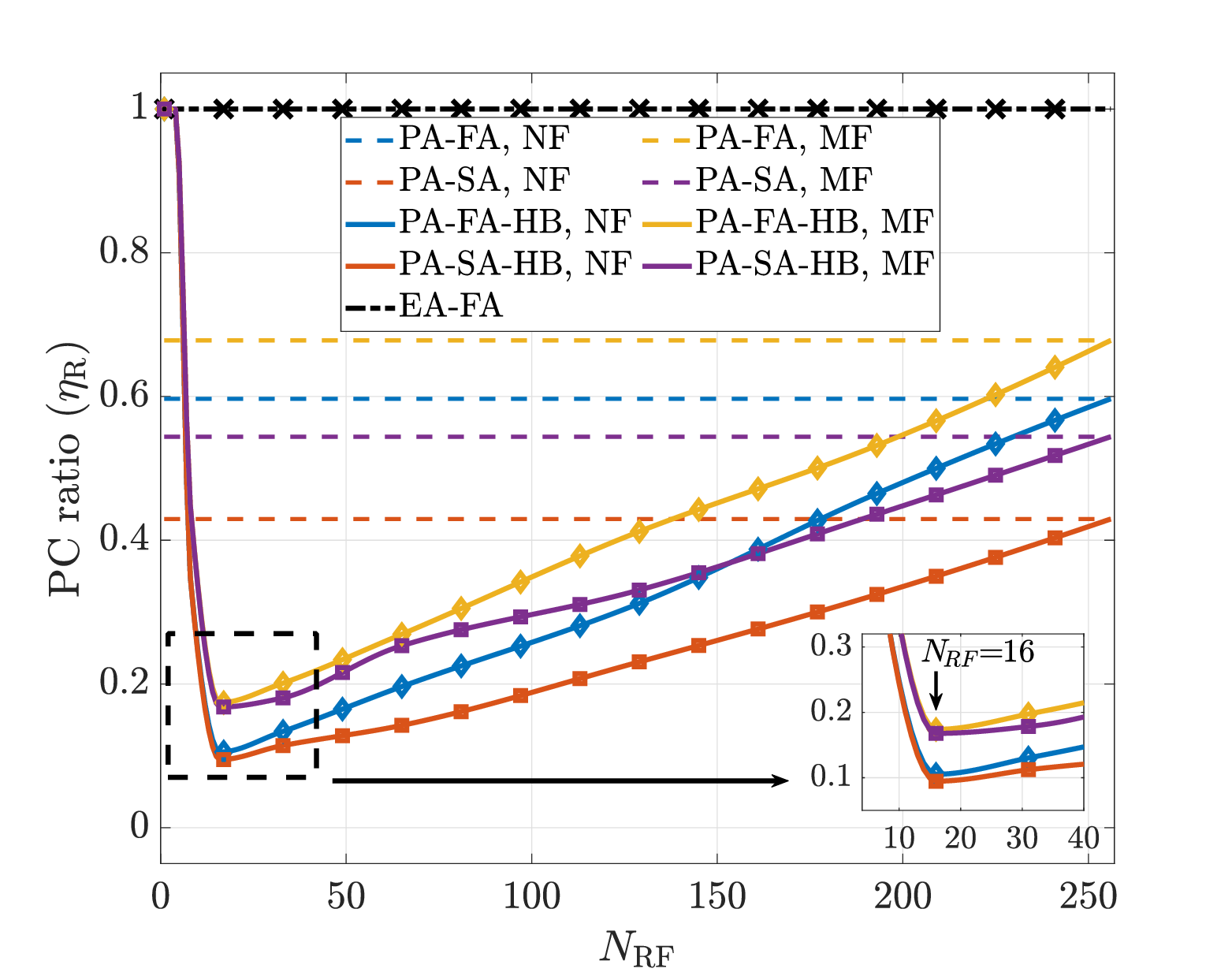}
     \vspace{-1.2em}
    \caption{\small  PC ratio versus the number of RF chains with $S=4$ subarrays for $\mathrm{MF}$-cases (Fig.~\ref{fig:Scenario_XL_MIMO}(\subref{fig:NF_VR2}) \& Fig.~\ref{fig:Scenario_XL_MIMO}(\subref{fig:MF_VR2})).\normalsize}
    \vspace{0cm}
    \label{fig:RF_chains_MF}
    \end{minipage}
\vspace{-1.5em}
\end{figure*}

    
    
\subsubsection{Asymptotic SE and HE Results}
We now focus on the convergence behavior of the mean DL SE $\ME\{\mathrm{SE}^{\mathrm{EA}}\}$ and DL HE $\ME\{\Xi^{\mathrm{EA}}\}$ with equal PA and FA activation. Before the introduction of optimized PA and SA variables, we can decipher  the system performance under the influence of user topological configurations for different VR scenarios and NF/FF regions. In this regard, we generate the user locations using the $d_{sk}$ limits for $d_{\rm{f}}$ of the benchmark case $S=4$. In Fig. \ref{fig:Asymptotic_DL}, we can observe that the DL SE improves monotonically for an increasing number of subarrays for both the NF scenarios (NF-VR1, NF-VR2), and saturates to a certain level as the relative signal contribution from the peripheral subarrays diminishes. For the MF cases, we notice a degradation in the mean SE performance with more subarrays before approaching convergence. This observation can be attributed to the increase in interference signals intended for the FF ID users and the continuous increase in EH power with gradual convergence trends as depicted in Fig. \ref{fig:Asymptotic_EH}. We experience even more degraded SE performance for the MF-VR1 case because all ID users are in the FF region. The HE performance improves for all topological scenarios with an increasing number of subarrays as the inter-user-interference is considered to be a positive contribution to HE at the user receiver. The important point to observe is that the harvested power is more for the scenarios when the ID users are either concentrated within a single NF VR or they are located within the FF region, in contrast to the scenarios of spatially distributed ID groups. The DL SE and HE values at $S=15$ are already within an average of $1.1\%$ and $4.2\%$ of their respective asymptotic limits, conclusively validating that most useful SE and HE contributions arise from a limited VR-dominant subset of subarrays.

\subsubsection{PA-SA Optimized PC Reduction}
From this point forward, we present the performance evaluation of the proposed joint PA-SA optimization in \textbf{Algorithm \ref{alg:opt_process}}, in terms of the overall XL-MIMO power consumption $P_C$. For a judicious assessment of the proposed optimization scheme, we consider three operational methods: (i) \textit{Equal power allocation \& full array activation} \textbf{(EA-FA)}: In this scheme, all subarrays are activated ($\qa=\textbf{1}_s$) with equal PA coefficients for all the users ($\Omega^{\mathsf{ID-EA}}_{sl}\!,\Omega^{\mathsf{EH-EA}}_{sm}$). (ii) \textit{Optimized power allocation \& full array activation }\textbf{(PA-FA)}: All subarrays are activated ($\qa=\textbf{1}_s$) along with optimized PA (${\boldsymbol{\Omega}}^{\mathsf{*ID}}$,${\boldsymbol{\Omega}}^{\mathsf{*EH}}$) by using only the PA routine in \textbf{Algorithm \ref{alg:opt_process}}. (iii) \textit{Joint optimized power allocation \& sub-array activation}\textbf{ (PA-SA)}: In this proposed method, both PA and SA coefficients (${\boldsymbol{\Omega}}^{\mathsf{*ID}}$,${\boldsymbol{\Omega}}^{\mathsf{*EH}}$,$\qa^{*}$) are optimized by utilizing the developed  \textbf{Algorithm \ref{alg:opt_process}}.
Furthermore, the HB variants of \textbf{PA-FA} and \textbf{PA-SA} methods are also considered to analyze the performance improvement due to the introduction of the HB technique. 

Based on the aforementioned discussion, we first analyze the mean ratio of active subarrays ($S_a/S$) in Fig. \ref{fig:subarray_active} to evaluate the performance of PA-SA optimization within the scope of spatial SnS effects. As the number of subarrays increases from $S=1$ to $S=10$, much less active subarrays are required for the NF scenarios (solid lines) to satisfy the SE and HE QoS thresholds, in comparison to the MF scenarios (dotted lines). However, this ratio is smaller for the scenarios with more spatial clustered VRs; NF-VR2 $0.51$, MF-VR2 $0.59$, against the dispersed user location scenarios;  NF-VR1 $0.68$, MF-VR1 $0.89$. It can also be noticed that the activation pattern for HB variants for all the spatial scenarios closely matches that of the fully digital architecture. 

Similar to the approach used during the optimization problem formulation in \eqref{eq:optimization_1}, and for better understanding of comparative performance, the PC ratio  is defined against the benchmark \textbf{EA-FA} method as $\eta_\mathrm{R} \triangleq P^{\mathrm{R}}_C / P^{\text{EA-FA}}_C$, where $\mathrm{R} \in \{\text{\textbf{PA-SA, PA-FA-HB, PA-FA, PA-FA-HB, EA-FA}}\}$. In this regard, the PC ratio is depicted in Fig. \ref{fig:Thresholds} over the mean DL HE $\ME\{\Xi^{\mathsf{EA}}_{th}\}$ and DL SE $\ME\{\mathrm{SE}^{\mathsf{EA}}_{th}\}$ with respect to the four network configurations in Fig. \ref{fig:Scenario_XL_MIMO}. For brevity, we consider only the lower \textbf{PA-FA} and upper \textbf{PA-SA-HB} bound optimization methods. It can be noticed that the PC increases as the thresholds $\ME\{\Xi^{\mathsf{EA}}_{th}\}$ and $\ME\{\mathrm{SE}^{\mathsf{EA}}_{th}\}$ increase. If more ID users are located in the NF region, we observe stronger PC variation from the lower threshold values to the higher one, such as $0.30-0.72$, $0.02-0.22$ for NF-VR1, $0.44-0.68$, $0.02-0.18$ for NF-VR2, $0.42-0.64$, $0.01-0.12$ for MF-VR2, with respect to \textbf{PA-FA} and \textbf{PA-SA-HB} respectively. 
On the other hand, we experience higher PC requirement when all the ID users are in the far-field MF-VR1, even at lower thresholds levels from $0.51$ to $0.62$, $0.05$ to $0.12$. Additionally, a careful assessment of the maximum thresholds ($\ME\{\Xi^{\mathsf{EA}}_{th}\}$, $\ME\{\mathrm{SE}^{\mathsf{EA}}_{th}\}$) also illuminates the peak performance capacities of the different network topologies in the context of the benchmark \textbf{EA-FA}. In the absence of both PA and SA optimization, these maximum threshold values remain constant at ($0.0229$ mW, $5.14$bps/Hz) when the ID users are situated in the NF, irrespective of the VR configurations (NF-VR1, NF-VR2). 
Although the HE level remains almost similar for MF scenarios since all the EH users remain within a single NF VR, the maximum SE decreases to $3.53$ bps/Hz for MF-VR1 and $4.15$ bps/Hz for MF-VR2 cases.

\begin{figure*}[t]
    \centering
    \begin{minipage}[t]{0.32\textwidth}
        \centering
        \includegraphics[trim=19 0cm 2cm 0cm,clip,width=\textwidth]{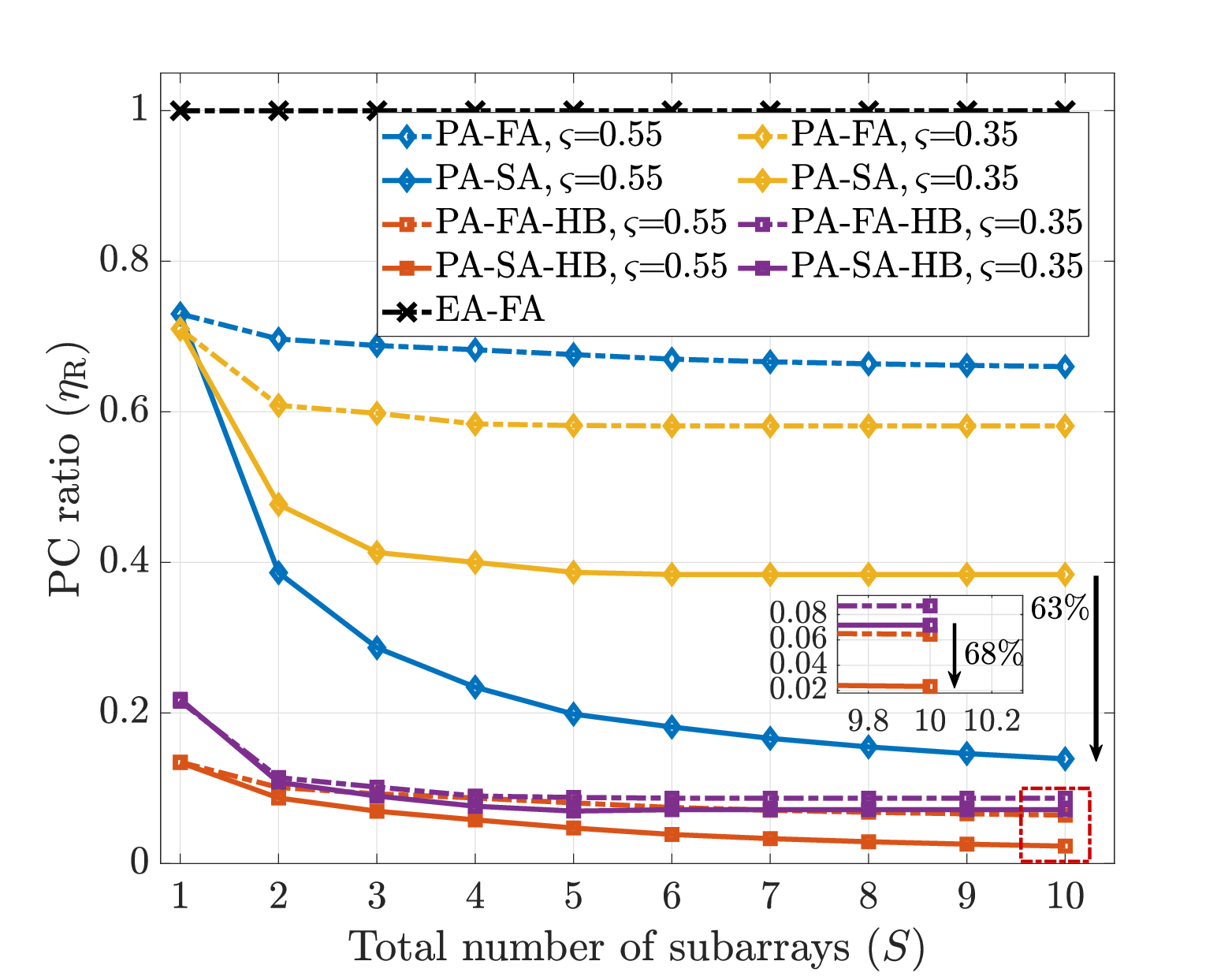}
     \vspace{-1.2em}
    \caption{\small  PA efficiency ($\varsigma$) effects on PC ratio versus the number of subarrays for $\mathrm{MF}$-case (Fig.~\ref{fig:Scenario_XL_MIMO}(\subref{fig:MF_VR2})).\normalsize}
    \vspace{0cm}
    \label{fig:PA_eff}
    \end{minipage}
    \hfill
    \begin{minipage}[t]{0.32\textwidth}  
        \centering
        \includegraphics[trim=10 0cm 2cm 0cm,clip,width=\textwidth]{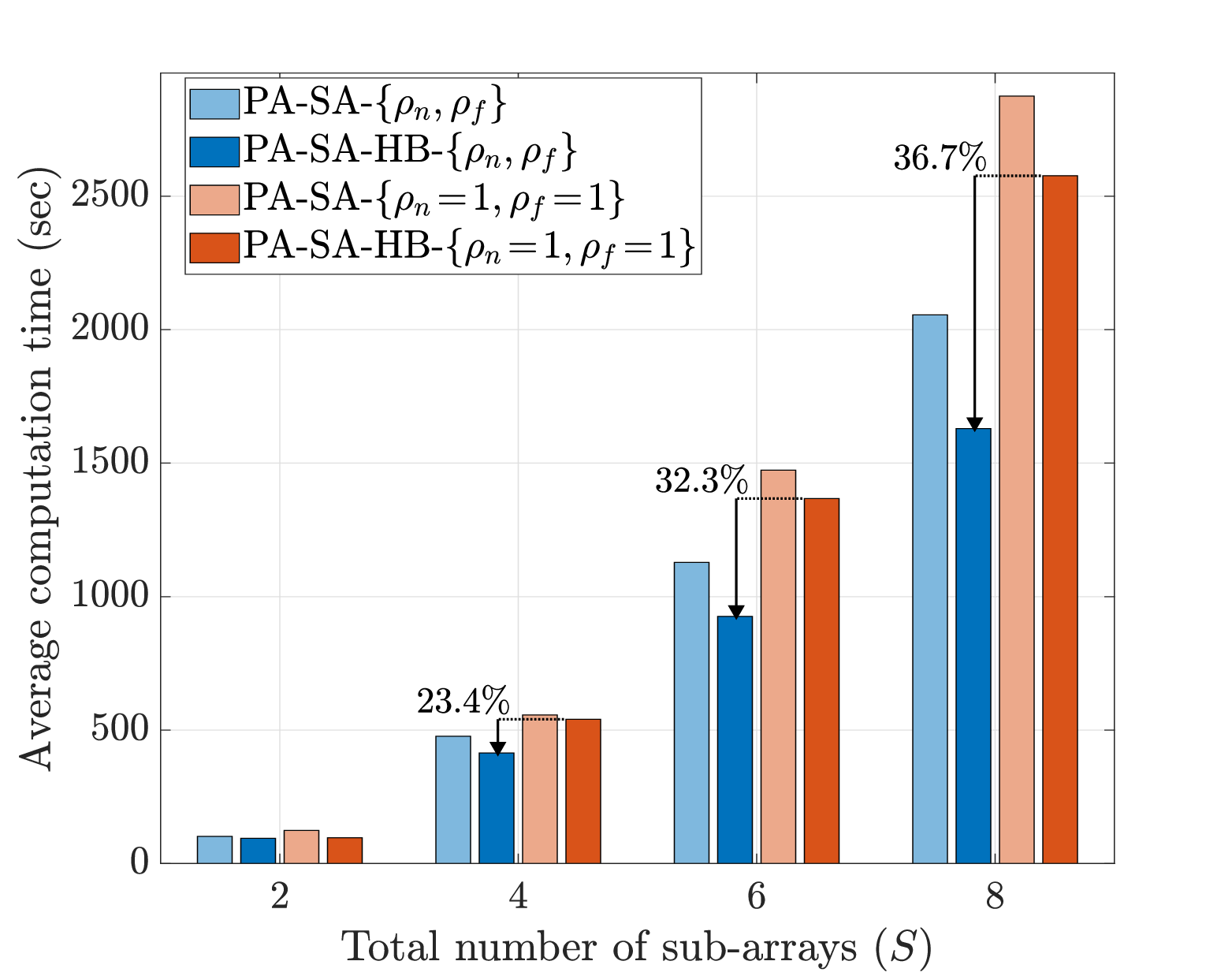}
    \vspace{-1.2em}
    \caption{\small Average computational time versus the number of subarrays.\normalsize}
    \label{fig:Computation}
    \end{minipage}
    \hfill
    \begin{minipage}[t]{0.32\textwidth}  
        \centering
        \includegraphics[trim=10 0cm 2cm 0cm,clip,width=\textwidth]{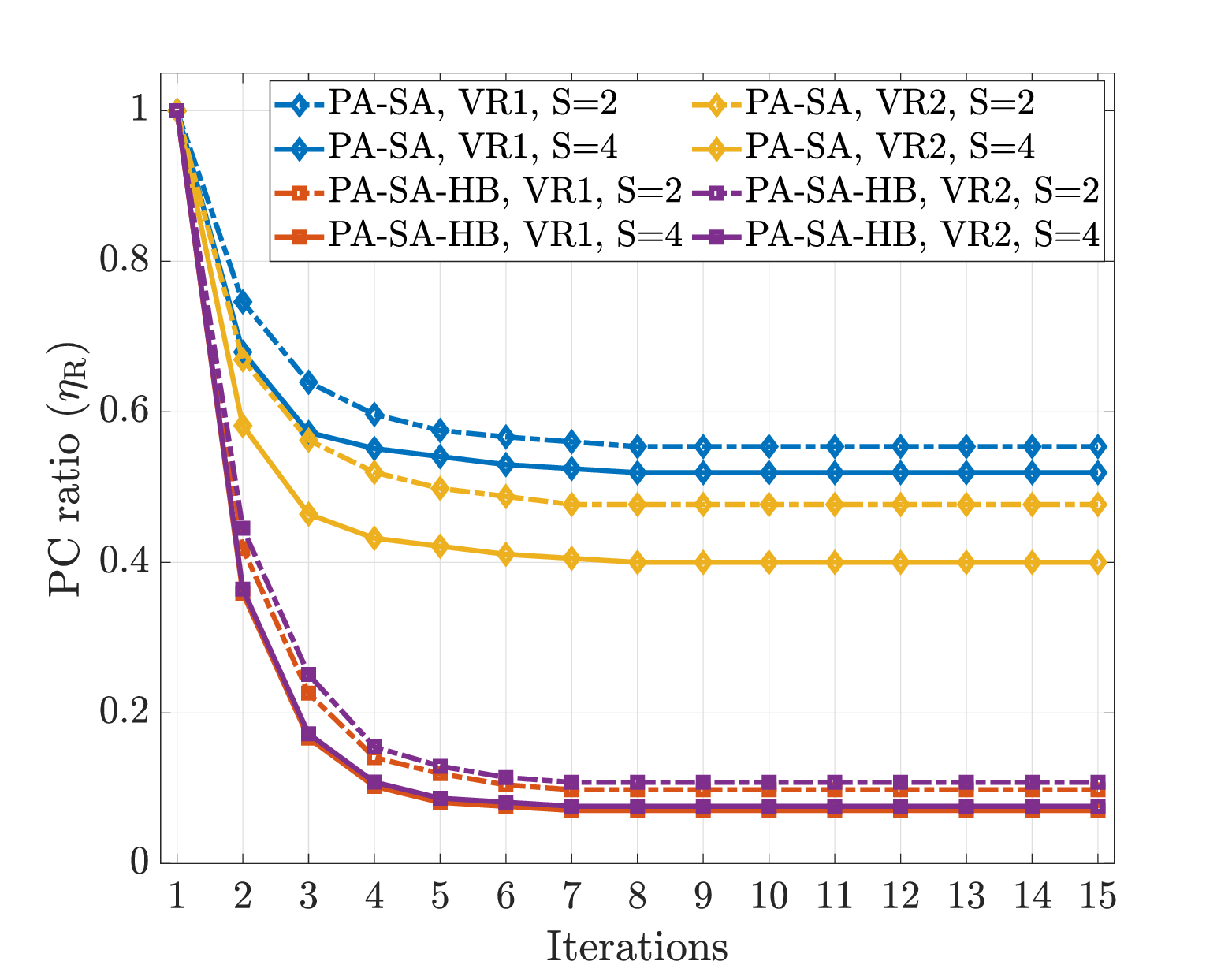}
     \vspace{-1.2em}
    \caption{\small  Algorithm convergence versus the number of iterations with $S=4$ subarrays for $\mathrm{MF}$-cases.\normalsize}
    \vspace{0cm}
    \label{fig:Convergence_MF_PA_SA}
    \end{minipage}
\vspace{-1.5em}
\end{figure*}

Furthermore, the decreasing trends of the PC ratio $\eta_{\mathrm{R}}$ are assessed with the increasing number of subarrays $S$ in Fig. \ref{fig:NF_power_consumption_vs_subarrays}  for NF and Fig. \ref{fig:MF_power_consumption_vs_subarrays} for MF scenarios with all the $\mathrm{R}$ methods with respect to the benchmark \textbf{EA-FA} case. It can be clearly seen that all the \textbf{SA} variants (solid lines) have smaller $\eta_{\mathrm{R}}$ values than the corresponding \textbf{FA} variants (dotted lines). This difference is even more pronounced for the fully digital cases with more NF spatial distributed VRs; \textbf{PA-SA, NF-VR2} and \textbf{PA-SA, MF-VR2} with $45\%$ and $34\%$ reduction against the cases \textbf{PA-FA, NF-VR2} and \textbf{PA-FA, MF-VR2}, respectively. The PC reductions of $15\%$ and $18\%$ are still significant for the corresponding HB variants of distributed VR cases, which reinforces the idea of concentrated energy focusing into the distinct NF VRs. On the other hand, all the optimized HB cases exhibit average $93\%$ PC reduction in comparison to the EA-FA case. Keeping in view these results, it is very important to understand our proposed PA optimization routine assigns the PA coefficients in a manner that the low contributing subarrays are already allocated small PA values. Hence, the proposed SA routine assists the PA routine by binary switching on/off the subarray RF chains to further improve the PC performance by re-adjusting the PA variables.

\subsubsection{Impact of HB Design on PC Reduction}
For the benchmark $S=4$ configuration, Fig.~\ref{fig:RF_chains_MF} illustrates the impact of the number of RF chains on the normalized PC for \textbf{PA-FA-HB} and \textbf{PA-SA-HB} under the \textbf{NF-VR2} and \textbf{MF-VR2} scenarios. As the number of RF chains increases from $\nrf=1$ to the fully digital case $\nrf=256$, the PC ratio decreases from its maximum value ($\eta_\mathrm{R}=1$), reaches a minimum, and then gradually increases until it converges to the PC levels of the corresponding fully digital counterparts. Starting from $\nrf=1$, increasing $\nrf$ provides the HB architecture with greater spatial flexibility, leading to a sharp reduction in $\eta_\mathrm{R}$ and a minimum around $\nrf=16$, which confirms the effectiveness of HB for PC reduction. Beyond this point, however, the incremental HB gains become limited, while the RF-chain circuit power continues to increase linearly, causing the overall PC to rise. As $\nrf$ approaches $N=256$, the HB curves converge to their corresponding fully digital \textbf{PA-FA} and \textbf{PA-SA} benchmarks, as expected. Moreover, the NF curves consistently outperform their MF counterparts, since NF beamfocusing enables more precise energy concentration, leading to greater PC reduction. Overall, the results indicate that an intermediate RF-chain regime is the most power-efficient operating region, while the proposed \textbf{PA-SA-HB} scheme achieves the lowest PC among all the  methods.

\subsubsection{Impact of PA Efficiency on PC Reduction} Consistent with Fig.~\ref{fig:MF_power_consumption_vs_subarrays}, Fig.~\ref{fig:PA_eff} shows that the PC ratio decreases with $S$ for all schemes, with more pronounced reductions as the PA efficiency improves from $\varsigma=0.35$ to $\varsigma=0.55$. In particular, for the fully digital \textbf{PA-SA} scheme, the PC performance improves by $63\%$, while for the HB-based \textbf{PA-SA-HB} scheme, the improvement reaches $68\%$. These results confirm that a higher PA efficiency directly enhances the overall PC reduction by lowering the transmission power dependent term in the power model developed in \eqref{eq:power_consumed}. Although most related works adopt the practical value $\varsigma=0.35$ \cite{Asaad,Jun_Zhang}, we also consider the more optimistic case of $\varsigma=0.55$ \cite{victor}, as an extreme setting to assess the best-case PC reduction capability of the proposed \textbf{PA-SA} algorithm.


 \subsubsection{Algorithm Computational Requirements} In Fig.~\ref{fig:Computation}, we present the average computational time of the proposed joint \textbf{PA–SA} optimization with HB deployments. As expected, the computation time increases with the XL-array size, while the HB-based variants consistently require less time compared to the fully digital counterparts. We also examine the influence of the NF and FF balancing factors $\{\rho_{\rm{n}}$, $\rho_{\rm{f}}\}$ defined in \eqref{eq:subarray_relative_contribution} on the SA routine. For the cases 
$S=4,6$ and $8$, incorporating the balancing factors reduces the SA computational time by $23.4\%$, $32.3\%$ and $36.7\%$, respectively, compared to the baseline case without these factors ( i.e, setting $\rho_{\rm{n}}=1, \rho_{\rm{f}}=1$). This notable reduction highlights the effectiveness of the proposed NF/FF user classification in improving computational efficiency. For large-scale implementations, subarray-level parallelization and offline optimization of recurrent activation patterns are promising directions to further reduce the operational runtime.

\subsubsection{Algorithm Convergence Analysis} 
Figure~\ref{fig:Convergence_MF_PA_SA} shows the convergence behavior of the proposed joint \textbf{PA-SA} optimization for the \textbf{MF-VR1} and \textbf{MF-VR2} scenarios, considering both \textbf{PA-SA} and \textbf{PA-SA-HB} architectures with $S=2$ and $S=4$. The normalized PC ratio $\eta_{\mathrm{R}}$ decreases monotonically with the iteration index for all cases, confirming the stable convergence of the proposed iterative framework. Most of the PC reduction is achieved within the first few iterations, after which the curves gradually approach their steady-state values and become nearly unchanged after around $8$--$11$ iterations. For the representative $S=4$ cases, the fully digital and HB variants attain approximately $89.3\%$ and $89.6\%$ of their total reduction toward the final converged value within the first three iterations, respectively, highlighting the fast convergence rate of the proposed algorithm. Moreover, the HB curves converge to substantially lower PC ratios than their fully digital counterparts, while preserving a similarly rapid convergence profile across both VR configurations.

\vspace{-0.5em}
\section{Conclusion}
We investigated the PC reduction for a modular XL-MIMO system providing SWIPT services to  users situated in MF EM regions, while considering the NF SnS. For this purpose, we first devised a novel classification method, based on the  Frobenius norm of the channel estimate correlation matrix, to categorize the users in distinct NF and FF groups. In this context, we minimized the overall system PC by joint PA and SA optimization, while satisfying the SE and HE QoS thresholds. The proposed two-tier optimization framework addresses the original mixed-integer problem by decoupling the PA and SA sub-problems, and solving them by employing pivotal optimization approaches, like DR splitting based ADMM and surrogate auxiliary function, by utilizing the NF and FF balancing. Using this simple NF/FF classification, we could reduce the computation time of the proposed algorithm up to $32.3 \%$, in comparison to the baseline case. We also presented a holistic asymptotic study which substantiated the assumption that peripheral subarrays have progressively diminishing contribution to  both the SE and HE performance, thereby validating the concept of VR-associated SnS effects. Future work may incorporate the proposed frequency-correlation features into lightweight learning-based NF/FF and VR classifiers and extend the \textbf{PA-SA-HB} design to account for phase-shifter impairments.


\appendices

\vspace{-0.5em}
\section{Proof of Proposition \ref{prop:CE_stat}}\label{app:CE_stat}
To evaluate the statistics of the channel estimate $\hgsk$, we segregate the analysis into NF and FF regions.
\subsubsection{NF channel estimate} The mean of the NF channel estimate $\hgskn$ is $\ME\{\hgskn\} \!=\! \bgskn \!=\! \gamsk\bhskn$, with zero-mean estimation error $\tekl$. Moreover, the variance $\MV\{[\hgskn]_{t}\}$ only depends on $\tekl$ in \eqref{eq:CE}, which is further dependent on the covariance of the function $\qf=\hat{\qW}^{\dagger}_{A,s} \tilde{\qn}_s$ given as
\begin{align}\label{eq:cov_f}
    \operatorname{Cov}\{\qf\}\!&=\!\ME\{\hat{\qW}^{\dagger}_{\!\!A,s} \tilde{\qn}_s \tilde{\qn}^{\mathrm{H}}_{s} (\hat{\qW}^{\dagger}_{\!\!A,s})^{\mathrm{H}}\}\nonumber\\
    & \stackrel{(a)}{=}\!\sigma^2 \big(\ME\{ (\hat{\qW}^{\rm H}_{\!\!A,s} \!\hat{\qW}_{\!\!A,s})^{\!-\!1} \}\big)\nonumber\\
    & \stackrel{(b)}{\approx} \sigma^2  \dfrac{N}{Q \nrf}\Big(1+ \dfrac{N -1}{Q \nrf} \Big),
\end{align}
where step (a) follows from the fact that $\tilde{\qn}_s$ is zero-mean, uncorrelated across its entries, and independent of $\hat{\qW}^{\dagger}_{\!\!A,s}$, with
    \begin{align}
    \ME\{\tilde{\qn}_s \tilde{\qn}^{\rm H}_s\}&= \sigma^2 \qI_{Q N_{RF}}.
\end{align}
 
To derive step (b), we first use the definition of $\qW_{A,s}$ in \eqref{eq:rx_pilot_sig} to establish that each zero-mean entry $\hat{\qW}_{A,s}(i,j)$ has variance $\ME\{\lvert\hat{\qW}_{A,s}(i,j)\rvert^2\}\!=\!1/N$, while the entries of the matrix product $$[\hat{\qW}^{\rm{H}}_{A,s}\hat{\qW}_{A,s}]_{p,q}= \sum\nolimits^{Q \nrf}_{n'=1} \hat{\qW}^{\rm H}_{A,s}(p,n') \hat{\qW}_{A,s}(q,n'), $$ are sums of $Q \nrf$ independent terms. The diagonal terms $[\hat{\qW}^{\rm{H}}_{A,s}\hat{\qW}_{A,s}]_{p,q}$ have the deterministic value of $\sum^{Q \nrf}_{x,y \in \mathcal{S}_{(s)}}\!\! 1/N\!=\!Q \nrf/N$, while the non-diagonal terms have zero mean and variance of $\MV\{[\hat{\qW}^{\rm{H}}_{A,s}\hat{\qW}_{A,s}]_{p,q}\}=Q \nrf/N^2$. Now, we define a matrix $\boldsymbol{\Delta}_{\qW}$ which represents the difference from the identity matrix, as $\hat{\qW}^{\rm{H}}_{A,s}\!\hat{\qW}_{A,s}\!=\! (Q \nrf/N)(\qI_{N}\!+\!\boldsymbol{\Delta}_{\qW})$. This matrix $\boldsymbol{\Delta}_{\qW}$ has the scaled non-diagonal entries of $\hat{\qW}^{\rm{H}}_{A,s}\hat{\qW}_{A,s}$, along with zero diagonal entries. The expectation in step (a) of  \eqref{eq:cov_f} can be derived by using second order Taylor series as
\begin{align}
    \!\ME\{[\hat{\qW}^{\rm{H}}_{\!\!A,s}\!\hat{\qW}_{\!\!A,s}]^{-1}\}&=(N/Q \nrf) \ME\big\{\big(\qI_{N}+\boldsymbol{\Delta}_{\qW}\big)^{-1}\big\}\nonumber\\
    &\approx(N/Q \nrf)\big(\qI_{N}+\ME\{\boldsymbol{\Delta}^2_{\qW}\}\big)\nonumber\\
    &\! \approx \!\dfrac{N}{Q \nrf}\Big(1+ \dfrac{N -1}{Q \nrf} \Big)\qI_{N},
\end{align}
where the diagonal entries are $[\boldsymbol{\Delta}^2_{\qW}]_{pp}= \sum\nolimits^{N}_{n''=1}\Delta^2_{pn''}= \sum\nolimits_{n''\neq p} \Delta^2_{pn''}$ with mean $\ME\{[\boldsymbol{\Delta}^2_{\qW}]_{pp}\}= (N -1)\big((Q \nrf/N)^2(Q \nrf/N^2))=(N-1)/(Q \nrf)$. On the other hand, the off-diagonal entries of $\boldsymbol{\Delta}^2_{\qW}$ are independent and zero mean. 

Using the results in \eqref{eq:cov_f}, we define the variance of $\MV\{[\tekl]_{t}\}$ as $\nue$ and present it in \eqref{eq:var_error} .

\subsubsection{FF channel estimate} The mean of the FF channel estimate $\hgskf$ is $\ME\{\hgskf\}\!=\!\bgskf=\!\sqrt{\bklq} \bhskf$. In addition to $\MV\{[\tekl]_{t}\}$, the channel has a NLoS factor $\tgkl$ which also contributes to the variation in the channel estimates, given as $\MV\{[\hgskf]_{t}\} = \bsk \!+\!\nue$.
\vspace{-0.5em}
\section{Proof of Proposition \ref{prop:asymptotic_DL_SE}}\label{app:asymptotic_DL_SE}
The expectation of each received signal component in calculation of asymptotic limit on the DL SE involves both NF and FF user cases, which have been derived in this appendix. 
\subsubsection{Coherent ID signal}
The expectation term in \eqref{eq:asymptotic_coherent_DL_ID} for a NF user can be calculated as
\vspace{-0.1em}
\begin{align}
   &\ME\{\lvert\qg^{\rm{T}}_{sl\!,\rm{n}} {\qv}^*_{sl\!,\rm{n}} \rvert^2\}\!\!=\!\!\kappa^{2}_{sl} \gamma^{2}_{sl}{\sum\limits^K_{k=1}\! \sum\limits^K_{k'=1}\! \vartheta_{sl\!,kk'} \ME\{\hat{\qg}^{\rm{T}}_{sk\!,\rm{n}}\bar{\qh}^*_{sl\!,\rm{n}} \bar{\qh}^{\rm{T}}_{sl\!,\rm{n}}\hat{\qg}^*_{sk'\!,\rm{n}}\}\!}\nonumber\\
   &\hspace{0cm}=\!\kappa^{2}_{sl} \gamma^{2}_{sl}\sum\nolimits^K_{k=1}\! \sum\nolimits^K_{k'=1} \! \vartheta_{sl\!,kk'}\big(\! \gamma_{sk}\gamma_{sk'}{\qh}^{\rm{T}}_{sk\!,\rm{n}}\bar{\qh}^*_{sl\!,\rm{n}} \bar{\qh}^{\rm{T}}_{sl\!,\rm{n}}{\qh}^*_{sk'\!,\rm{n}}\!\nonumber\\
   &\hspace{0cm}+\ME\{\tilde{\hat{\pmb{\varepsilon}}}^{\rm{T}}_{sk}\bar{\qh}^*_{sl,\rm{n}} \bar{\qh}^{\rm{T}}_{sl,\rm{n}}\tilde{\hat{\pmb{\varepsilon}}}^*_{sk'}\}\big) \nonumber\\
   &\hspace{0cm}=\kappa^{2}_{sl} \gamma^{2}_{sl}\sum\nolimits^K_{k=1} \sum\nolimits^K_{k'=1}  \!\vartheta_{sl,kk'}\Theta_{n}(s,l,k,k').
\end{align}





The expectation for the case of a FF user follows as
\begin{align}
   \ME\{\lvert\qg^{\rm{T}}_{sl,\rm{f}} {\qv}^*_{sl,\rm{f}} \rvert^2\}&\!=\kappa^{2}_{sl}\!\sum^K_{k=1}\! \sum^K_{k'=1}\! \vartheta_{sl\!,kk'} \!\big(\varsigma_{sl} \ME\{\hat{\qg}^{\rm{T}}_{sk\!,\rm{f}}\bar{\qh}^*_{sl\!,\rm{f}} \bar{\qh}^{\rm{T}}_{sl\!,\rm{f}} \hat{\qg}^*_{sk'\!,\rm{f}}\!\}\nonumber\\
   &+\beta_{sl} \ME\{\hat{\qg}^{\rm{T}}_{sk,\rm{f}}\tilde{\qg}^*_{sl,\rm{f}} \tilde{\qg}^{\rm{T}}_{sl,\rm{f}} \hat{\qg}^*_{sk',\rm{f}}\}\big)\nonumber\\
   &\hspace{-1em}=\!\kappa^{2}_{sl}\sum\nolimits^K_{k=1} \!\sum\nolimits^K_{k'=1} \!\!\vartheta_{sl,kk'} \!\Theta_{\rm{f}}(s,l,k,k').
\end{align}
Here, we compute the term $\ME\{\tilde{\qg}^{\rm{T}}_{sk,\rm{f}}\bar{\qh}^*_{sl,\rm{f}} \bar{\qh}^{\rm{T}}_{sl,\rm{f}} \tilde{\qg}^*_{sk',\rm{f}}\}=\delta_{kk'}\trace(\qB \bm{\Sigma})=\delta_{kk'} N$, by using the relationship $\Ex\{{\qu}^H \qB  {\qu} \} = \boldsymbol{\mu}^H \qB \boldsymbol{\mu} + \trace (\qB \boldsymbol{\Sigma})$ \cite[Eq. (15.14)]{Kay} with  $\qu=\tilde{\qg}^*_{sk',\rm{f}}$, $\boldsymbol{\Sigma}=\In$ and $\qB=\bar{\qh}^*_{sl,\rm{f}} \bar{\qh}^{\rm{T}}_{sl,\rm{f}}$. Moreover, the term $\ME\{\tilde{\qg}^{\rm{T}}_{sk,\rm{f}}\tilde{\qg}^*_{sl,\rm{f}} \tilde{\qg}^{\rm{T}}_{sl,\rm{f}} \tilde{\qg}^*_{sk',\rm{f}}\}= \delta_{kl}\delta_{k'l} N(N+1)$ is determined by using the Wishart matrix transformation \(\qW= \qA \qA^{\mathrm{H}} \sim \mathcal{W}_m(n,\qI)\) with $\qA=\tilde{\qg}^{\rm{T}}_{sl,\rm{f}}$, $m=1$ and $n=N$ to obtain $\ME\{\lVert\tgkl\rVert^{2v}\}={\Gamma(n+v)}/{\Gamma(n)}$ for the special case $v=2$ ~\cite[Eq. (2.10)]{verdu}.

\subsubsection{Interfering ID signal}
The expectation expression for the non-coherent signal in \eqref{eq:asymptotic_noncoherent_DL_ID} for NF and FF channel-precoding can be obtained by using a similar approach as applied for the coherent case. Specifically, we have
\vspace{-0.1em}
\begin{align}
    &\!\!\!\ME\{\lvert\qg^{\rm{T}}_{sl,\rm{n}} {\qv}^*_{sl',\rm{n}} \rvert^2\}\!= \!\kappa^2_{sl'} \gamma^2_{sl} \sum^K_{k=1} \sum^K_{k'=1}\!\vartheta_{sl'\!,kk'}\Theta_{n}(s,l,k,k'),\!\\
    &\!\ME\{\lvert\qg^{\rm{T}}_{sl,\rm{f}} {\qv}^*_{sl',\rm{f}} \rvert^2\}\!=\!\kappa^{2}_{sl'}\sum^K_{k=1} \sum^K_{k'=1} \vartheta_{sl',kk'} \Theta_{\rm{f}}(s,l,k,k').
\end{align}

For the case of NF channel and FF precoding vector, we can derive 
\vspace{-0.1em}
\begin{align}
    &\!\ME\{\lvert\qg^{\rm{T}}_{sl\!,\rm{n}} {\qv}^*_{sl',\rm{f}} \rvert^2\}\!\!=\!\kappa^2_{sl'}\gamma^2_{sl}\!\sum^K_{k=1}\! \sum^K_{k'=1} \! \vartheta_{sl'\!,kk'}\!\big(\ME\{{\qg}^{\rm{T}}_{sk\!,\rm{f}}\bar{\qh}^*_{sl\!,\rm{n}} \bar{\qh}^{\rm{T}}_{sl\!,\rm{n}} {\qg}^*_{sk'\!,\rm{f}}\!\}\nonumber\\
    &\hspace{0.5em}+\ME\{\tilde{\hat{\pmb{\varepsilon}}}^{\rm{T}}_{sk,\rm{f}}\bar{\qh}^*_{sl,\rm{n}} \bar{\qh}^{\rm{T}}_{sl,\rm{n}} \tilde{\hat{\pmb{\varepsilon}}}^*_{sk',\rm{f}}\}\big)\nonumber\\
    &\hspace{0.5em}=\kappa^2_{sl'}\gamma^2_{sl}\sum\nolimits^K_{k=1} \sum\nolimits^K_{k'=1} \vartheta_{sl',kk'}\big(\delta_{kk'}\varsigma_{sk}\doubleacute{\varrho}^2_{s,kl}\!+\!(1-\delta_{kk'})\nonumber\\
    &\hspace{0.5em}\times\!\sqrt{\varsigma_{sk}\varsigma_{sk'}}\doubleacute{\varrho}_{s,kl}\doubleacute{\varrho}_{s,k'l}\!+\!\sqrt{\beta_{sk}\beta_{sk'}} \delta_{kk'}N\!+\!\delta_{kk'}\nue\big).\!
\end{align}

On the other hand, the expectation term for the case of FF channel and NF precoding case can be evaluated as
\begin{align}
    &\!\!\ME\{\lvert\qg^{\rm{T}}_{sl\!,\rm{f}} {\qv}^*_{sl',\rm{n}} \rvert^2\}\! = \!\kappa^2_{sl'}\! \sum\nolimits^K_{k=1}\! \sum\nolimits^K_{k'=1} \vartheta_{sl'\!,kk'}\big(\gamma_{sk}\gamma_{sk'}
    \nonumber\\
    &\hspace{0em}\times \ME\{\bar{\qh}^{\rm{T}}_{sk\!,\rm{n}}\qg^*_{sl\!,\rm{f}} \qg^{\rm{T}}_{sl\!,\rm{f}} \bar{\qh}^*_{sk'\!,\rm{n}}\}+\ME\{\tilde{\hat{\pmb{\varepsilon}}}^{\rm{T}}_{sk,\rm{n}}\qg^*_{sl,\rm{f}} \qg^{\rm{T}}_{sl,\rm{f}} \tilde{\hat{\pmb{\varepsilon}}}^*_{sk',\rm{n}}\}\big)
    \nonumber\\
    &\hspace{0cm} = \kappa^2_{sl'} \sum\nolimits^K_{k=1} \sum\nolimits^K_{k'=1}  \vartheta_{sl',kk'}\big(\gamma_{sk}\gamma_{sk'}\big( \varsigma_{sl}\doubleacute{\varrho}_{s,kl}(\delta_{kk'}\doubleacute{\varrho}_{s,kl}\nonumber\\
    &\hspace{1cm}+\doubleacute{\varrho}_{s,k'l})+\beta_{sl}\delta_{kk'}N\big)+(\varsigma_{sl}+\beta_{sl})\delta_{kk'}\nue \big).
\end{align}

\subsubsection{Interfering EH signal}
The expectation of the EH signal received \eqref{eq:asymptotic_ID_EH} at an NF ID user can be given as
\begin{align}
   \ME\{\lvert\qg^{\rm{T}}_{sl,\rm{n}} {\qw}^*_{sm} \rvert^2\}=& \kappa^2_{sm}\big(\ME\{ {\qg}^{\rm{T}}_{sm} \qg^*_{sl,\rm{n}}  \qg^{\rm{T}}_{sl,\rm{n}} {\qg}^*_{sm} \}\nonumber\\
   & + \ME\{ \tilde{\hat{\pmb{\varepsilon}}}^{\rm{T}}_{sm} \qg^*_{sl,\rm{n}}  \qg^{\rm{T}}_{sl,\rm{n}} \tilde{\hat{\pmb{\varepsilon}}}^*_{sm} \}\big)\nonumber\\
   = & \kappa^2_{sm} \gamma^2_{sl}\big( \gamma^2_{sm} \rsmlnsq + \nue\big).
\end{align}

Finally, the expectation of the EH interference signal experienced by a FF ID user can be derived as
\begin{align}
    \!\ME\{ \lvert\qg^{\rm{T}}_{sl\!,\rm{f}} {\qw}^*_{sm} \rvert^2\}\!=& \,\kappa^2_{sm} \big(\ME\{ {\qg}^{\rm{T}}_{sm} \qg^*_{sl,\rm{f}}  \qg^{\rm{T}}_{sl,\rm{f}} {\qg}^*_{sm} \}\nonumber \\
   &\,+ \ME\{ \tilde{\hat{\pmb{\varepsilon}}}^{\rm{T}}_{sm} \qg^*_{sl,\rm{f}}  \qg^{\rm{T}}_{sl,\rm{f}} \tilde{\hat{\pmb{\varepsilon}}}^*_{sm} \}\big)\nonumber\\
   &\hspace{-4.6em} = \kappa^2_{sm} \Big( \gamma^2_{sm} \varsigma_{sl} \ME\{\bar{\qh}^{\rm{T}}_{sm} \bar{\qh}^*_{sl,\rm{f}}  \bar{\qh}^{\rm{T}}_{sl,\rm{f}} \bar{\qh}^*_{sm}\} \nonumber\\
   &\hspace{-4.6em}+ \gamma^2_{sm} \beta_{sl}  \ME\{\bar{\qh}^{\rm{T}}_{sm}\tilde{\qg}^*_{sl,\rm{f}}  \tilde{\qg}^{\rm{T}}_{sl,\rm{f}} \bar{\qh}^*_{sm}\} \nonumber\\
   &\hspace{-4.6em}+ \varsigma_{sl}\ME\{ \tilde{\hat{\pmb{\varepsilon}}}^{\rm{T}}_{sm}\bar{\qh}^*_{sl,\rm{f}}  \bar{\qh}^{\rm{T}}_{sl,\rm{f}} \tilde{\hat{\pmb{\varepsilon}}}^*_{sm} \}+ \beta_{sl}\ME\{ \tilde{\hat{\pmb{\varepsilon}}}^{\rm{T}}_{sm} \tilde{\qg}^*_{sl,\rm{f}}  \tilde{\qg}^{\rm{T}}_{sl,\rm{f}} \tilde{\hat{\pmb{\varepsilon}}}^*_{sm} \} \Big)\nonumber\\
   &\hspace{-4.6em}=  \kappa^2_{sm} \big(\gamma^2_{sm} (\varsigma_{sl}\doubleacute{\varrho}^2_{s\!,ml} \!+\!  \beta_{sl} N)\!+\!(\varsigma_{sl}\!+\!\beta_{sl})\nue \big). 
\end{align}
\vspace{-2em}
\section{Proof of Proposition \ref{prop:asymptotic_DL_EH}}\label{app:asymptotic_DL_EH}
The asymptotic analysis for DL HE includes the expectation terms of the received signals from both EH and ID  signals. 

\subsubsection{DL HE contribution by EH signals}
The expectation of the term $\ME\{\Upsilon_{s,s'\!,m, m'}\}\!=\! \kmrtsmp \kmrtmp \ME\{\hat{\qg}^{\rm{T}}_{s'm'} \qg^*_{s'm}  \qg^{\rm{T}}_{sm} \hat{\qg}^*_{sm'}\}$ in \eqref{eq:asymptotic_intended_EH} is de-synthesized into coherent and non-coherent subarray cases. First, we compute the expectation term for the coherent case $(s=s')$ as 
\begin{align}
    &\!\ME\{\hat{\qg}^{\rm{T}}_{sm'} \qg^*_{sm}  \qg^{\rm{T}}_{sm} \hat{\qg}^*_{sm'}\}\!= \gamma^2_{sm}\gamma^2_{sm'}\bar{\qh}^{\rm{T}}_{sm'} \bar{\qh}^*_{sm}  \bar{\qh}^{\rm{T}}_{sm} \bar{\qh}^*_{sm'}\nonumber\\
    &+\gamma^2_{sm}\ME\{\tilde{\hat{\pmb{\varepsilon}}}^{\rm{T}}_{sm'} \bar{\qh}^*_{sm}  \bar{\qh}^{\rm{T}}_{sm}\tilde{\hat{\pmb{\varepsilon}}}^*_{sm'}\}\nonumber \\
   &= \gamma^2_{sm}\big(\gamma^2_{sm}\delta^{EH}_{mm'} N^2 
   + \gamma^2_{sm'}(1-\delta^{EH}_{mm'})\rsmmpsq +\nue \big).
\end{align}

Now, we focus on the non-coherent case $(s' \neq s)$, which we derive as 
\begin{align}
    \!\ME\{\hat{\qg}^{\rm{T}}_{s'm'} \qg^*_{s'm}  \qg^{\rm{T}}_{sm} \hat{\qg}^*_{sm'}\!\}&\!=  \gamma_{s'm'} \gamma_{s'm} \gamma_{sm} \gamma_{sm'}\bar{\qh}^{\rm{T}}_{s'm'}\bar{\qh}^*_{s'm}  \bar{\qh}^{\rm{T}}_{sm} \bar{\qh}^*_{sm'}\nonumber\\
    &+ \gamma_{s'm}\gamma_{sm}\ME\{\tilde{\hat{\pmb{\varepsilon}}}^{\rm{T}}_{s'm'} \bar{\qh}^*_{s'm}  \bar{\qh}^{\rm{T}}_{sm} \tilde{\hat{\pmb{\varepsilon}}}^*_{sm'}\}\nonumber\\
    &\hspace{-2em}= \gamma_{s'm} \gamma_{sm}\gamma_{s'm'} \gamma_{sm'}\rspmmp \rsmmp. 
\end{align}
\subsubsection{DL HE contribution by ID signals}
The expectation term $\ME\{\Upsilon_{s,s'\!,m, l}\}=\ME\{{\qw}^{\rm{T}}_{s'l} \qg^*_{s'm} \qg^{\rm{T}}_{sm} {\qv}^*_{sl}\}$ in \eqref{eq:asymptotic_ID_EH} covers the received EH power by the ID signals. Here, we derive the coherent case $(s'=s)$ for the NF ID users as
\begin{align}
    &\ME\{{\qv}^{\rm{T}}_{sl,\rm{n}} \qg^*_{sm} \qg^{\rm{T}}_{sm} {\qv}^*_{sl,\rm{n}}\}=  \kappa^2_{sl} \gamma^2_{sm} \sum^K_{k=1} \sum^K_{k'=1} \vartheta_{sl,kk'}\big(\gamma_{sk}\gamma_{sk'}\nonumber\\
    &\hspace{0cm}\times {\bar{\qh}}^{\rm{T}}_{sk,\rm{n}} \bar{\qh}^*_{sm} \bar{\qh}^{\rm{T}}_{sm}{\bar{\qh}}^*_{sk',\rm{n}} +\ME\{\tilde{\hat{\pmb{\varepsilon}}}^{\rm{T}}_{sk,\rm{n}} \bar{\qh}^*_{sm} \bar{\qh}^{\rm{T}}_{sm} \tilde{\hat{\pmb{\varepsilon}}}^*_{sk',\rm{n}}\}\big)\nonumber\\ 
    &\hspace{0cm} = \kappa^2_{sl} \gamma^2_{sm} \sum\nolimits^K_{k=1} \sum\nolimits^K_{k'=1} \vartheta_{sl,kk'}\big(\gamma_{sk}\gamma_{sk'}\varrho_{s,kl}(\delta_{kk'}\varrho_{s,kl}\nonumber\\
    &\hspace{1cm}\,+\varrho_{s,k'l})+\delta_{kk'}\nue\big).
\end{align}

The expectation of the HE contribution by the FF ID users can be obtained as
\begin{align}
    &\ME\{{\qv}^{\rm{T}}_{sl,\rm{f}} \qg^*_{sm} \qg^{\rm{T}}_{sm} {\qv}^*_{sl,\rm{f}}\}\!=\!\kappa^2_{sl} \gamma^2_{sm} \sum\nolimits^K_{k=1} \sum\nolimits^K_{k'=1}\!\!\! \vartheta_{sl,kk'}\big(\sqrt{\varsigma_{sk}\varsigma_{sk'}}\nonumber\\
    &\hspace{0.8cm}\times \bar{\qh}^{\rm{T}}_{sk,\rm{f}} \bar{\qh}^*_{sm} \bar{\qh}^{\rm{T}}_{sm} \bar{\qh}^*_{sk'\!,\rm{f}}\!+\!\sqrt{\beta_{sk}\beta_{sk'}}\ME\{\tilde{\qg}^{\rm{T}}_{sk,\rm{f}} \bar{\qh}^*_{sm} \bar{\qh}^{\rm{T}}_{sm} \tilde{\qg}^*_{sk'\!,\rm{f}}\!\}\nonumber\\
    &\hspace{0.8cm}+\ME\{\tilde{\hat{\pmb{\varepsilon}}}^{\rm{T}}_{sk,\rm{f}} \bar{\qh}^*_{sm} \bar{\qh}^{\rm{T}}_{sm} \tilde{\hat{\pmb{\varepsilon}}}^*_{sk'\!,\rm{f}}\}\big)\nonumber\\
    &\hspace{0.8cm}=\kappa^2_{sl} \gamma^2_{sm} \sum\nolimits^K_{k=1} \sum\nolimits^K_{k'=1} \vartheta_{sl,kk'}\big(\sqrt{\varsigma_{sk}\varsigma_{sk'}}\varrho_{s,kl}(\delta_{kk'}\varrho_{s,kl}\nonumber\\
    &\hspace{1cm}\,+\varrho_{s,k'l})+\beta_{sk}\delta_{kk'}N+\delta_{kk'}\nue\big).\nonumber
\end{align}

Considering the non-coherent case $(s' \neq s)$, the expectation for interference from the NF ID users can be computed as
\vspace{-0.3em}
\begin{align}
    &\ME\{{\qv}^{\rm{T}}_{s'l,\rm{n}} \qg^*_{s'm} \qg^{\rm{T}}_{sm} {\qv}^*_{sl,\rm{n}}\!\}\!\!=\!\kappa_{s'l}\kappa_{sl} \gamma_{s'm}\gamma_{sm}\! \sum\nolimits^K_{k=1}\! \sum\nolimits^K_{k'=1}\!\!  \acute{\vartheta}_{ss',l}^{kk'}
    \nonumber\\
    &\hspace{0cm}\times\!\! \big(\gamma_{s'k}\gamma_{sk'} \ME\{{\qh}^{\rm{T}}_{s'k,\rm{n}} \qh^*_{s'm} \qh^{\rm{T}}_{sm} {\qh}^*_{sk'\!,\rm{n}}\!\}\!+\!\!\ME\{\tilde{\hat{\pmb{\varepsilon}}}^{\rm{T}}_{s'k,\rm{n}} \qh^*_{s'm} \qh^{\rm{T}}_{sm} \tilde{\hat{\pmb{\varepsilon}}}^*_{sk'\!,\rm{n}}\!\}\big)
    \nonumber\\
    &\hspace{0cm}\!=\!\kappa_{s'l}\kappa_{sl} \gamma_{s'm}\gamma_{sm} \! \sum\nolimits^K_{k=1}
    \!\sum\nolimits^K_{k'=1} \!\!\!\acute{\vartheta}_{ss',l}^{kk'}\gamma_{s'k}\gamma_{sk'}\varrho_{s',km} \varrho_{s,k'm}.
\end{align}

Finally, this expectation for the case due to the FF ID users can be derived as
\begin{align}
    &\ME\{{\qv}^{\rm{T}}_{s'l,\rm{f}} \qg^*_{s'm} \qg^{\rm{T}}_{sm} {\qv}^*_{sl,\rm{f}}\}\!=\!\kappa_{s'l}\kappa_{sl}\gamma_{s'm}\gamma_{sm}\!\!\sum\nolimits^K_{k=1}\! \sum\nolimits^K_{k'=1} \!\!\!\acute{\vartheta}_{ss',l}^{kk'} \nonumber\\
    &\hspace{0cm}\times\big(\ME\{{\qg}^{\rm{T}}_{s'k,\rm{f}} \qh^*_{s'm} \qh^{\rm{T}}_{sm} {\qg}^*_{sk'\!,\rm{f}}\!\}\!+\!\ME\{\tilde{\hat{\pmb{\varepsilon}}}^{\rm{T}}_{s'k,\rm{f}} \qh^*_{s'm} \qh^{\rm{T}}_{sm} \tilde{\hat{\pmb{\varepsilon}}}^*_{sk'\!,\rm{f}}\!\}\big)\nonumber\\
    &\hspace{0cm}\!=\!\kappa_{s'l}\kappa_{sl}\gamma_{s'm}\gamma_{sm}\!\!\sum^K_{k=1} \!\sum^K_{k'=1} \!\acute{\vartheta}_{ss',l}^{kk'}\sqrt{\varsigma_{s'k,\rm{f}} \varsigma_{sk,\rm{f}}}\varrho_{s'\!,km} \varrho_{s,k'm}\!. 
\end{align}

\bibliographystyle{IEEEtran}
\bibliography{biblo}

\begin{IEEEbiography}[{\includegraphics[width=1in,height=1.25in,clip,keepaspectratio]
{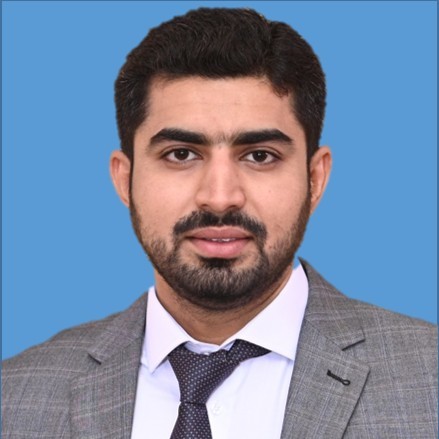}}]
{Muhammad Zeeshan Mumtaz} (Graduate Student Member, IEEE) received the B.E. degree in Avionics engineering from National University of Sciences \& Technology (NUST), Pakistan, in 2014, and the M.S. degree in Avionics Engineering from Air University, Pakistan, in 2021. He is currently pursuing the Ph.D. degree in Electrical \& Electronics Engineering at Queen's University Belfast, U.K. From 2021 to 2023, he was a Lecturer of Data Communications \& Networking at College of Aeronautical Engineering, NUST Pakistan. His research interests include cell-free massive MIMO systems, NOMA communication systems, MIMO radars, autonomous modulation classification and the application of deep learning techniques for contemporary communication challenges.
\end{IEEEbiography}

\begin{IEEEbiography}[{\includegraphics[width=1in,height=1.25in,clip,keepaspectratio]
{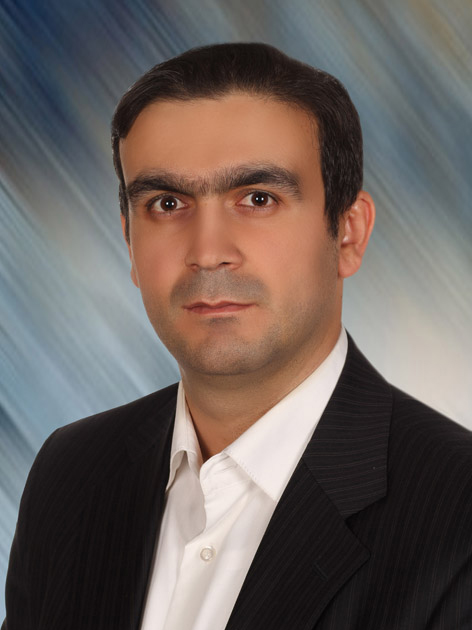}}]
{Mohammadali Mohammadi} (Senior Member, IEEE) is currently a Lecturer at the Centre for Wireless Innovation (CWI), Queen’s University Belfast, U.K. He previously held the position of Research Fellow at CWI from 2021 to 2024. His research interests include signal processing for wireless communications, cell-free massive MIMO, integrated sensing and communications, and reconfigurable intelligent surfaces. He has published more than 80 research papers in accredited international peer reviewed journals and conferences in the area of wireless communication and has co-authored two invited book chapters. He serves as an Associate Editor for IEEE Communications Letters and IEEE Open Journal of the Communications Society. He was a recipient of the Exemplary Reviewer Award for IEEE Transactions on Communications in 2020 and 2022, and IEEE Communications Letters in 2023. He has been a member of Technical Program Committees for many IEEE conferences, such as ICC, GLOBECOM, and VTC.
\end{IEEEbiography}

\begin{IEEEbiography}[{\includegraphics[width=1in,height=1.25in,clip,keepaspectratio]{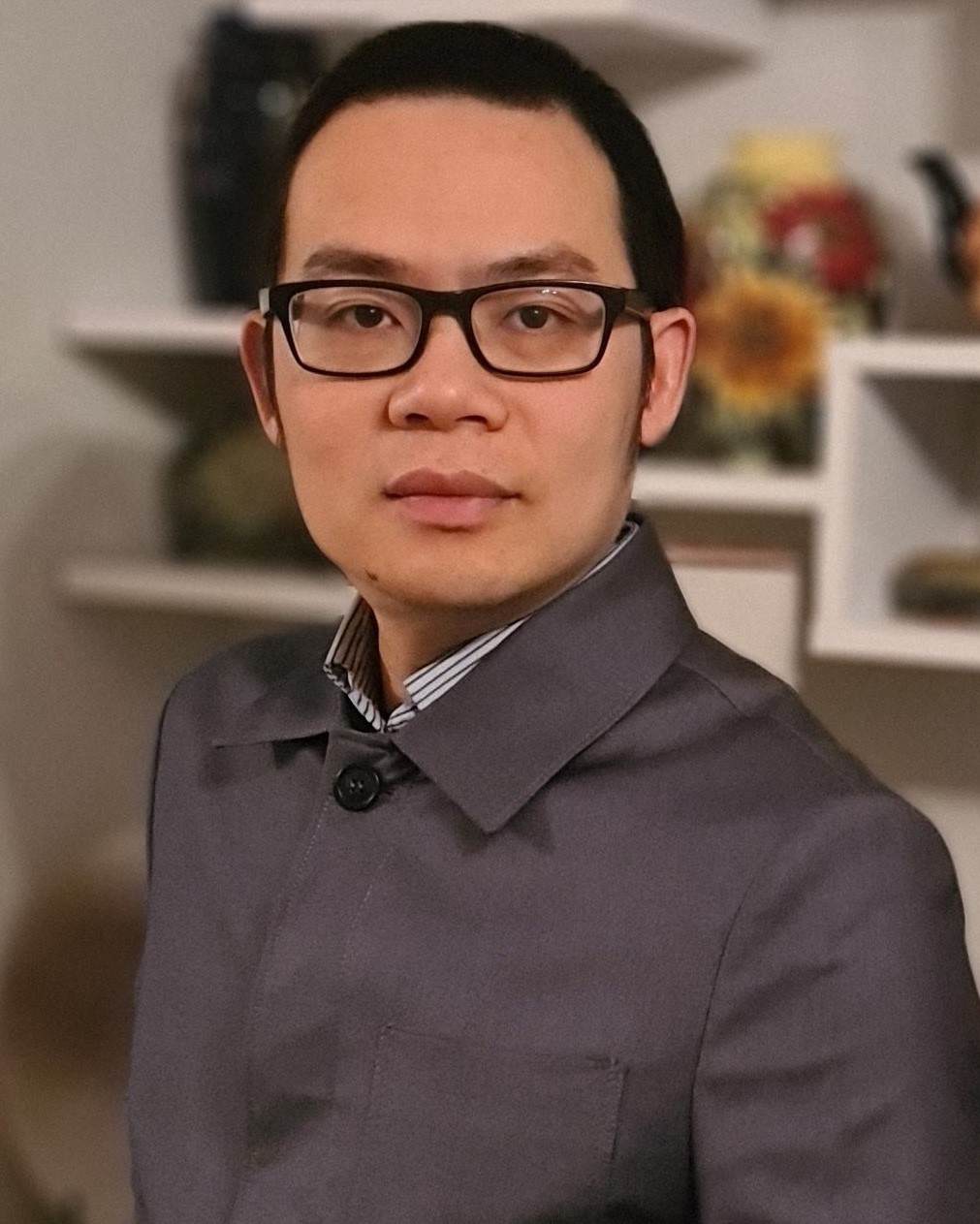}}]
{Hien Quoc Ngo} (Fellow, IEEE)  is currently a Reader with Queen's University Belfast, U.K. His main research interests include massive MIMO systems, cell-free massive MIMO, reconfigurable intelligent surfaces, physical layer security, and cooperative communications. He has co-authored many research papers in wireless communications and co-authored the Cambridge University Press textbook \emph{Fundamentals of Massive MIMO} (2016).

He received  the IEEE ComSoc Test of Time Paper Award for Advances in Communications in 2026,  the IEEE ComSoc Stephen O. Rice Prize in 2015, the IEEE ComSoc Leonard G. Abraham Prize in 2017, the Best Ph.D. Award from EURASIP in 2018, and the IEEE CTTC Early Achievement Award in 2023. He also received the IEEE Sweden VT-COM-IT Joint Chapter Best Student Journal Paper Award in 2015. He was awarded the UKRI Future Leaders Fellowship in 2019. He serves as the Editor for the IEEE Transactions on Wireless Communications, IEEE Transactions on Communications, the Digital Signal Processing, and the Physical Communication (Elsevier). He was an Editor of the IEEE Wireless Communications Letters, a Guest Editor of IET Communications, and a Guest Editor of IEEE ACCESS in 2017.
\end{IEEEbiography}

\begin{IEEEbiography}[
{\includegraphics[width=1in,height=1.25in,clip,keepaspectratio]{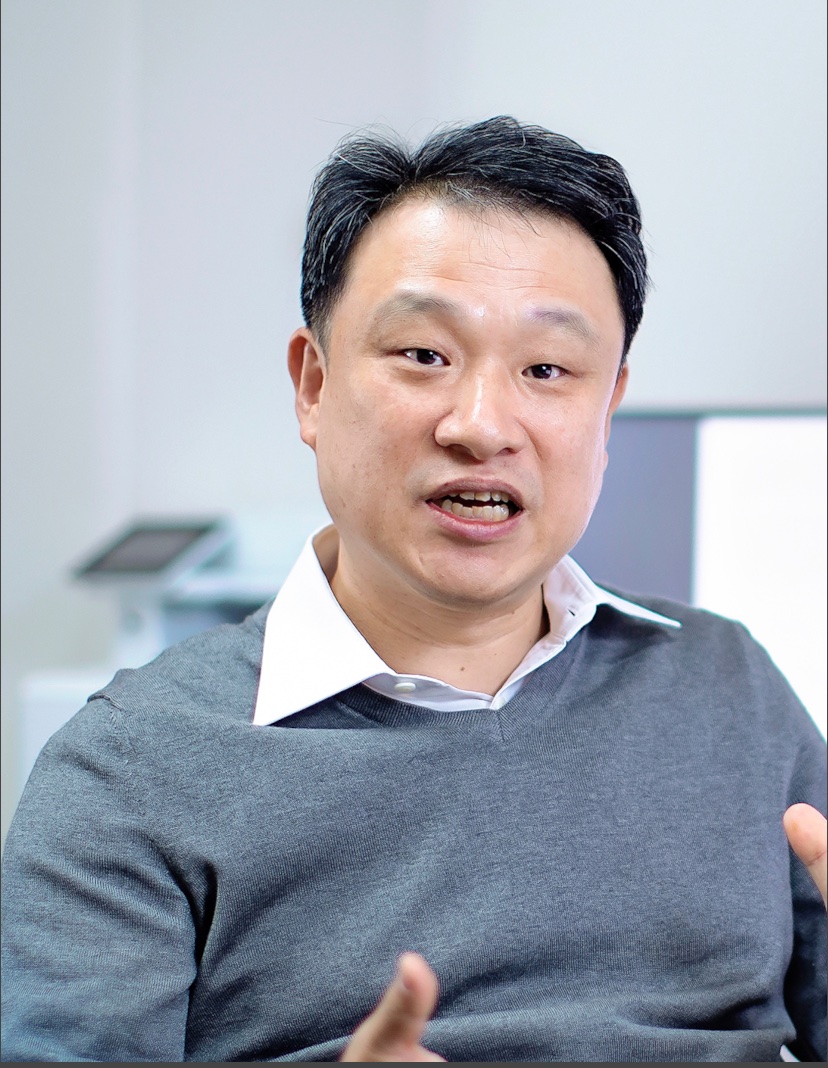}}
]{Hyundong Shin} (Fellow, IEEE)  
received the B.S. degree in Electronics Engineering from Kyung Hee University (KHU), Yongin-si, Korea, in 1999, and the M.S. and Ph.D. degrees in Electrical Engineering from Seoul National University, Seoul, Korea, in 2001 and 2004, respectively.
During his postdoctoral research at the Massachusetts Institute of Technology (MIT) from 2004 to 2006, he was with the Laboratory for Information Decision Systems (LIDS). In 2006, he joined the KHU, where he is currently a Professor in the Department of Electronic Engineering. His research interests include quantum information science, wireless communication, and machine intelligence.
Dr. Shin received the IEEE Communications Society’s Guglielmo Marconi Prize Paper Award and William R. Bennett Prize Paper Award. He served as the Publicity Co-Chair for the IEEE PIMRC and the Technical Program Co-Chair for the IEEE WCNC and the IEEE GLOBECOM. He was an Editor of \textsc{IEEE Transactions on Wireless Communications} and \textsc{IEEE Communications Letters}.

\end{IEEEbiography}


\vspace{-2em}
\begin{IEEEbiography}[{\includegraphics[width=1in,height=1.35in,clip,keepaspectratio]{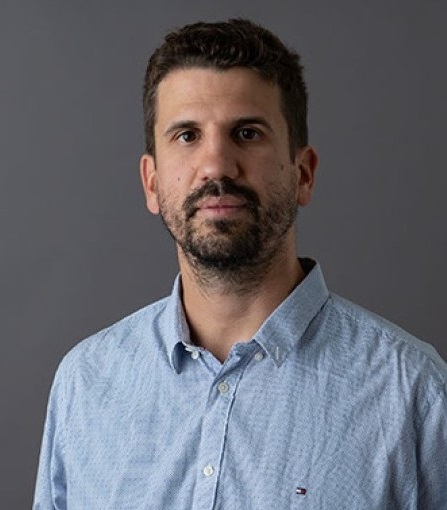}}]
{Michail Matthaiou}(Fellow, IEEE) obtained his Ph.D. degree from the University of Edinburgh, U.K. in 2008. 
He is currently a Professor of Communications Engineering and Signal Processing and Deputy Director of the Centre for Wireless Innovation (CWI) at Queen’s University Belfast, U.K. He is also an Eminent Scholar at the Kyung Hee University, Republic of Korea. He has held research/faculty positions at Munich University of Technology (TUM), Germany and Chalmers University of Technology, Sweden. His research interests span signal processing for wireless communications, beyond massive MIMO, reflecting intelligent surfaces, mm-wave/THz systems and AI-empowered communications.

Dr. Matthaiou and his coauthors received the IEEE Communications Society (ComSoc) Leonard G. Abraham Prize in 2017. He currently holds the ERC Consolidator Grant BEATRICE (2021-2026) focused on the interface between information and electromagnetic theories. To date, he has received the prestigious 2023 Argo Network Innovation Award, the 2019 EURASIP Early Career Award and the 2018/2019 Royal Academy of Engineering/The Leverhulme Trust Senior Research Fellowship. His team was also the Grand Winner of the 2019 Mobile World Congress Challenge. He was the recipient of the 2011 IEEE ComSoc Best Young Researcher Award for the Europe, Middle East and Africa Region and a co-recipient of the 2006 IEEE Communications Chapter Project Prize for the best M.Sc. dissertation in the area of communications. He has co-authored papers that received best paper awards at the 2018 IEEE WCSP and 2014 IEEE ICC. In 2014, he received the Research Fund for International Young Scientists from the National Natural Science Foundation of China. He is currently the Editor-in-Chief of Elsevier Physical Communication, a Senior Editor for \textsc{IEEE Wireless Communications Letters} and \textsc{IEEE Signal Processing Magazine}, an Area Editor for \textsc{IEEE Transactions on Communications} and Editor-in-Large for \textsc{IEEE Open Journal of the Communications Society}. 
\end{IEEEbiography}

\end{document}